\documentclass[11pt,fleqn]{article}
\usepackage{amssymb,latexsym,amsmath,amsfonts,graphicx, ebezier}

\usepackage{pictex,color}

    \advance\textheight by \topskip

    \numberwithin{equation}{section}

\makeatletter
\def\eqalign#1{\null\vcenter{\def\\{\cr}\openup\jot\m@th
  \ialign{\strut$\displaystyle{##}$\hfil&$\displaystyle{{}##}$\hfil
      \crcr#1\crcr}}\,}
\makeatother

\newcommand{\pq}{\left[^{\delta}_{\epsilon}\right]}

\newcommand{\be}{\begin{equation}}
\newcommand{\ee}{\end{equation}}

\newcommand{\al}{\alpha}
\newcommand{\bt}{\beta}

\newcommand{\lb}{\lambda}

\def\pmtwo#1#2#3#4{\begin{pmatrix}#1&#2\\#3&#4\end{pmatrix}}

    \def\e{{\epsilon}}

    \def\Re{{\rm Re \,}}
    \def\Im{{\rm Im \,}}
    \def\Ai{{\rm Ai \,}}

    \def\bigO{{\cal O}}
    
\newcommand{\res[1]}{\operatornamewithlimits{Res}_{#1}}
    \def\Res{{\rm Res}}
    
    \def\P2n{{\rm P}_{{\rm II}}^{(n)}}

    \newtheorem{theorem}{Theorem}[section]
    \newtheorem{lemma}[theorem]{Lemma}
    
    \newtheorem{proposition}[theorem]{Proposition}
    
    \newtheorem{Definition}[theorem]{Definition}
    
    \newtheorem{Remark}[theorem]{Remark}
    \newenvironment{remark}{\begin{Remark}\rm}{\end{Remark}}
    \newtheorem{Example}[theorem]{Example}
    
    \newtheorem{Assumptions}[theorem]{Assumptions}

    \newenvironment{proof}%
    {\rm \trivlist \item[\hskip \labelsep{\bf Proof. }]}%
    {\hspace*{\fill}$\Box$\endtrivlist}
    \newenvironment{varproof}%
    {\rm \trivlist \item[\hskip \labelsep{\bf Proof}]}%
    {\hspace*{\fill}$\Box$\endtrivlist}

    \newcommand{\supp}{{\operatorname{supp}}}

    \DeclareMathOperator*{\Tr}{Tr}
    \newcommand\PVint{\mathop{\setbox0\hbox{$\displaystyle\intop$}%
        \hskip0.2\wd0%
        \vcenter{\hrule width0.6\wd0height0.5pt depth0.5pt}%
        \hskip-0.8\wd0%
        }\mskip-\thinmuskip\intop\nolimits}

\begin{document}
\title{Asymptotics for the partition function in two-cut random matrix models}
\author{T. Claeys\footnote{tom.claeys@uclouvain.be}, T. Grava\footnote{grava@sissa.it and tamara.grava@bristol.ac.uk},  and K. D. T-R McLaughlin\footnote{mcl@math.arizona.edu}}
\maketitle

\begin{abstract}
We obtain large $N$ asymptotics for the random matrix partition function
\[Z_N(V)=\int_{\mathbb R^N}\prod_{i<j}(x_i-x_j)^2
\prod_{j=1}^N e^{-N V(x_j)}dx_j,\]
in the case where $V$ is  a polynomial such that the random matrix eigenvalues accumulate on two disjoint intervals (the two-cut case). We  compute leading and sub-leading terms in the asymptotic expansion for $\log Z_N(V)$, up to terms that are small as $N\to\infty$. 
Our approach is based on the explicit computation of the first terms in the asymptotic expansion for a quartic symmetric potential $V$. 
Afterwards, we use deformation theory of the partition function and of the associated equilibrium measure to generalize our results to general two-cut potentials $V$.
The  asymptotic expansion of $\log Z_N(V)$ as $N\to\infty$   contains  terms  that  depend analytically on   the potential $V$  and that have already appeared in the literature.
 In addition our method allows to compute the $V$-independent terms  of the asymptotic expansion of  $\log Z_N(V)$ which, to the best of our knowledge, 
 had not appeared before in the literature. 
We use rigorous orthogonal polynomial and Riemann-Hilbert techniques which had so far been successful to compute asymptotics for the partition function only in the one-cut case.
\end{abstract}

\section{Introduction}
Consider the problem of finding asymptotics as $N\to\infty$ for
\begin{equation}\label{def Zn}
Z_{N}(V)=\int_{\mathbb R^N}\prod_{i<j}(x_i-x_j)^2
\prod_{j=1}^N e^{-N V(x_j)}dx_j,
\end{equation}
where $V$ is a polynomial  of even degree with positive leading coefficient (so the above integral converges).
$Z_N(V)$ is the partition function corresponding to the unitary invariant random matrix ensemble on the ensemble of Hermitian $N\times N$ matrices with the probability measure
\begin{equation}\label{ensemble}
\frac{1}{\tilde Z_N(V)}e^{-N\Tr V(M)}dM,\qquad dM=\prod_{i<j}d\Re M_{ij} d\Im M_{ij}\, \prod_{j=1}^NdM_{jj},
\end{equation}
where $\tilde Z_N(V)$ is a normalizing constant chosen such that we have a probability measure.
The induced measure on the (random) eigenvalues of this ensemble of random matrices is explicitly given by 
\begin{equation}\label{process}
\frac{1}{Z_N(V)}\prod_{i<j}(x_i-x_j)^2\prod_{j=1}^Ne^{-NV(x_j)}dx_j,
\end{equation}
where $Z_N(V)$ is given by (\ref{def Zn}).

\medskip

The study of the asymptotic behavior of $Z_{N}$ gained notoriety because of an ingenious connection to the combinatorics of graphs embedded in Riemann surfaces \cite{BIZ}.  It is now an established fact  \cite{BIPZ, BIZ, EM}  that there is a simple geometric description of a collection of potentials $V$ for which the large $N$ asymptotic expansion is 
\begin{equation}
\label{largeNF}
-\log Z_N(V)+\log Z_{N}(V_{0})=N^2F_0(V)+F_1(V)+\sum_{j=2}^J F_{j}(V) N^{-2j} +\bigO(N^{-2J-2}),
\end{equation}
in which each term $F_{j}$ depends analytically on the coefficients of the polynomial $V$, and is a generating function for graphs which can be embedded in a Riemann surface of genus $j$ \cite{BIZ,EM}.  

In the above, the ``base potential'' $V_{0}$ is quadratic, $V_{0}(x) = x^{2} / 2 $, and the interest originally was in the asymptotic behavior for potentials $V$ near this base potential.
The partition function can be evaluated exactly for the Gaussian external field $\frac{2x^2}{\sigma}$ in terms of a Selberg integral. Indeed in such a case 
\[
Z_{N,\sigma}^{GUE}:=Z_N\left(V(x)=2x^2/\sigma\right)
\]
 is the partition function of the (re-scaled) Gaussian Unitary Ensemble (GUE), and it is given by
\begin{equation}\label{ZGUE-intro}
Z_{N,\sigma}^{GUE}=(2\pi)^{N/2}\left(\frac{\sigma}{4N}\right)^{N^2/2}\prod_{n=1}^{N}(n!),
\end{equation}
see for example \cite{BI}. 
The large $N$ expansion of $\log Z_{N,\sigma}^{GUE}$ is given by
\begin{multline}\label{expansion GUE1}
\log Z_{N,\sigma}^{GUE}=-N^2 \left(\frac{3}{4}-\frac{1}{2}\log\frac{\sigma}{4}\right)
+N\log N +N(\log 2\pi -1) +\frac{5}{12}\log N \\+\frac{1}{2}\log 2\pi +\zeta'(-1)+\bigO(N^{-1}),
\end{multline}
where $\zeta'$ denotes the derivative of the Riemann $\zeta$-function,
or equivalently
\begin{equation}\label{expansion GUE}
\log \left[\frac{1}{N!}Z_{N,\sigma}^{GUE}\right]=-N^2 \left(\frac{3}{4}-\frac{1}{2}\log\frac{\sigma}{4}\right)+N\log 2\pi -\frac{1}{12}\log N +\zeta'(-1)+\bigO(N^{-1}).
\end{equation}

  Writing $ V=V_0(x)+ \sum_{j=1}^m t_jx^j$, the asymptotic expansion (\ref{largeNF}) was shown to be true in \cite{EM} for coefficients ${\bf t} = (t_{1}, \ldots, t_{m})$  such that $ |{\bf t}| < T$, and $t_{m} \ge \gamma \sum_{j=1}^{m-1} |t_{j}|$ for some $T>0$ and $\gamma>0$.   But while those were simple conditions on the external field, the asymptotic expansion (\ref{largeNF}) was expected to be true in greater generality.  

From the seminal work \cite{BIZ}, the collection of potentials $V$ for which the above expansion should hold true is best described in terms of support properties of the {\it equilibrium measure}, which is the unique probability measure $\mu_V$ achieving the minimum (over all probability measures $\mu$ on $\mathbb{R}$) of 
\begin{equation}\label{energy}
I_V(\mu)=\iint \log|x-y|^{-1}d\mu(x)d\mu(y)+\int V(x)d\mu(x) \ .
\end{equation}
For real analytic potentials $V$ with sufficient growth at $\pm\infty$,   the equilibrium measure $\mu_{V}$  is supported on a finite union of bounded disjoint intervals and has a smooth density \cite{DKM}. 
We say that $V$ is $k$-cut if $\mu_V$ is supported on $k$ disjoint intervals, and that it is regular if certain generic conditions hold which we will specify below.

In \cite{BI} it was shown that the partition function has the asymptotic expansion   (\ref{largeNF}) under the assumption that $V(x)$ is  $1$-cut and regular and the deformed external field $V_\tau(x)=(1-\tau^{-1})x^2+V(\tau^{-1/2}x)$ is also one-cut regular for all $\tau\in[1,+\infty)$.  The asymptotic expansion  (\ref{largeNF}) is expected to hold true in the general case in which only the potential  $V(x)$ is $1$-cut regular.  A by-product of the work in the present paper is a proof of this very general result.

{\bf The main result of this paper is the existence of a complete asymptotic expansion for the partition function for any polynomial external field $V$ which is $2$-cut and regular.  }

 \medskip
 
We note that it is possible to establish an asymptotic expansion for {\it derivatives} of $\log Z_{N}(V)$ with respect to parameters of  the external potential,  since there are  a variety of explicit formulae relating such logarithmic derivatives of the partition function to quantities for which the asymptotic behavior is computable via Riemann-Hilbert  (RH) techniques \cite{DKMVZ1}.
In the seminal work \cite{DKMVZ1} it is shown that in the multi-cut case the asymptotic expansion of  recurrence coefficients of orthogonal polynomials is expressed  in the large $N$ limit through highly oscillatory $\theta$  functions. This result shows that the asymptotic expansion of the  partition function $\log Z_{N}(V)$  in the large $N$ limit contains oscillatory terms.  However the main issue is to integrate the expansion in parameter space.
 In the $1$-cut case,  the knowledge of an asymptotic expansion for {\it derivatives} of $\log Z_{N}(V)$ with respect to parameters of  the external potential,  together with the complete knowledge of the behavior of the partition function for the reference potential   $V_{0}(x) = x^{2}/2$ allows one to {\it integrate in parameter space} and obtain an asymptotic expansion for the partition function itself.  However, for the two-cut case, integration from the Gaussian case is fraught with difficulties since one must integrate asymptotic expansions across phase transitions, where the nature of the asymptotic expansion changes see e.g. \cite{BI, BI1, BB, BMP, CV1}.  The central difficulty is actually finding an explicitly computable two-cut regular  case.  We are able to do this for a symmetric quartic two-cut regular external field,  $V_0(x)=x^4-4x^2$.
Once asymptotics for one particular two-cut external field $V_0$ are found, one can write any other polynomial two-cut external field $V_{\vec t}$, $\vec t=(t_1,\dots,t_{2d})$, as
\[V_{\vec t} (x)=V_0(x)+ \sum_{j=1}^{2d}t_jx^j,\qquad t_{2d}>0.\]
Using a suitable deformation of the form
\[V_{\vec t(\tau)}(x)=V_0(x)+ \sum_{j=1}^m t_j(\tau)x^j,\qquad t_j(0)=0,\quad t_j(1)=t_j,  \;\;\]
one can
obtain an identity for
$\frac{d}{d\tau}\log Z_N(V_{\vec t(\tau)})$ in terms of orthogonal polynomials.
We then show that  the collection of all two-cut regular  polynomial external fields is  path-wise connected.
 Once a deformation $V_{\vec{t}(\tau)}(x)$ exists connecting $V_{0}$ to $V_{\vec{t}}(x)$ in such a way that $V_{\vec{t}(\tau)}(x)$ is two-cut regular for all $\tau \in [0,1]$, the next task is to establish a uniform asymptotic expansion as $N\to \infty$ for logarithmic derivatives of the partition function which is valid for any two-cut regular polynomial external field.   We determine such expansion and show that each term of the expansion depends analytically on the potential $V$.
By integrating from our new base point through the space of parameters to any other two-cut regular case, we establish the existence of the full partition function expansion for any two-cut regular external field.

Let $[a_1,a_2]\cup[a_3,a_4]$, $a_1<a_2<a_3<a_4$ be the end points of the support of the equilibrium measure $\mu_V$  with respect to the two-cut regular  potential $V$  and let $\Omega$ be the fraction of eigenvalues in the interval $[a_3,a_4]$, namely $\int_{a_3}^{a_4}d\mu_V=\Omega$.

\noindent
{\bf  Our main result is the determination of the   asymptotic expansion of the partition function for  a two-cut  regular polynomial  potential $V$:}
%
\begin{equation}\label{ZNas3}
\begin{split}
\log Z_N(V)&= \log\left( \frac{N!}{  \lfloor\frac{N}{2}\rfloor ! \lfloor\frac{N+1}{2}\rfloor !} Z_{\lfloor\frac{N}{2}\rfloor,\sigma^*}^{GUE} Z_{\lfloor\frac{N+1}{2}\rfloor,\sigma^*}^{GUE}\right)\\
& -N^2F_0(V)-F_1(V)+\log\theta(N\Omega(V);B(V))
+ \bigO(N^{-1}),
\end{split}
\end{equation}
where $F_0=I_V(\mu_V)$ defined in (\ref{energy})  is the  so called planar limit \cite{Joh}, $F_1=F_1(V)$  is defined as 
\begin{equation}
\label{F10}
F_1=\frac{1}{24}\log \left(\dfrac{1}{2^8}\left(\dfrac{|A|}{2\pi}\right)^{12}|\Delta(a)|^3\prod_{j=1}^4 |\widehat\psi(a_j)|\right), \qquad \Delta(a)=\prod_{1\leq i<j\leq 4}(a_j-a_i),
\end{equation} where $a_1<a_2<a_3<a_4$ are the endpoints of the support of the equilibrium measure minimizing (\ref{energy}), and
the constant $A$ is the $\alpha$-period of the holomorphic differential  defined on the elliptic Riemann surface $y^2=(z-a_1)(z-a_2)(z-a_3)(z-a_4)$ given by 
\begin{align}
A=2\int_{a_2}^{a_3}\frac{1}{\prod_{j=1}^4\sqrt{\lambda-a_j}}d\lambda=-\frac{4K(k)}{\sqrt{(a_4-a_2)(a_3-a_1)}},
\end{align}
with  $K(k)$ the complete elliptic integral of the first kind with argument $k>0$ defined by
\[
\qquad k^2=\frac{(a_3-a_2)(a_4-a_1)}{(a_3-a_1)(a_4-a_2)}.
\]
Finally, the values $\widehat\psi(a_j)$ in (\ref{F10}) are given by
\begin{equation}
\widehat\psi(a_j)=\lim_{x\to a_j, x\in\supp\,\mu}\frac{\pi \psi(x)}{\sqrt{|x-a_j|}},
\end{equation}
where $\psi$ is the density of the equilibrium measure $\mu_V$ minimizing (\ref{energy}).
The quantity $F_1=F_1(V)$   has been first obtained, up to a constant,  by solving the loop equations in \cite{Akemann}.
The oscillatory term in (\ref{ZNas3})  is expressed by the Jacobi theta function with argument $N\Omega$  and period $B$, $\Im(B)>0$ defined as 
\begin{equation}
\label{theta0}
\theta(z;B)=\sum_{n\in\mathbb Z}e^{\pi in^2B + 2\pi inz},\quad B=i\frac{K(k')}{K(k)}, \qquad k'=\sqrt{1-k^2}.
\end{equation}
Oscillatory terms first appeared in the asymptotics of orthogonal polynomials related to random matrices in the seminal paper \cite{DKMVZ1}. 
Later  oscillatory terms appeared in the asymptotic expansion of the partition function itself in \cite{BDE}.
Numerical evidence of the oscillatory behaviour of the partition function in the large $N$ limit was first obtained  in \cite{Jurkiewicz}.
The first term on the right hand side of (\ref{ZNas3}) does not depend on the potential $V$ and is given by  products of GUE partition functions defined in (\ref{ZGUE-intro}) with $\sigma^*=4e^{3/2}$ and
where $ \lfloor\,.\,\rfloor$ stands for the integer part. The constant $\sigma^*$  is such that $\lim_{N\rightarrow\infty}N^{-2}\log Z_{N,\sigma^*}^{GUE}=0$.
We remark that while we have established the existence of a complete expansion to all orders, and while with our techniques we can in principle calculate recursively further terms of the asymptotic expansion (\ref{ZNas3}), obtaining explicit formulae beyond the first 3 or 4 terms remains a significant, albeit algebraic, challenge.

Up to the first term, the  asymptotic expansion  in (\ref{ZNas3}) has been computed in the physics literature for general classes of $V$, see e.g. \cite{Akemann,  BDE, dFGZJ},  including multi-cut $V$ \cite{Eynard, CE}. In particular in \cite{Eynard}, all terms of the asymptotic expansion of the logarithm of the  partition function $\log Z_N(V)$  have been calculated up to  an additive term. However, rigorous asymptotic expansions, without any prior assumptions on their existence, are a very delicate issue and have only been obtained very recently \cite{BG, Shcherbina}. The manuscripts \cite{BG, Shcherbina} attack the more general problem of determining the asymptotic behavior for general $\beta$.
For any  two-cut regular potential $V$ the asymptotic expansion (\ref{ZNas3}) clearly shows that the additive term 
\[
\log\left( \frac{N!}{  \lfloor\frac{N}{2}\rfloor ! \lfloor\frac{N+1}{2}\rfloor !} Z_{\lfloor\frac{N}{2}\rfloor,\sigma^*}^{GUE} Z_{\lfloor\frac{N+1}{2}\rfloor,\sigma^*}^{GUE}\right)
\]
is independent of $V$ and to the best of our knowledge, such expression  has not appeared in the literature  before. 
The representation of the asymptotic expansion in \cite{BG} and \cite{Shcherbina} is exceedingly implicit, making comparisons between those papers and our explicit results difficult.
 Expanding this  constant term as $N\to\infty$ we obtain
\begin{multline}\label{ZNas4}
\log Z_{N}(V)=-N^2F_0(V)+N\log N +(\log 2\pi-1)N +\frac{1}{3}\log N+\frac{1}{2}\log 2\pi +2\zeta'(-1)\\
+\frac{1}{6}\log 2-F_1(V)+\log\theta(N\Omega(V);B(V))
+ \bigO(N^{-1}).
\end{multline}
\begin{remark}
The above asymptotic expansion shows that the term of order $N\log N$ and the term of order $N$ are the same as in the Gaussian case  (see (\ref{expansion GUE1}))  while the term of order $\log N$ is different from the one-cut case as well as all the smaller order terms. 
For a general potential  $V$  with low regularity properties   the fact that $\log Z_N=-N^2F_0+N\log N +N(\log 2\pi -1)+\bigO(\log N)$  is  proved in the works \cite{Serfaty1, Serfaty2}, see also \cite{BPS}. \end{remark}

\section{Statement of results}

The main contributions of this paper are the following:
\begin{itemize}
\item[(1)] we use an approach using orthogonal polynomials and RH problems to obtain large $N$ asymptotics of the partition function $Z_N(V)$ for symmetric quartic polynomials $V$;
\item[(2)] we prove that any one-cut regular polynomial $V(x)$ can be deformed continuously to the Gaussian $x^2/2$ in such a way that the deformed external field remains one-cut regular and polynomial of degree $\leq \deg\,V$ throughout the deformation;
\item[(3)] we prove that any two-cut regular polynomial $V$ can be deformed continuously to a symmetric quartic polynomial $V_0$ in such a way that the deformed external field remains two-cut supported and polynomial of degree $\leq \deg\, V$ throughout the deformation;
\item[(4)]  for any two-cut regular polynomial external field  $V_{\vec t}(x)=V_0(x)+\sum_{j=1}^{2d}t_jx^j$, with $t_{2d}>0$ and $V_0$ some quartic symmetric reference potential, we establish an asymptotic expansion for {\it derivatives} of $\log Z_{N}(V_{\vec t})$,  with respect to the parameters $t_j$. We then obtain  the  large $N$ asymptotics of the partition function $\log Z_{N}(V_{\vec t})$  by integration in parameter space from the  reference potential $V_0$ to $V_{\vec t}$.

\end{itemize}

\subsection{Symmetric quartic $V$}
We first consider the case where $V$ is a symmetric quartic polynomial:
\begin{equation}\label{def V}
V_{r,s}(z)=\frac{1}{s}(x^4-r x^2),\qquad r> 2 \sqrt{s}.\end{equation}
Then the equilibrium measure minimizing (\ref{energy}) is given by \cite{BI, MOR}
\begin{equation}\label{mu}
d\mu_{r,s}(x)=\frac{2}{\pi s }|x|\sqrt{(b-x^2)(x^2-a)}dx,\qquad
x\in[-\sqrt{b},-\sqrt{a}]\cup[\sqrt{a},\sqrt{b}],
\end{equation}
where
\begin{equation}\label{a b}
a = \frac{1}{2} \left( r - 2 \sqrt{ s}\right), \ \ \ \ b = \frac{1}{2} \left( r + 2 \sqrt{ s} \right) \ .
\end{equation}
The condition $r> 2 \sqrt{s}$ is needed to have a two-cut supported equilibrium measure.
In order to derive large $N$ asymptotics for $Z_N(r,s):=Z_N(V_{r,s})$, we will need a number of identities which hold for finite $N$.
The following identities will be shown in Section \ref{section: diff id} and are crucial for our analysis.
\begin{proposition}\label{prop identities}
Let $V$ be given by (\ref{def V}) and $Z_N(r,s)$ by (\ref{def Zn}).
Then we have
\begin{align}&\label{partition rec2-intro}
\log Z_{2N}(r,s)=\log(2N)!+\log\widehat Z_{N}(-1/2;r,s)+\log\widehat Z_{N}(1/2;r,s),\\
& \log Z_{2N+1}(r,s)=\log(2N+1)!+\log\widehat
Z_{N+1}(-1/2;r,s_+)+\log\widehat Z_{N}(1/2;r,s_-),\label{partition rec2-intro2}
\end{align}
where
\begin{equation}\label{spm}
s_\pm = s\left(1\pm\frac{1}{2N+1}\right),
\end{equation}
and
\begin{equation}\label{def Zn2-intro}
\widehat Z_{N}(\alpha;r,\sigma)=\frac{1}{N!}\int_{\mathbb
R_+^N}\prod_{i<j}(x_i-x_j)^2 \prod_{j=1}^N
x_j^\alpha e^{-\frac{2N}{\sigma}(x_j^2-rx_j)}dx_j.
\end{equation}
\end{proposition}
This reduces the problem of finding asymptotics for $Z_N(r,s)$  to the problem of finding asymptotics for $\log\widehat Z_{N}(\alpha;r,\sigma)$ as $N\to\infty$ where $\sigma=s$ and $\sigma=s_\pm=s\left(1\mp \frac{1}{2N+1}\right)$. Defining 
$\dfrac{d}{d\alpha} \log \widehat Z_{N}(\alpha;r,\sigma)=G_{N}(\alpha;r,\sigma)$, one has
\begin{equation}\label{Zalpha-intro}
\log \widehat Z_{N}(\alpha;r,\sigma)=\log \widehat Z_{N}(0;r,\sigma)+\int_0^\alpha G_{N}(\alpha';r,\sigma)d\alpha',\end{equation}
and we will use a second crucial identity, which we will also prove in Section \ref{section: diff id},
expressing  $G_{N}(\alpha;r,\sigma)$ completely in terms of orthogonal polynomials with respect to the weight $x^\alpha e^{-N V_{r,\sigma}(x)}$ on the half-line $[0,+\infty)$, see (\ref{GN}) and (\ref{Y}) below.
Moreover, we will show in Section \ref{section: diff id} that $\log \widehat Z_{N}(0;r,\sigma)$ can be expressed for large $N$, up to an exponentially small error, in terms of the partition function for the GUE, namely
\begin{equation}\label{sec3: 3}
\log \widehat Z_{N}(0;r,\sigma)=\log\left(\dfrac{e^{\frac{N^2r^2}{2\sigma}}}{N!}Z^{GUE}_{N,\sigma}\right)+\bigO(e^{-cN}),\quad c>0,\qquad N\to\infty.
\end{equation}
We will show in Section \ref{section: as G} using a RH analysis that $G_{N}(\alpha,r,\sigma)$ admits a full asymptotic expansion in powers of $1/N$, in which the leading term is proportional to $N$.
Combining (\ref{partition rec2-intro})-(\ref{partition rec2-intro2}) with (\ref{Zalpha-intro}), and substituting asymptotics for $G_{N}(\alpha;r,\sigma)$, we will be able to prove the following result (see Section~\ref{section: as G}).

\begin{theorem}\label{theorem 1}
Let $Z_N=Z_{N}(r,s)$ be the partition function defined by (\ref{def Zn}) with $V=V_{r,s}$ given by (\ref{def V}). As $N\to\infty$ for fixed $s>0$, $r>2\sqrt{s}$, there exist constants $c^{(j)}, \widetilde c^{(j)}$ for $j\in\mathbb N$ such that we have an asymptotic expansion of the form
\begin{multline}\label{ZNas}
\log Z_{N}(r,s)=\log\left[\frac{N!Z_{\lfloor\frac{N}{2}\rfloor,\sigma^*}^{GUE} Z_{\lfloor\frac{N+1}{2}\rfloor,\sigma^*}^{GUE}}{\left\lfloor\frac{N}{2}\right\rfloor ! \left\lfloor\frac{N+1}{2}\right\rfloor !}\right]-N^2F_0\\
+\begin{cases}\frac{1}{2}\log \frac{\sqrt{b}+\sqrt{a}}{2(ab)^{1/4}} + \sum_{j=1}^k c^{(j)}N^{-j} & N \mbox{even}\\
\frac{1}{2}\log\frac{\sqrt{b}-\sqrt{a}}{2(ab)^{1/4}}+ \sum_{j=1}^k \widetilde c^{(j)}N^{-j}& \mbox{ N odd }\end{cases}
+\bigO(N^{-k-1}),
\end{multline}
for any $k\in\mathbb N$,
where $a=\frac{1}{2}(r-2\sqrt{s})$ and $b=\frac{1}{2}(r+2\sqrt{s})$, $F_0$ is the
limit of the free energy given by
\begin{equation}\label{F0}
F_0(r,s):=-\lim_{N\to\infty}\frac{1}{N^2}\log Z_{N}(r,s)=\frac{3}{8} -\frac{1}{4}\log \frac{s}{4}-\frac{r^2}{4s},
\end{equation}
and $Z_{N,\sigma^*}^{GUE}$ is defined in (\ref{ZGUE-intro}) with $\sigma^*=4e^{3/2}$.
The coefficients $c^{(k)}(r,s)$, $\tilde c^{(k)}(r,s)$ are real analytic functions of $s>0$ and $r>2\sqrt{s}$, independent of $N$.
\end{theorem}
\begin{remark}
Using the expansion (\ref{expansion GUE1}), we can write (\ref{ZNas}) in the form
\begin{multline}\label{ZNas2}
\log Z_{N}=-N^2F_0+N\log N +(\log 2\pi-1)N +\frac{1}{3}\log N+\frac{1}{2}\log 2\pi +2\zeta'(-1)\\
+\frac{1}{6}\log 2+\begin{cases}\frac{1}{2}\log \frac{\sqrt{b}+\sqrt{a}}{2(ab)^{1/4}} & N \mbox{even}\\
\frac{1}{2}\log\frac{\sqrt{b}-\sqrt{a}}{2(ab)^{1/4}}& N \mbox{odd }\end{cases}
+ \bigO(N^{-1}).
\end{multline}
Some terms of the above expansions have appeared also in the work \cite{schiappa}.
\end{remark}
\begin{remark}
Our formula (\ref{ZNas}) can be written in the following more general form which is more familiar  in the physics literature \cite{Akemann, Eynard}:
\begin{equation}\label{ZNas3b}
\log \left[\frac{Z_{N}\left\lfloor\frac{N}{2}\right\rfloor ! \left\lfloor\frac{N+1}{2}\right\rfloor !}{N!Z_{\lfloor\frac{N}{2}\rfloor,\sigma^*}^{GUE} Z_{\lfloor\frac{N+1}{2}\rfloor,\sigma^*}^{GUE}}\right]
=-N^2F_0-F_1+\log\theta(\frac{N}{2};B)
+ \bigO(N^{-1})
\end{equation}
with $F_1$ defined in (\ref{F10}) and the theta-function defined in (\ref{theta0}).


In the symmetric case we consider, by (\ref{mu}), we have
\begin{equation}\label{sym1}
a_4=-a_1=\sqrt{b},\qquad a_3=-a_2=\sqrt{a},\quad \psi(x)= \frac{2}{\pi s }|x|\sqrt{(b-x^2)(x^2-a)},
\end{equation}
and 
\begin{equation}\label{sym2}
k=\frac{2(ab)^{1/4}}{\sqrt{a}+\sqrt{b}},\qquad \Delta(a)=4\sqrt{ab}(b-a)^2.
\end{equation} Using the identity
$
K(k)=\frac{\pi}{2}\theta(0;B)^2$,
we have \begin{equation}\label{sym3}
A=-\frac{2\pi\theta(0;B)^2}{\sqrt{b}+\sqrt{a}}.\end{equation}
For $N$ even, we can use the periodicity property of the $\theta$-function, $\theta(z+1;B)=\theta(z;B)$, and substitute (\ref{sym1}), (\ref{sym2}), and (\ref{sym3}) into (\ref{ZNas3b}). Using also the fact that $s=(b-a)^2/4$, we 
recover (\ref{ZNas}) from (\ref{ZNas3b}).
For $N$ odd, (\ref{ZNas}) again follows from (\ref{ZNas3b}), but one must also use the identity 
\[\log\theta(\frac{1}{2};B)-\log\theta(0;B)=\frac{1}{4}\log(1-k^2)=\frac{1}{2}\log\frac{\sqrt{b}-\sqrt{a}}{\sqrt{b}+\sqrt{a}}.\] 
\end{remark}
\begin{remark}
The quantity $F_1$ defined in (\ref{F10}) coincides with the one defined in \cite{Akemann, DubrovinZhang, Eynard}  up to a constant factor $\dfrac{1}{24}\log(2^4\pi^{-12})$.

\end{remark}
\begin{remark}
For values of $r$ near $2\sqrt{s}$, a phase transition takes place between one-cut and two-cut supported external fields. If $r=2\sqrt{s}$, the equilibrium density vanishes at an interior point. This situation is studied in detail in \cite{BI}: the limit of the free energy $F_0(r,s)$ is not analytic near this point, and the subleading term in the large $N$ expansion can be expressed in terms of the Tracy-Widom distribution in a double scaling limit.
\end{remark}

\subsection{Equilibrium measures}

We define $P_m$ as the space of real polynomials $V$ of even degree $\leq m$ with positive leading coefficient and such that $V(0)=0$, and we write $\mathcal P$ for the space of Borel probability measures on $\mathbb R$.
To $V\in P_m$, we associate as before the unique probability measure $\mu_V\in\mathcal P$ which minimizes the logarithmic energy $I_V(\mu)$ defined by (\ref{energy}).
The equilibrium measure is characterized by the variational conditions \cite{SaTo}
\begin{align}
&\label{var eq} 2\int\log|x-y|d\mu_V(y)- V(x)=\ell_V,&\mbox{ on $\supp\,\mu_V$},\\
&\label{var ineq} 2\int\log|x-y|d\mu_V(y)- V(x)\leq \ell_V,&\mbox{ on $\mathbb R$}.
\end{align}
For a general polynomial $V\in P_m$, $\mu_V\in \mathcal P$ is supported on a finite union of at most $\deg V/2$ disjoint bounded intervals. If $V\in P_m$ is $k$-cut supported, the measure $\mu_V$ has the form
\begin{equation}\label{densitymu}
d\mu_V(x)=\frac{1}{c}\sqrt{|\mathcal{R}(x)|}|h(x)|dx, \qquad x\in\cup_{j=1}^k[a_j,b_j],
\end{equation}
where \begin{equation}\sqrt{\mathcal \mathcal R(x)}=\prod_{j=1}^k(x-a_j)(x-b_j),\end{equation} $h$ is a monic polynomial of degree ${\rm deg\,}V-k-1$, and $c$ is a normalizing constant \cite{DKM}. Generically (\ref{var ineq}) is strict for $x\in\mathbb R\setminus\supp\,\mu_V$, but in critical cases, equality can hold at points exterior to the support \cite{KM}.
If $h(x)\neq 0$ for all $x\in\supp\,\mu_V$ and if (\ref{var ineq}) is strict for all $x\in\mathbb R\setminus\supp\,\mu_V$, then $V$ is called $k$-cut regular. We write $P_m^{(k)}\subset P_m$ for the subset of $P_m$ containing all $k$-cut regular external fields of degree $\leq m$.

\begin{Example}\label{exampel: 1}
A simple example of a one-cut regular external field which we will encounter later on is $V(x)=\frac{2}{s}(x^2-rx)$. The corresponding equilibrium measure is supported on $[a,b]$ with $a$ and $b$ given by (\ref{a b}), and is given by
\begin{equation}d\mu_V(x)=\frac{2}{\pi s}\sqrt{(b-x)(x-a)}dx.\end{equation}
The quartic symmetric polynomial (\ref{def V}) is two-cut regular for $r>2\sqrt{s}$, one-cut regular for $r<2\sqrt{s}$, and one-cut singular for $r=2\sqrt{s}$, with an equilibrium density which vanishes in the middle $r/2$ of its support.
\end{Example}

The following result states that any one-cut regular external field $V(x)$ can be deformed continuously to the Gaussian $x^2/2$ in such a way that the deformed external field remains one-cut regular.

\begin{theorem}\label{theorem: 1}
Write 
$V_{\vec
t}(x)=\frac{x^2}{2}+\sum_{j=1}^mt_jx^j$, and let $V=V_{\vec t}\in P_m^{(1)}$, $m\geq 2$ be one-cut regular. There exist a continuous path $\vec t:[0,1]\to\mathbb R^m:\tau \mapsto \vec t(\tau)=(t_1(\tau),\ldots, t_m(\tau ))$ in $\mathbb R^m$ such that
\begin{itemize}
\item[(i)] $\vec t(0)=0$, or equivalently $V_{\vec t(0)}(x)=\frac{x^2}{2}$,
\item[(ii)] $\vec t(1)=\vec t$,
\item[(iii)] for all $\tau\in[0,1]$, $V_{\vec t(\tau)}\in P_m^{(1)}$.
\end{itemize}
\end{theorem}
The above theorem will be proved in Section~\ref{Sec1-cut}.
\begin{theorem}\label{theorem: 2}
Write 
\[V_{\vec
t}(x)=V_0(x)+\sum_{j=1}^mt_jx^j,\qquad V_0(x)=x^4-4x^2,\] and
let $V=V_{\vec t}\in P_m^{(2)}$, $m\geq 4$,  be two-cut regular. There exist  a continuous path $\vec t:[0,1]\to\mathbb R^m:\tau \mapsto \vec t(\tau)=(t_1(\tau),\ldots, t_m(\tau ))$ in $\mathbb R^m$ such that
\begin{itemize}
\item[(i)] $\vec t(0)=0$, or equivalently $V_{\vec t(0)}(x)=x^4-4x^2$,
\item[(ii)] $\vec t(1)=\vec t$,
\item[(iii)] for all $\tau\in[0,1]$, $V_{\vec t(\tau)}\in P_m^{(2)}$.
\end{itemize}
\end{theorem}
The proof of the above theorem is given in Section~\ref{section: 2-cut}.

\begin{remark}
The above results imply that the sets $P_m^{(1)}$ and $P_m^{(2)}$ are path-connected. Indeed, given two polynomials $V_1, V_2$ in $P_m^{(k)}$ for $k=1,2$, the theorems provide two $k$-cut regular paths which can be composed to obtain a continuous path connecting $V_1$ to $V_2$. We believe that the set $P_m^{(k)}$ is path-connected for any $k\in\mathbb
N$, but a generalization of our proofs would become considerably more complicated for
$k>2$. Theorem \ref{theorem: 1} is not needed for our study of the partition function in the two-cut case, but it allows to extend the proof in \cite{EM} of the existence of an asymptotic expansion of the form (\ref{largeNF}) to the entire set of one-cut regular polynomials $V$. This is explained in the remark after Theorem 1.1 in \cite{EM}.
\end{remark}

\subsection{General two-cut regular $V$}
We consider a deformation of the quartic symmetric external field of the form $V_{\vec t}(x)=V_0(z)+\sum_{j=1}^{2d}t_j x^j,$ with $t_{2d}>0$, $d\geq 2$, and $V_0(x)=x^4-4x^2$. The deformed external field does not need to be symmetric.
Let $P_n(x)$ be the orthonormal polynomials with respect to the weight $e^{-NV_{\vec t}(x)}$, defined by
\[
\int_{\mathbb R}P_{n}(x)P_{m}(x)e^{-NV_{\vec t}(x)}dx=\delta_{mn},\qquad
P_{n}(x)=\kappa_{n}x^n+ \ldots,
\]
and let us consider the $2\times 2 $ matrix 
\begin{equation}\label{X20}
X(z)=X^{(N)}(z;\vec t)=\begin{pmatrix} \kappa_N^{-1}P_N(z)&
\frac{1}{2\pi i}\kappa_N^{-1}\int_{\mathbb R}P_N(s)\frac{e^{-NV_{\vec t}(s)}ds}{s-z}\\
-2\pi i\kappa_{N-1}P_{N-1}(z)&
-\kappa_{N-1}\int_{\mathbb R}P_{N-1}(s)\frac{e^{-NV_{\vec t}(s)}ds}{s-z}
\end{pmatrix},
\end{equation}
defined for $z\in\mathbb C\setminus\mathbb R$.
 
For the partition function  (\ref{def Zn}) associated to the external field $V_{\vec t}$  we use the relation (see (\ref{diff id t res}))
\begin{equation}
\label{timederivative0}
\frac{d}{d t_k}\log Z_{N}(\vec t)=\frac{N}{2}\Res_{z=\infty}\left[\Tr\left(X^{-1}(z;\vec t)X'(z;\vec t)\sigma_3\right)z^k\right],
\end{equation}
where $'$ denotes derivative with respect to $z$.  The r.h.s. of the above expression is the derivative with respect of the parameters  $t_k$ of the isomonodromic $\tau$-function introduced by the Japanese school  \cite{JMU}.  The identification of the partition function $Z_{N}(\vec t)$ with the isomonodromic $\tau$-function was obtained  in \cite{BEH} (see also \cite{BEH1}). For a general potential which can also have singularities, the deformation of the $\tau$-function with respect to parameters in the potential was obtained in \cite{B3}.
The matrix function $X$  in (\ref{timederivative0}) is the solution to the standard RH problem for the orthogonal polynomials $P_n(z)$ \cite{FIK}. For a general $j$-cut regular external field $V$, the leading order term in the large $N$ asymptotic expansion for $X$ has been obtained using the Deift/Zhou steepest descent method in \cite{DKMVZ2}, and the method used in this paper allows to compute subleading terms as well. 

For our purposes, in the case of a two-cut regular external field $ V_{\vec t}$  with support $[a_1,a_2]\cup [a_3,a_4]$ and equilibrium measure $\mu_{V_{\vec t}}$, we will need the leading and first subleading term in the asymptotic expansion of the matrix $X(z)$ as $N\to \infty$. 
 Evaluating the residue in (\ref{timederivative0}) and using Fay's identities \cite{fay}, we  obtain the large $N$ asymptotic expansion
\begin{multline*}
-\frac{\partial}{\partial t_k}\log Z_N(V_{\vec t})=N^2\dfrac{\partial F_0}{\partial t_k} -N(\log\theta\left(N\Omega;B\right))'\dfrac{\partial \Omega}{\partial t_k}+\dfrac{\partial F_1}{\partial t_k}-
\dfrac{\partial}{\partial B}\log\theta\left(N\Omega;B\right)\dfrac{\partial B}{\partial t_k} \\
-(\log\theta\left(N\Omega;B\right))''F_1^{(1)}\dfrac{\partial \Omega}{\partial t_k}-\dfrac{F_0^{(3)}}{6}\left(\dfrac{\left(\theta\left( N \Omega;B\right)\right)'''}{\theta\left( N \Omega;B\right)}\right)'\dfrac{\partial \Omega}{\partial t_k}
+\bigO(N^{-1}),
\end{multline*}
where $F_0=I_{V_{\vec t}}(\mu_{V_{\vec t}})$ is  defined in (\ref{energy}), the functional form of $F_1=F_1(V_{\vec t})$ is defined in (\ref{F10}), and $\theta(N\Omega,B)$ is defined in (\ref{theta0}) with $\Omega=\int_{a_3}^{a_4}d\mu_{V_{\vec t}}$. 
The quantities   $F_1^{(1)}$ and $F_0^{(3)}$ are defined in  (\ref{F11}) and (\ref{F03})  respectively and coincide, up to  multiplicative  factors, with quantities  first obtained in  \cite{Eynard}. The above expansion has been proved in theorem~\ref{theorem86}.
From the above expansion, it is straightforward to  arrive to the following theorem.
\begin{theorem}\label{theorem211}
Let $V_{\vec t}\in P_m^{(2)}$ be a two-cut regular polynomial external field.
The derivative with respect to the parameter $t_k$, $k=1,\dots,2d$ of $\log  Z_N(V_{\vec t})$ has the asymptotic expansion 
\begin{equation}
\label{dlogZ}
-\frac{\partial}{\partial t_k}\log Z_N(V_{\vec t})=\dfrac{\partial}{\partial t_k}\left(N^2 F_0(V_{\vec t})+F_1(V_{\vec t})-\log\theta\left(N\Omega(V_{\vec t});B(V_{\vec t})\right)\right)+\bigO\left(N^{-1}\right),
\end{equation}
where $F_0=I_{V_{\vec t}}(\mu_{V_{\vec t}})$ is  defined in (\ref{energy}), the functional form of $F_1=F_1(V_{\vec t})$ is defined in (\ref{F10}), and $\theta(N\Omega;B)$ is defined in (\ref{theta0}) with $\Omega=\int_{a_3}^{a_4}d\mu_{V_{\vec t}}$. The above expansion is uniform for $\vec t$ in compact subsets of $\{\vec t\in\mathbb R^m:V_{\vec t}\in P_m^{(2)}\}$.  
\end{theorem}
From Theorem~\ref{theorem: 2}, we can choose a continuous path $\vec t:[0,1]\to\mathbb R^m:\tau\mapsto \vec t(\tau)=(t_1(\tau),\ldots, t_m(\tau))$ in $\mathbb R^m$ such that
 $V_{\vec t(0)}(x)=x^4-4x^2$, $V_{\vec t(1)}(x)=V_{\vec t}(x)$, and for all $\tau \in[0,1]$, $V_{\vec t(\tau)}$ is  two-cut regular.
 Therefore one can integrate (\ref{dlogZ}) to obtain
 \begin{multline*}
\log Z_N(V_{\vec t})-\log Z_N(V_{0})=-N^2(F_0(V_{\vec t})-F_0(V_0))-F_1(V_{\vec t})+F_1(V_0)\\
+\log\theta\left(N\Omega(V_{\vec t});B(V_{\vec t})\right)-\log\theta\left(N\Omega(V_{0});B(V_0)\right)+\bigO(N^{-1}),\qquad N\to\infty.
\end{multline*}
Comparing the above relation with (\ref{ZNas}), we obtain our main result.
\begin{theorem}\label{theorem212}
For a general two-cut regular polynomial potential $V$, the partition function  $Z_N(V)$ defined in (\ref{def Zn})  has the large $N$ asymptotic expansion
\begin{multline}\label{ZNas4b}
\log Z_N(V)= \log\left( \frac{N!}{  \lfloor\frac{N}{2}\rfloor ! \lfloor\frac{N+1}{2}\rfloor !} Z_{\lfloor\frac{N}{2}\rfloor,\sigma^*}^{GUE} Z_{\lfloor\frac{N+1}{2}\rfloor,\sigma^*}^{GUE} \right) \\
 -N^2F_0(V)-F_1(V)+\log\theta(N\Omega(V);B(V))
+ \bigO(N^{-1}),
\end{multline}
where $F_0=I_{V}(\mu_V)$ is  defined in (\ref{energy}),    the functional form of $F_1=F_1(V)$ is defined in (\ref{F10}), and $\theta(N\Omega,B)=\theta(N\Omega(V));B(V)$ 
is defined in (\ref{theta0}) with $\Omega(V)=\int_{a_3}^{a_4}d\mu_{V}$.   The constant $\sigma^*$  in the GUE partition function, defined in (\ref{ZGUE-intro}), is $\sigma^*=4e^{3/2}$.
\end{theorem}
The above two theorems are proved in Section~\ref{sec:FinalProof}
\begin{remark}
In our derivation of the above formula, we consider the Riemann surface $y^2=(z-a_1)(z-a_2)(z-a_3)(z-a_4)$ with the canonical homology basis of cycles  $\{\alpha,\beta\}$ as shown in Figure \ref{fig1}.
\begin{figure}[htb!]
    \includegraphics[width=1.0\textwidth]{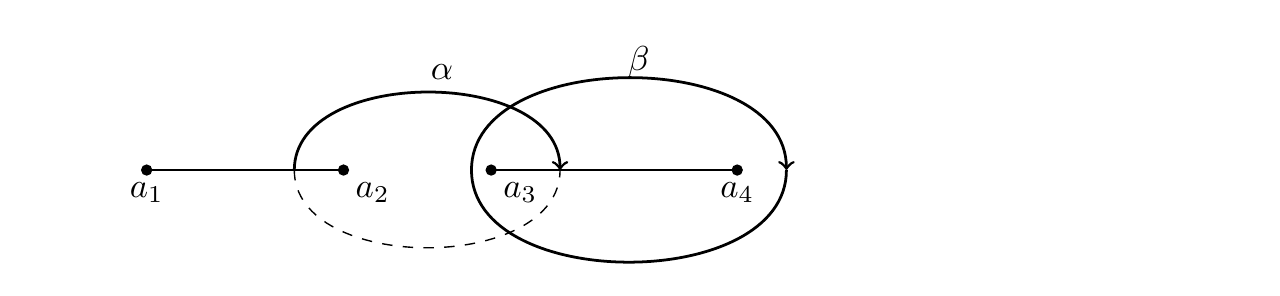}
 \caption{The canonical homology basis.}
 \label{fig1}
\end{figure}
By performing a symplectic transformation
\[
\tilde{\alpha}=\beta,\quad \tilde{\beta}=-\alpha,
\]
the  corresponding modulus $\tilde{B}=-\dfrac{1}{B}$  and the Jacobi $\theta$-function  transform as
\begin{equation}
\label{modular}
\theta(z;B)=\dfrac{1}{\sqrt{-iB}}e^{\pi i z^2\tilde{B}}\theta(z\tilde{B};\tilde{B}).
\end{equation}
Introducing the $\theta$-function with characteristics as 
\begin{equation}\label{theta car}
\theta[^\delta_\epsilon](z;B)=\sum_{n\in\mathbb{Z}} e^{\pi i (n+\delta)^2B+2\pi i (z+\epsilon)(\delta+n)},
\end{equation}
one verifies by (\ref{modular}) that 
\begin{equation}
\label{modular1}
\theta(N\Omega;B)=\theta[^{\,\,\,0}_{N\Omega}](0;B)=\dfrac{1}{\sqrt{-iB}}\theta[_{\,\,\,0}^{N\Omega}](0;\tilde{B}).
\end{equation}
 The period  $\tilde{A}$ of the holomorphic differential  with respect to the basis $\tilde{\alpha}$ 
takes the form
\begin{equation}
\label{modular2}
\tilde{A}=(-iB)A.
\end{equation}
Therefore, combining (\ref{modular1}) and (\ref{modular2}), the  expansion of the partition function with respect to the homology basis $(\tilde{\alpha},\tilde{\beta}) $  can be written in the form
\begin{multline}\label{ZNas5}
\log Z_N(V_{\vec t})= \log\left( \frac{N!}{  \lfloor\frac{N}{2}\rfloor ! \lfloor\frac{N+1}{2}\rfloor !} Z_{\lfloor\frac{N}{2}\rfloor,\sigma^*}^{GUE} Z_{\lfloor\frac{N+1}{2}\rfloor,\sigma^*}^{GUE} \right) \\
 -N^2\tilde{F}_0-\tilde{F}_1+\log\theta[^{N\Omega}_{\,\,\,0}](0;\tilde{B})
+ \bigO(N^{-1}),
\end{multline}
where $\tilde{F}_0=F_0$
and
\[
\tilde{F}_1=\frac{1}{24}\log \left(\dfrac{1}{2^8}\left(\dfrac{|\tilde{A}|}{2\pi}\right)^{12}|\Delta(a)|^3\prod_{j=1}^4 |\widehat\psi(a_j)|\right).
\]
This second choice of homology basis is the one used in \cite{BG} and \cite{Eynard}.
\end{remark}

\subsection*{Outline}
The remaining part of the manuscript is organized as follows.
In Section~\ref{SecOP}, we use the orthonormal polynomials with respect to the weight $w_N(x;\alpha,r,\sigma)=x^{\alpha}e^{-\frac{2N}{\sigma}(x^2-rx)} $ on the 
half line $[0,+\infty)$ to prove the identities in Proposition \ref{prop identities}, and
we express the partition function $Z_N(r,s)$ with respect to the potential $
V_0(x)=V_{r,s}(x)=\dfrac{1}{s}(x^4-rx^2)$ in terms of  quantities  related to those polynomials.
This will lead to a proof of (\ref{Zalpha-intro}) and (\ref{sec3: 3}).

In Section~\ref{SecAsym}, we obtain leading and subleading asymptotics for the orthogonal polynomials on the half line using a steepest descent analysis of the associated RH problem.

In Section~\ref{section: as G}, we use the results from Section~\ref{SecOP} and Section~\ref{SecAsym} to derive  the asymptotic expansion of the partition function
$\log Z_N(r,s)$  as $N\to \infty$ up to terms of order $\bigO(N^{-1})$, thus proving Theorem \ref{theorem 1}.

In Section~\ref{Sec1-cut}, we show that the space of  one-cut regular  polynomial potential is path-wise connected (Theorem~\ref{theorem: 1}). 

In Section~\ref{section: 2-cut}, we show that the space of 
two-cut regular polynomial potential is path-wise connected (Theorem~\ref{theorem: 2}).

In the last section, we determine the differential identities for the partition function $\log(Z_N(V_{\vec t}))$  with respect to the parameters $t_k$ for a two-cut regular polynomial potential of the form
$V_{\vec t}(z)=V_0(z)+\sum_{k=1}^{2d}t_kz^k$ and we calculate the asymptotic expansion as $N\to \infty$  of these differential identities up to terms of order $\bigO(N^{-1})$.
We then integrate in the space of parameters  from the reference potential $V_0(z)$ up to any two-cut potential $V_{\vec t}(z)$, thus proving our main result, Theorem~\ref{theorem212}.

\section{Orthogonal polynomials and differential identities\label{SecOP}}

\subsection{ Proof of Proposition~\ref{prop identities}}
 Let $Z_N(V)$ be  the partition function defined in (\ref{def Zn}).  It is a standard fact that, for general $V$,
\begin{equation}\label{partition rec}
Z_{N}(V)=N!\ \prod_{n=0}^{N-1}\kappa_{n,N}^{-2}(V),
\end{equation}
where $\kappa_{n,N}=\kappa_{n,N}(V)>0$ is the leading coefficient of the degree $n$ orthonormal polynomial defined by
\begin{equation}\label{OP}
\int_{\mathbb R}P_{n,N}(x)P_{m,N}(x)e^{-NV(x)}dx=\delta_{mn},\qquad
P_{n,N}(x)=\kappa_{n,N}x^n+ \ldots 
\end{equation}
Now let   $V=V_0=V_{r,s}$ be  the quartic symmetric polynomial given by (\ref{def V}), and consider orthonormal polynomials $p_{n,N}(x;\alpha,r,\sigma)$ on $\mathbb
R^+$ with respect to the weight
\begin{equation}w_N(x;\alpha,r,\sigma)=x^{\alpha}e^{-\frac{2N}{\sigma}(x^2-rx)}, \label{w1}
\end{equation}
so that we have the orthogonality conditions
\begin{equation}\label{OP2}
\int_0^{+\infty}p_{n,N}(x;\alpha,r,\sigma)p_{m,N}(x;\alpha,r,\sigma)w_N(x;\alpha,r,\sigma)dx=\delta_{mn}.
\end{equation}
We express the leading and subleading terms as $x\to\infty$ of $p_{k,N}$ as follows,\begin{equation}\label{ckdk}
p_{k,N}(x;\alpha,r,\sigma)=\gamma_{k,N}(\alpha,r,\sigma)
x^k(1+\frac{c_{k,N}(\alpha,r,\sigma)}{x}+\frac{d_{k,N}(\alpha,r,\sigma)}{x^2}+\ldots),\quad
\gamma_{k,N}>0.
\end{equation}
Applying the change of variable $x=u^2$ to the integral in (\ref{OP2}) for $\sigma=s$, one obtains orthogonality relations for the polynomials $p_{n,N}(u^2)$ and $up_{n,N}(u^2)$, and it is straightforward to verify that they obey precisely the orthogonality relations  (\ref{OP}) with $V=V_0=V_{r,s}$ for $P_{2n,N}$ (with $\alpha=-1/2$) and $P_{2n+1,N}$ (with $\alpha=1/2$), so that by uniqueness we have the identities
\begin{equation}\label{OP id}
P_{2n,2N}(u;r,s)=p_{n,N}(u^2;-1/2,r,s),\qquad P_{2n+1,2N}(u;r,s)=up_{n,N}(u^2;1/2,r,s).
\end{equation}
Similar identities have been derived previously, see for example \cite[Appendix B]{ADDV}, \cite{Forrester}, and \cite{CK2}.
In particular, we have by (\ref{OP id}),
\begin{equation}\label{lc id}
\kappa_{2k,2N}(r,s)=\gamma_{k,N}(-1/2;r,s),\qquad
\kappa_{2k+1,2N}(r,s)=\gamma_{k,N}(1/2;r,s).
\end{equation}
Hence, by (\ref{partition rec}), we obtain
\begin{align}&\label{partition rec2}
\log Z_{2N}(r,s)=\log(2N)!+\log\widehat Z_{N}(-1/2;r,s)+\log\widehat Z_{N}(1/2;r,s),\\
&\label{partition rec3} \log Z_{2N+1}(r,s)=\log(2N+1)!+\log\widehat
Z_{N+1}(-1/2;r,s_+)+\log\widehat Z_{N}(1/2;r,s_-),
\end{align}
with $s_\pm$ as in (\ref{spm}), and
\begin{equation}\label{hat Zn}
\log \widehat Z_{N}(\alpha;r,\sigma)=-2\sum_{n=0}^{N-1}\log\gamma_{n,N}(\alpha;r,\sigma).
\end{equation}
The quantity $\widehat Z_{N}(\alpha;r,\sigma)$ is the partition function corresponding to a unitary random matrix ensemble for positive-definite matrices, and following standard arguments, it can be expressed as the $N$-fold integral (\ref{def Zn2-intro}). This proves Proposition \ref{prop identities}.
In order to obtain asymptotics for $Z_{N}(r,s)$, it is thus sufficient to
have asymptotics for $\widehat Z_{N}(1/2;r,\sigma)$ and $\widehat Z_{N}(-1/2;r,\sigma)$ as $N\to\infty$, with $\sigma=s$ and $\sigma=s_\pm$.
The advantage of
this formula is that the equilibrium problem associated to the partition function $\widehat Z_N(\alpha;r,\sigma)$, is to minimize the quantity
\begin{equation}\label{energy2}
I_{r,\sigma}(\nu)=\iint \log|x-y|^{-1}d\nu(x)d\nu(y)+\frac{2}{\sigma}\int (x^2-rx)d\nu(x)
\end{equation}
among all Borel probability measures on $\mathbb R^+$, and that the support of the minimizer consists of one single interval, as opposed to the support of $\mu_{r,s}$.
We have indeed
\begin{equation}\label{nu}
d\nu_{r,\sigma}(x)=\frac{2}{\pi \sigma}\sqrt{(b_{r,\sigma}-x)(x-a_{r,\sigma})}dx,\qquad x\in [a_{r,\sigma},b_{r,\sigma}],
\end{equation}
which is a semi-circle distribution, with $a_{r,\sigma}>0$ and $b_{r,\sigma}>0$ given by
\begin{equation}
a_{r,\sigma}=\frac{1}{2}(r-2\sqrt{\sigma}),\qquad b_{r,\sigma}=\frac{1}{2}(r+2\sqrt{\sigma})
,\label{absigma}\end{equation}
 see Example \ref{exampel: 1}.

\subsection{Differential identity for $\log\widehat Z_{N}(\alpha)$}\label{section: diff id}
In what follows, we will consider $r$ and $\sigma$ as fixed parameters, and we omit them in our notations, writing for example $a,b,\nu, \widehat Z_N(\alpha)$ instead of $a_{r,\sigma},b_{r,\sigma},\nu_{r,\sigma}, \widehat Z_N(\alpha;r,\sigma)$.
In this section, $N$ will also be fixed and we will also use the abbreviated notations $p_{n},\gamma_n, c_n, d_n, w(x)$ instead of $p_{n,N},\gamma_{n,N}, c_{n,N}, d_{n,N}, w_N(x)$
when there is no possible confusion.
We will now derive a differential identity for $\log\widehat Z_{N}(\alpha)$
as a function of $\alpha$, following similar ideas as in \cite{K}.  By
(\ref{hat Zn}), we have
\begin{eqnarray}
\frac{d}{d\alpha}\log
\widehat Z_{N}(\alpha)&=&-2\sum_{n=0}^{N-1}\frac{\gamma_n'(\alpha)}{\gamma_n(\alpha)}  \nonumber\\
&=&-\sum_{n=0}^{N-1}\int_0^{+\infty}\frac{\partial}{\partial\alpha}\left(p_{n}^2(x;\alpha)\right)w(x;\alpha)dx\label{idhat Zn}.
\end{eqnarray}
Recalling the Christoffel-Darboux formula
\begin{equation}
\label{CD} \sum_{n=0}^{N-1}p_{n}^2(x;\alpha)=a_{N-1}(\alpha)
(p_{N}'(x;\alpha)p_{N-1}(x;\alpha)-p_{N}(x;\alpha)p_{N-1}'(x;\alpha)),
\end{equation}
where $a_n$ is the recurrence coefficient defined by
\begin{equation}\label{recurr}
xp_{n}(x;\alpha)=a_{n}(\alpha)p_{n+1}(x;\alpha)+b_n(\alpha)p_{n}(x;\alpha)+a_{n-1}(\alpha)p_{n-1}(x;\alpha),
\end{equation}
and substituting it into (\ref{idhat Zn}), we obtain
\begin{multline}\label{diff id alpha}
\frac{d}{d\alpha}\log \widehat Z_{N}(\alpha)
=-\int_0^{+\infty}\frac{\partial}{\partial\alpha}\left(a_{N-1}(\alpha)p_{N}'(x;\alpha)p_{N-1}(x;\alpha)\right)w(x;\alpha)dx
\\ \qquad+\int_0^{+\infty}\frac{\partial}{\partial\alpha}\left(a_{N-1}(\alpha)p_{N}(x;\alpha)p_{N-1}'(x;\alpha)\right)
w(x;\alpha)dx.\end{multline}
By the orthogonality conditions for the orthogonal polynomials, we obtain
\begin{multline}
\frac{d}{d\alpha}\log \widehat Z_{N}(\alpha)
=-Na_{N-1}'(\alpha) \frac{\gamma_N(\alpha)}{\gamma_{N-1}(\alpha)}
-a_{N-1}(\alpha)\int_0^{+\infty}\left( \frac{\partial}{\partial\alpha} p_{N}'(x;\alpha) \right)
p_{N-1}(x;\alpha)w(x;\alpha)dx\\
\qquad -a_{N-1}(\alpha)\int_0^{+\infty}p_{N}'(x;\alpha)
\left( \frac{\partial}{\partial\alpha}p_{N-1}(x;\alpha) \right) w(x;\alpha)dx \\
\qquad\qquad
+a_{N-1}(\alpha)\int_0^{+\infty}\left( \frac{\partial}{\partial\alpha} p_{N}(x;\alpha) \right)p_{N-1}'(x;\alpha)
w(x;\alpha)dx
\\ \qquad+a_{N-1}(\alpha)\int_0^{+\infty}p_{N}(x;\alpha)
\left( \frac{\partial}{\partial\alpha} p_{N-1}'(x;\alpha) \right)  w(x;\alpha)dx.
\end{multline}
The last term in the above equation vanishes because of the
orthogonality, and using the orthogonality also for the
other integrals we obtain
\begin{equation}
\frac{d}{d\alpha}\log \widehat
Z_{N}(\alpha)=-N a_{N-1}'(\alpha) \frac{\gamma_N(\alpha)}{\gamma_{N-1}(\alpha)}
-Na_{N-1}(\alpha)\frac{\gamma_N'(\alpha)}{\gamma_{N-1}(\alpha)}+a_{N-1}(\alpha)(J_1-J_2)
,
\end{equation}
where
\begin{align}
&J_1=\int_0^{+\infty}\left( \frac{\partial}{\partial\alpha} p_{N}(x;\alpha) \right) p_{N-1}'(x;\alpha)
w(x;\alpha)dx,\\
&J_2=\int_0^{+\infty}p_{N}'(x;\alpha)
\left( \frac{ \partial }{\partial\alpha} p_{N-1}(x;\alpha) \right) w(x;\alpha)dx.
\end{align}
From the recurrence relation (\ref{recurr}), it follows directly that $a_{n}=\frac{\gamma_{n}}{\gamma_{n+1}}$, and using this relation we
obtain
\begin{equation}\label{hat Zn2}
\frac{d}{d\alpha}\log \widehat
Z_{N}(\alpha)=-N \frac{\gamma_{N-1}'(\alpha)}{\gamma_{N-1}(\alpha)}
+\frac{\gamma_{N-1}}{\gamma_N}(\alpha)(J_1-J_2)
.
\end{equation}
We will simplify this expression further, but for that purpose we need to introduce $w(z;\alpha)=z^{\alpha}e^{-\frac{2N}{\sigma}(z^2-rz)}$ as the analytic continuation to $\mathbb C\setminus [0,+\infty)$ of the weight function $w(x;\alpha)$, in such a way that \[\lim_{\epsilon\to 0}w(x+i\epsilon;\alpha)=w(x;\alpha),\qquad \lim_{\epsilon\to 0}w(x-i\epsilon;\alpha)=e^{2\pi i\alpha}w(x;\alpha)\]
for $x\in\mathbb R$. In this way we can write integrals on $(0,+\infty)$ as integrals on a contour $\mathcal{C}$ that encircles clockwise  the line $(0,+\infty)$: we have
\[\int_0^{+\infty} P(x)w(x;\alpha)dx=\dfrac{1}{1-e^{2\pi i\alpha}}\int\limits_{\mathcal{C}}P(z)w(z;\alpha)dz\]
for any polynomial $P$. The main advantage of the contour $\mathcal C$ for us is that it allows for integration of functions with a pole at $0$.
We can simplify the formula for $J_2$ by integrating by parts and by using the fact that
\[\frac{\partial}{\partial\alpha} p_{N}(x;\alpha)=\left(\frac{\partial}{\partial\alpha} p_{N}(x;\alpha)-\frac{\partial}{\partial\alpha} p_{N}(0;\alpha)\right)+\frac{\partial}{\partial\alpha} p_{N}(0;\alpha),\]
and this leads to
\begin{equation}
\begin{split}
J_2&=\dfrac{1}{1-e^{2\pi i\alpha}} \int\limits_{\mathcal{C}}p_{N}'(z;\alpha)
\left( \frac{ \partial }{\partial\alpha} p_{N-1}(z;\alpha) \right) w(z;\alpha)dz\\
&=-\dfrac{1}{1-e^{2\pi i \alpha}}  \int\limits_{\mathcal{C}}p_{N}(z;\alpha)
\left( \frac{\partial}{\partial\alpha} p_{N-1}(z;\alpha) \right) (\frac{\alpha}{z}-4\frac{N}{\sigma} z+2\frac{N}{\sigma} r)w(z;\alpha)dz\\
&=4\frac{N}{\sigma}\frac{\gamma_{N-1}'(\alpha)}{\gamma_N(\alpha)}-\dfrac{\alpha }{1-e^{2\pi i \alpha}} \left( \frac{\partial}{\partial\alpha} p_{N-1}(0;\alpha) \right)  \int\limits_{\mathcal{C}}p_{N}(z;\alpha)
\frac{1}{z}w(z;\alpha)dz.\label{J1}
\end{split}
\end{equation}
Similarly
\begin{multline}
\label{J2}
J_1=-N\frac{\gamma_N'(\alpha)}{\gamma_{N-1}(\alpha)}-\dfrac{\alpha}{1-e^{2\pi i\alpha}} \int\limits_{\mathcal{C}}\left( \frac{\partial}{\partial\alpha} p_{N}(z;\alpha) \right)
p_{N-1}(z;\alpha)\frac{1}{z}w(z;\alpha)dz\\
-\int_0^{+\infty}\left( \frac{\partial}{\partial\alpha} p_{N}(x;\alpha) \right)
p_{N-1}(x;\alpha)(-4\frac{N}{\sigma}x+2\frac{N}{\sigma} r)w(x;\alpha)dx.
\end{multline}
The second term can be computed as in (\ref{J1}), for the last term we expand $\frac{\partial}{\partial\alpha} p_{N}(x;\alpha)$ for large $x$ using (\ref{ckdk}) as
\[\frac{\partial}{\partial\alpha} p_{N}(x;\alpha)=\gamma_N'(\alpha)x^N+(\gamma_Nc_N)'(\alpha)x^{N-1}+(\gamma_Nd_N)'(\alpha)x^{N-2}+\bigO(x^{N-3}).\]
This yields
\begin{multline}
J_1=-(N+\alpha)\frac{\gamma_N'(\alpha)}{\gamma_{N-1}(\alpha)}
- \left( \frac{\partial}{\partial \alpha} p_{N}(0;\alpha) \right)\dfrac{\alpha}{1-e^{2\pi i\alpha}} \int\limits_{\mathcal{C}}
p_{N-1}(z;\alpha)\frac{1}{z}w(z;\alpha)dz \\
+2\frac{N}{\sigma}  \frac{2(\gamma_Nd_N)'(\alpha)-r(\gamma_Nc_{N})'(\alpha)}{\gamma_{N-1}(\alpha)}+4\frac{N}{\sigma}\gamma_N'(\alpha)\int_0^{+\infty}
p_{N-1}(x;\alpha)x^{N+1}w(x;\alpha)dx
\\ + 2\frac{N}{\sigma}\left(2(\gamma_Nc_N)'(\alpha)-r\gamma_N'(\alpha)\right)\int_0^{+\infty}
p_{N-1}(x;\alpha)x^{N}w(x;\alpha)dx.
\end{multline}
Note that the above formulas are valid for any $\alpha\in[-1/2,1/2]$ including $\alpha=0$.
The constants $c_N$ and $d_N$ are defined by (\ref{ckdk}) and determine, together with $\gamma_N$, the sub-leading coefficients of $p_N$.

\subsection{Differential identity in terms of RH solution $Y$}

If we write
\begin{equation}\label{Y}
Y(z;\alpha)=Y^{(N)}(z;\alpha,r,\sigma)=\begin{pmatrix} \gamma_N^{-1}p_N(z)&
\frac{1}{2\pi i}\gamma_N^{-1}\int_0^{+\infty}p_N(s)\frac{w(s)ds}{s-z}\\
-2\pi i\gamma_{N-1}p_{N-1}(z)&
-\gamma_{N-1}\int_0^{+\infty}p_{N-1}(s)\frac{w(s)ds}{s-z}
\end{pmatrix},
\end{equation}
$Y$ satisfies the RH problem \cite{FIK}
\subsubsection*{RH problem for $Y$}
\begin{itemize}
\item[(a)] $Y$ is analytic in $\mathbb C\setminus \mathbb [0,+\infty)$,
\item[(b)] $Y$ has boundary values $Y_\pm$ for $x\in(0,+\infty)$, and they are related by the jump property
\begin{equation}\label{jump Y}
Y_+(x)=Y_-(x)\begin{pmatrix}1&w(x;\alpha)\\0&1\end{pmatrix},\qquad x\in (0,+\infty).
\end{equation}
\item[(c)] As $z\to\infty$, $Y$ has an asymptotic expansion of the form
\begin{equation}
\label{Yinfty}Y(z)z^{-n\sigma_3}=I+Y_1z^{-1}+Y_2z^{-2}+Y_3z^{-3}+\bigO(z^{-4}),
\end{equation}
where $Y_1,Y_2, Y_3$ are matrices that may depend on $\alpha, N,r,\sigma$ but not on $z$.
\end{itemize}
It is straightforward to derive from (\ref{Y}) the identities
\begin{align}
&c_N=Y_{1,11},&& d_N=Y_{2,11},\\
&\gamma_{N-1}=\left(\frac{Y_{1,21}}{-2\pi i}\right)^{1/2},&& \gamma_N=\left(-\frac{1}{2\pi iY_{1,12}}\right)^{1/2},\\
&\int_{0}^{+\infty}p_{N-1}(x)x^{N}w(x;\alpha)dx=\frac{Y_{1,22}}{\gamma_{N-1}}, &&\int_{0}^{+\infty}p_{N-1}(x)x^{N+1}w(x;\alpha)dx=\frac{Y_{2,22}}{\gamma_{N-1}}.
\end{align}
We can use those identities to express $J_1$, $J_2$, and $\frac{d}{d\alpha}\log \widehat Z_{N}(\alpha)$ in terms of $Y=Y^{(N)}$.
Indeed, for $J_2$, we obtain
\begin{equation}\label{J2-2}
J_2=2\frac{N}{\sigma}Y_{1,21}'\sqrt{\frac{Y_{1,12}}{Y_{1,21}}}- \left(\frac{Y_{21}(0;\alpha)}{-i\sqrt{2\pi i Y_{1,21}}}\right)'  \dfrac{\alpha }{1-e^{2\pi i \alpha}}\int\limits_{\mathcal{C}}p_{N}(z;\alpha)
\frac{1}{z}w(z;\alpha)dx,
\end{equation}
and for $J_1$,
\begin{multline}\label{J1-2}
J_1=\frac{N+\alpha}{2}\frac{Y_{1,12}'}{Y_{1,12}}\frac{1}{\sqrt{Y_{1,12}Y_{1,21}}}\\
+\frac{2N}{\sigma\sqrt{Y_{1,12}Y_{1,21}}}\frac{ Y_{1,12}'}{Y_{1,12}}\left(-\Tr Y_{2}+\frac{r}{2}\Tr Y_1-Y_{1,11}Y_{1,22}\right)\\
+\frac{2N}{\sigma\sqrt{Y_{1,12}Y_{1,21}}}\left(2Y_{2,11}'-rY_{1,11}'+2Y_{1,11}'Y_{1,22}\right)\\
-  \left(\frac{Y_{11}(0;\alpha)}{\sqrt{-2\pi iY_{1,12}}} \right)'\dfrac{\alpha}{1-e^{2\pi i\alpha}} \int\limits_{\mathcal{C}}
p_{N-1}(z;\alpha)\frac{1}{z}w(z;\alpha)dz
.\end{multline}
In the above formulas, the primes denote derivatives with respect to $\alpha$.
The jump condition for $Y$ implies that $\det Y$ is an entire function, and the asymptotics for $Y$ together with Liouvilles theorem yield $\det Y\equiv 1$. In view of (\ref{Yinfty}), this can only be true if
$\Tr Y_1=0$.
Combining  (\ref{hat Zn2}),  (\ref{J1-2}) and  (\ref{J2-2}) one arrives at the formula
\begin{equation}\label{hat Zn22}
\begin{split}
&\frac{d}{d\alpha}\log \widehat Z_{N}(\alpha)=
\frac{\alpha}{2}\frac{ Y_{1,12}'}{Y_{1,12}}+2\frac{N}{\sigma}\left((\det Y_1)'+2Y_{2,11}'-rY_{1,11}'\right)\\
&+\frac{N}{2}\left(\frac{Y_{1,12}'}{Y_{1,12}}-\frac{Y_{1,21}'}{Y_{1,21}}\right)\\
&+\sqrt{Y_{1,12}Y_{1,21}} \left(\frac{Y_{21}(0;\alpha)}{-i\sqrt{2\pi i Y_{1,21}}}\right)'  \dfrac{\alpha }{1-e^{2\pi i \alpha}}\int\limits_{\mathcal{C}}p_{N}(z;\alpha)\frac{1}{z}w(z;\alpha)dz\\
&-  \sqrt{Y_{1,12}Y_{1,21}}\left(\frac{Y_{11}(0;\alpha)}{\sqrt{-2\pi iY_{1,12}}} \right)'\dfrac{\alpha}{1-e^{2\pi i\alpha}} \int\limits_{\mathcal{C}}
p_{N-1}(z;\alpha)\frac{1}{z}w(z;\alpha)dz.
  \end{split}
\end{equation}

Regarding the last two integrals in the above formulas,  we have
\begin{eqnarray}
Y_{22}(z) = \frac{- 2 \pi i \gamma_{N-1}}{1 - e^{ 2 \pi i \alpha}}  p_{N-1}(z;\alpha) w(z;\alpha) -  \frac{ \gamma_{N-1}}{1 - e^{ 2 \pi i \alpha}} \int\limits_{\mathcal{C}} \frac{p_{N-1}(z';\alpha) w(z')}{z' - z} dz',
\end{eqnarray}
if $z$ lies inside the contour $\mathcal C$.
If we set
\begin{equation}
\label{def tilde Y}
\tilde Y(z)=Y(z)\begin{pmatrix}1&-\frac{w(z;\alpha)}{1-e^{2\pi i\alpha}}\\0&1\end{pmatrix},
\end{equation}
it follows that
\begin{equation}
\label{J1a}
\begin{split}
\frac{\alpha}{1 - e^{ 2 \pi i \alpha}}  &\int\limits_{\mathcal{C}} p_{N-1}(z';\alpha) \frac{1}{z'} w(z';\alpha) dz' \\
&=-\dfrac{\alpha}{\gamma_{N-1}}\lim_{z\rightarrow 0}\left( Y_{22}(z;\alpha)-\dfrac{Y_{21}(z;\alpha)w(z;\alpha)}{1-e^{2\pi i \alpha}}\right)
=-\dfrac{\alpha\sqrt{2\pi}}{\sqrt{iY_{1,21}}}\tilde Y_{22}(0;\alpha).
\end{split}
\end{equation}
Now, for the other integral in (\ref{hat Zn22}), the relevant identity is
\begin{equation}
\gamma_N Y_{12}(z)=\dfrac{1}{2\pi i (1 - e^{ 2 \pi i \alpha})} \int\limits_{\mathcal{C}}\frac{p_{N}(z';\alpha) w(z';\alpha) }{z'-z}dz'+\dfrac{1}{1 - e^{ 2 \pi i \alpha}} p_{N}(z;\alpha) w(z;\alpha),
\end{equation}
from which it follows that
\begin{multline}
\label{J2a}
\frac{ \alpha }{1 - e^{ 2 \pi i \alpha}} \int_{\mathcal C} \frac{p_{N}(z';\alpha) w(z';\alpha)}{z'} dz' =
2\pi i\alpha\gamma_N\lim_{z\rightarrow 0 }\left(Y_{12}^{(N)}(z;\alpha)-\dfrac{Y_{11}^{(N)}(z;\alpha)w(z;\alpha)}{1-e^{2\pi i \alpha}}\right)
\\=\frac{\sqrt{2\pi}\, i\alpha}{\sqrt{-iY_{1,12}}}\tilde Y_{12}(0;\alpha)
.
\end{multline}
Substituting those identities in (\ref{hat Zn22}), we obtain
\begin{equation}\label{G1b}\frac{d}{d\alpha}\log \widehat
Z_{N}(\alpha)=\mathcal G_{N}(\alpha),\end{equation}
with
\begin{multline}\label{GN}
\mathcal G_{N}(\alpha)=\mathcal G_N(\alpha;r,\sigma)=\frac{\alpha}{2}\frac{ Y_{1,12}'}{Y_{1,12}}+2\frac{N}{\sigma}\left((\det Y_1)'+2Y_{2,11}'-rY_{1,11}'\right)\\
+\frac{N}{2}\left(\frac{Y_{1,12}'}{Y_{1,12}}-\frac{Y_{1,21}'}{Y_{1,21}}\right) +\sqrt{Y_{1,12}Y_{1,21}} \left(\frac{Y_{11}(0;\alpha)}{\sqrt{- iY_{1,12}}} \right)'\dfrac{\alpha}{\sqrt{iY_{1,21}}}\tilde Y_{22}(0;\alpha)\\
-\sqrt{Y_{1,12}Y_{1,21}} \left(\frac{Y_{21}(0;\alpha)}{\sqrt{ i Y_{1,21}}}\right)'\frac{\alpha}{\sqrt{-iY_{1,12}}}\tilde Y_{12}(0;\alpha).
  \end{multline}
  The identity (\ref{Zalpha-intro}) stated in the introduction follows from (\ref{G1b}).
It then follows from (\ref{partition rec2})-(\ref{partition rec3})
that
\begin{align}&\label{partition rec4}
\log Z_{2N}(r,s)=\log(2N)!+2\log\widehat Z_{N}(0;r,s)\\
&\qquad\qquad\qquad\qquad +\int_0^{-1/2}\mathcal G_{N}(\alpha;r,s)d\alpha+\int_0^{1/2}\mathcal G_{N}(\alpha;r,s)d\alpha,\\
& \log Z_{2N+1}(r,s)=\log(2N+1)!+\log\widehat
Z_{N+1}(0;r,s_+)+\log\widehat Z_{N}(0;r,s_-)\nonumber\\&\qquad\qquad\qquad\qquad+\int_0^{-1/2}\mathcal G_{N+1}(\alpha;r,s_+)d\alpha
+\int_0^{1/2}\mathcal G_{N}(\alpha;r,s_-)d\alpha.\label{partition rec5}
\end{align}
The equilibrium measure $\nu_{r,\sigma}$ is one-cut, and we can apply the
Deift-Zhou steepest descent method on the RH problem for $Y$ to
obtain large $N$ asymptotics for $Y,Y_1,Y_2$ and derivatives, uniformly for
$\alpha\in[-1/2,1/2]$. This means that we can, in principle, obtain
asymptotics for each of the integrals in the above formulas. The
only remaining problem is then to find asymptotics for $\widehat
Z_{N}(0;r,\sigma)$.

\subsection{Asymptotics for $\widehat
Z_{N}(0;r,\sigma)$}
By a shift of the integration variable, we have
\begin{eqnarray}
\widehat Z_{N}(0;r,\sigma)&=&\frac{1}{N!}\int_{\mathbb R_+^N}\prod_{i<j}(x_i-x_j)^2
\prod_{j=1}^N
e^{-\frac{2N}{\sigma}(x_j^2-rx_j)}dx_j
\\
&=&\frac{e^{\frac{N^2r^2}{2\sigma}}}{N!}\int_{\mathbb R_+^N}\prod_{i<j}(x_i-x_j)^2
\prod_{j=1}^N
e^{-2\frac{N}{\sigma}(x_j-\frac{r}{2})^2}dx_j\\
&=&\frac{e^{\frac{N^2r^2}{2\sigma}}}{N!}\int_{(-\frac{r}{2},+\infty)^N}\prod_{i<j}(x_i-x_j)^2
\prod_{j=1}^N
e^{-2\frac{N}{\sigma}x_j^2}dx_j,\end{eqnarray}
which is, up to a pre-facor, the partition function for a Gaussian Unitary Ensemble (GUE) with a cut-off at $-r/2$.
We will show that \[\ Z_{N,\sigma}^{(r)}:=\int_{(-\frac{r}{2},+\infty)^N}\prod_{i<j}(x_i-x_j)^2
\prod_{j=1}^N
e^{-\frac{2N}{\sigma}x_j^2}dx_j,\qquad \sigma< \frac{r^2}{4}\] is, as $N\to\infty$,  exponentially close to the GUE partition function without a cut-off, given by
\begin{equation}
Z_{N,\sigma}^{GUE}:=Z_{N,\sigma}^{(-\infty)}=\int_{\mathbb R^N}\prod_{i<j}(x_i-x_j)^2
\prod_{j=1}^N e^{-2\frac{N}{\sigma} x_j^2}dx_j.
\end{equation}
In the GUE (without cut-off) with joint probability distribution of eigenvalues given by
\[\frac{1}{Z_{N,\sigma}^{GUE}}\prod_{i<j}(x_i-x_j)\prod_{j=1}^Ne^{-\frac{2N}{\sigma}x_j^2}dx_j,\]
the ratio $Z_{N,\sigma}^{(r)}/Z_{N,\sigma}^{GUE}$ is equal to the probability that all eigenvalues are bigger than $-r/2$. It is well-known that this probability is exponentially close to $1$ since $-r/2$ is strictly smaller than the left endpoint of the support of the equilibrium measure (which is the limiting mean eigenvalue distribution) corresponding to this ensemble. Indeed, it follows for example by a large deviation principle, see \cite[Section 2.6.2]{AGZ}, that
\[\frac{Z_{N,\sigma}^{(r)}}{Z_{N,\sigma}^{GUE}}=1+\bigO(e^{-cN}),\qquad N\to\infty,\ c>0.\]

We can conclude that, as $N\to\infty$,
\begin{align}&\label{partition rec6}
\log Z_{2N}(r,s)=\log(2N)!+2\log \left[\frac{1}{N!}Z_{N,s}^{GUE}\right]\nonumber\\
&\qquad+\frac{N^2r^2}{s}+\int_0^{-1/2}\mathcal G_{N,}(\alpha;r,s)d\alpha+\int_0^{1/2}\mathcal G_{N}(\alpha;r,s)d\alpha\ +\mathcal{O} \left(e^{ - c N} \right),\\
& \log Z_{2N+1}(r,s)=\log(2N+1)! +\log \left[\frac{1}{(N+1)!}Z_{N+1,s_+}^{GUE}\right]+\log\left[\frac{1}{N!} Z_{N,s_-}^{GUE}\right]\nonumber\\&\qquad
+\frac{(2N+1)^2r^2}{4s}
+\int_0^{-1/2}\mathcal G_{N+1}(\alpha;r,s_+)d\alpha\nonumber\\&\qquad \qquad\qquad
+\int_0^{1/2}\mathcal G_{N}(\alpha;r,s_-)d\alpha\ +\mathcal{O} \left(e^{ - c N} \right),\label{partition rec7}
\end{align}
where we used the identity \[\frac{(N+1)^2 r^2}{2 s_+}+\frac{N^2 r^2}{2 s_-}=\frac{(2N+1)^2r^2}{4s}.\]

\section{Analysis of the RH problem for the orthogonal polynomials\label{SecAsym}}
In this section, we will obtain asymptotics for the RH problem for $Y$, stated in the previous section, as $N\to\infty$, and this will also lead us towards asymptotics for $\mathcal G_N$.
We recall that $Y$ depends on $N$, $\alpha$, $r$, and $\sigma$.
By (\ref{partition rec6})-(\ref{partition rec7}), we will need asymptotics as $N\to\infty$ for $-\frac{1}{2}\leq \alpha\leq \frac{1}{2}$, and with $\sigma=s$ and $\sigma=s_\pm$.
In the case where $\sigma=s_\pm$, $\sigma$ will be $N$-dependent by (\ref{spm}), but this will not cause any problems because the asymptotics for $Y$ will be uniform in $\sigma$ as long as $\sigma<\frac{r^2}{4}-\epsilon$, $\epsilon>0$. Since we assume that $s<\frac{r^2}{4}$, we also have that $s_\pm <\frac{r^2}{4}-\epsilon$ for $N$ sufficiently large and $\epsilon$ sufficiently small.

\subsection{The equilibrium measure}
Define the equilibrium measure $\nu_{r,\sigma}$ minimizing the logarithmic energy (\ref{energy2})
among all probability measures on $\mathbb R^+$.
As mentioned in Example \ref{exampel: 1}, one can compute $\nu_{r,\sigma}$ explicitly: it is supported on the interval $[a,b]=[\frac{r}{2}-\sqrt{\sigma}, \frac{r}{2}+\sqrt{\sigma}]$ and the density is given by (\ref{nu}).
For all $\sigma<\frac{r^2}{4}$, we have that $a>0$.
Now, define
\begin{eqnarray}
g(z) = \int_{a}^{b} \log{(z - s)} d \nu_{r,\sigma}(s),
\end{eqnarray}
where we choose the branch of the logarithm such that $g$ is analytic in $\mathbb C\setminus(-\infty,b]$.
It is directly verified that
\begin{equation}\label{gjump1}
g_{+}(x)-g_{-}(x)=2\pi i\int_x^{b}d\nu_{r,\sigma}(y),
\qquad \mbox{ for $x<b$,}
\end{equation}
where $\int_x^{b}d\nu_{r,\sigma}(y)$ is understood as $1$ for $x<a$.
For $x\in[a,b]$, the Euler-Lagrange variational conditions conditions for $\nu_{r,\sigma}$ imply that
\begin{equation}\label{gjump2}
g_{+}(x) + g_{-}(x) - W(x) - \ell = 0,\qquad W(x)=\frac{2}{\sigma}(x^2-rx),
\end{equation}
for some number $\ell$ depending on $r,\sigma$, and that the left hand side of (\ref{gjump2}) is strictly negative on $(0,+\infty)\setminus[a,b]$.

For later use, we compute the asymptotics for the $g$-function as $z\to\infty$, which are given by
\begin{eqnarray}
g(z) &=& \log{z} - \frac{1}{z} \int y d\nu_{r,\sigma}(y) - \frac{1}{2 z^{2}} \int y^{2} d\nu_{r,\sigma}(y) + \bigO(z^{-3}) \\
&=& \log{z} - \frac{m_{1}}{z} - \frac{m_{2}}{2 z^{2}} + \bigO(z^{-3}),
\end{eqnarray}
where we write $m_k=\int y^kd\nu_{r,\sigma}(y)$.
It follows that
\begin{eqnarray}\label{egasymp}
e^{ N g(z) \sigma_{3}} = z^{ N \sigma_{3}} \left( I + \frac{G_{1}}{z} + \frac{G_{2}}{z^{2}} + \bigO(z^{-3}) \right) \ ,\qquad\mbox{ as $z\to\infty$,}
\end{eqnarray}
with
\begin{eqnarray}\label{G1G2}
G_{1} =\begin{pmatrix}- N m_{1}&0\\0&N m_{1}\end{pmatrix}, \ \ \ \ \ \ G_{2} = \begin{pmatrix}\frac{N^{2} m_{1}^2 - N m_{2}}{2} &0\\0&\frac{N^{2} m_{1}^2 + N m_{2}}{2}\end{pmatrix}.
\end{eqnarray}
By (\ref{nu}), $m_1$ and $m_2$ can be computed explicitly, we have
\begin{equation}\label{m12}
m_1=\frac{a+b}{2}=\frac{r}{2},\qquad m_2=\frac{1}{16}(5a^2+6ab+5b^2)=\frac{1}{4}(r^2+\sigma).
\end{equation}

\subsection{Transformation $Y\mapsto T$}
We use the $g$-function in the first transformation of the RH problem and define $T$ by
\begin{eqnarray}\label{def T}
T(z) = e^{ - N \ell \sigma_{3} / 2 } Y(z) e^{ - N ( g(z) - \ell/2) \sigma_{3}} \ ,
\end{eqnarray}
where $Y$ is given by (\ref{Y}), with $w(x)=x^\alpha e^{-NW(x)}$.
From the RH conditions for $Y$, it is straightforward to obtain a RH problem for the new unknown $T$: we have
\subsubsection*{RH problem for $T$}
\begin{itemize}
\item[(a)] $T$ is analytic in $\mathbb C\setminus\mathbb [0,+\infty)$.
\item[(b)] For $x\in\mathbb R$, we have the jump relation
\begin{equation}
T_{+}(x) = T_{-}(x)J_T(x),\qquad  x \in [a,b] \ ,
\end{equation}
with
\begin{equation}
J_T(x)=\begin{pmatrix}e^{-N(g_+(x)-g_-(x))}&x^\alpha e^{N(g_+(x)+g_-(x) -W(x) -\ell)}\\0&e^{N(g_+(x)-g_-(x))}\end{pmatrix}.
\end{equation}
\item[(c)] As $z\to\infty$, $T(z)=I+\bigO(z^{-1})$.
\end{itemize}
Condition (c) follows from the logarithmic behavior of $g$ at infinity, condition (b) is a consequence of (\ref{def T}) and the jump relation (\ref{jump Y}) for $Y$.
It is useful to define $\phi(z)$ as follows:
\begin{eqnarray}
\label{eq:1.09}
\phi(z) = 2 g(z)- W(z)- \ell \ , \qquad z \in \mathbb{C} \setminus (-\infty,b] \ .
\end{eqnarray}
From (\ref{gjump2}), it follows that
\begin{align}
&g_{+}(x) - g_{-}(x) = \phi_{+}(x) \ , & x \in [a,b] \ , \\
&g_{+}(x) - g_{-}(x) = - \phi_{-}(x) \ ,& x \in [a,b] \ .
\end{align}
It turns out that $\phi$ can be evaluated explicitly: we have
\begin{multline}\label{phi}
\phi(z) = \frac{-4}{\sigma} \left(-\frac{(b-a)^{2}}{4} \log{\left[ \sqrt{z-a} + \sqrt{z-b} \right]} \right.\\
\left. + \frac{2z - ( a + b)}{4} \sqrt{z-b}\sqrt{z-a} + \frac{(b-a)^2 \log{(b-a)}}{8}\right) \ .
\end{multline}
Using (\ref{gjump1}) and (\ref{gjump2}), we can write the jump matrix $J_T$ in terms of $\phi$:
\begin{align}
&J_T(x)=\begin{pmatrix}e^{-N\phi_+(x)}&x^\alpha \\0&e^{-N\phi_-(x)}\end{pmatrix},&\mbox{ for $x\in[a,b]$,}\\
&J_T(x)=\begin{pmatrix}1&x^\alpha e^{N\phi(x)}\\0&1\end{pmatrix},&\mbox{ for $x\in(0,a)\cup (b,+\infty)$.}
\end{align}
Note that, although $\phi$ is not continuous on $(0,a)$, $e^{N\phi}$ is, and therefore we can safely write $e^{N\phi}$ instead of $e^{N\phi_\pm}$.

\subsection{Transformation $T\mapsto S$}
The jump matrix for $T$ can be re-written for $x\in(a,b)$ as
\begin{eqnarray}
J_T(x) &=&
\begin{pmatrix}1&0\\ \frac{e^{- N \phi_{-}(x)}}{x^{\alpha}}
&1 \end{pmatrix}
\begin{pmatrix}0&{x^{\alpha}}\\
{-\frac{1}{x^{\alpha}}}&
{0} \end{pmatrix} \begin{pmatrix}{1}&{0}\\{\frac{e^{- N \phi_{+}(x)}}
{x^{\alpha}}
}&{1} \end{pmatrix}\ .
\end{eqnarray}
This leads us to the second transformation, where we open lenses.  Define three regions as follows, see Figure \ref{Jump1}:
\begin{itemize}
\item Region I:  the upper lens region.
\item Region II:  the lower lens region.
\item Region III:  everything else.
\end{itemize}
\begin{figure}[htb!]
    \includegraphics[width=1.2\textwidth]{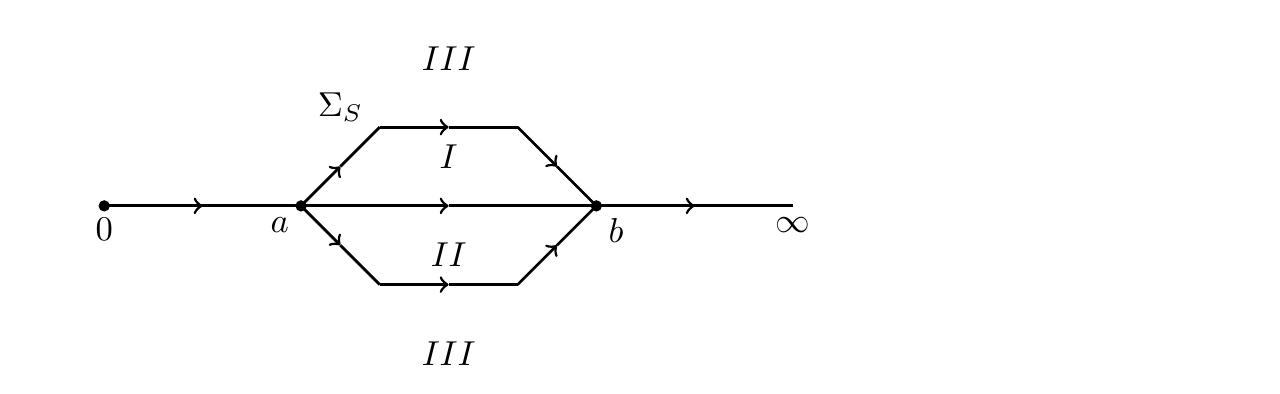}
 \caption{The contour $\Sigma_S$ consists of the two lenses and the interval $(0,\infty)$.}
 \label{Jump1}
\end{figure}
We then define $S$ in the three regions as follows:
\begin{eqnarray} \label{def S}
S(z) = T (z) \times\
\begin{cases}
I & \ z \in \mbox{Region III}, \\
\begin{pmatrix}1&0\\-\frac{e^{- N \phi_{+}(z)}}
{z^{\alpha}}
&1\end{pmatrix} & \ z \in \mbox{Region I}, \\
\begin{pmatrix}1&0\\
\frac{e^{- N \phi_{-}(z)}}
{z^{\alpha}}&1\end{pmatrix} & \ z \in \mbox{Region II}. \\
\end{cases}
\end{eqnarray}
Using the RH conditions satisfied by $T$, we can deduce the RH problem for $S$.
\subsubsection*{RH problem for $S$}
\begin{itemize}
\item[(a)] $S$ is analytic in $\mathbb C\setminus \Sigma_S$, where $\Sigma_S$ is the contour shown in Figure \ref{Jump1}.
\item[(b)] The jump relations for $S$ are
\begin{equation}
S_+(z)=S_-(z)J_S(z),\qquad z\in\Sigma_S,
\end{equation}
where
\begin{align}
&J_S(z)=
\begin{pmatrix}1&0\\
\frac{e^{- N \phi(z)}} {z^{\alpha}}
&1\end{pmatrix}, \ \ \  &z \in \mbox{ lower lens boundary,}
\\
&
J_S(z)=\begin{pmatrix}0&z^{\alpha}\\
-\frac{1}{z^{\alpha}}&
0\end{pmatrix}, &z \in (a,b),\\
& J_S(z)=\begin{pmatrix}1&0\\\frac{e^{- N \phi(z)}}
{z^{\alpha}}
&1\end{pmatrix}, &z \in \mbox{ upper lens boundary},\\
&J_S(z)=J_T(z), &z\in (0,a)\cup(b,+\infty).
\end{align}
\item[(c)] $S(z)=I+\bigO(z^{-1})$ as $z\to\infty$.
\end{itemize}
Using the explicit formula for $\phi$ given in (\ref{phi}), one can verify that the jump matrix for $S$ is exponentially small as $N\to\infty$ on $\Sigma_S\setminus[a,b]$, uniformly for $z$ bounded away from the endpoints $a$ and $b$.
The next critical step is to build approximations to $S$ in 4 regions: a neighborhood of $z=b$, a neighborhood of $z=a$, a neighborhood of $z=0$, and the complement of those three regions. We will call the parametrix in the latter region ``the outer parametrix'' and this one will take care of the jump on $(a,b)$. The local parametrices near $a$ and $b$ will deal with the jumps near $a$ and $b$, where the uniform convergence breaks down. The local parametrix near $0$ finally is needed for technical reasons because $0$ is an endpoint of the jump contour $\Sigma_S$.

\subsection{The outer Parametrix $P^{\infty}(z)$}

We seek for a function $P^{\infty}(z)$ which satisfies the same RH conditions as $S$, but where we ignore the jumps that are small in the large $n$ limit and the jumps in small neighborhoods of $a$ and $b$.
We thus obtain
\subsubsection*{RH problem for $P^{\infty}(z)$}
\begin{itemize}
\item[(a)] $P^{\infty}$ is analytic in $\mathbb C\setminus[a,b]$.
\item[(b)] For $z\in(a,b)$, we have
\begin{equation}
P^{\infty}_+(z)=P^{\infty}_-(z)\begin{pmatrix}0&z^{\alpha}\\
-\frac{1}{z^{\alpha}}&
0\end{pmatrix}.
\end{equation}
\item[(c)] $P^{\infty}(z)=I+\bigO(z^{-1})$ as $z\to\infty$.
\end{itemize}
This RH problem is solved by a function of the form
\begin{eqnarray}
\label{Pinfinity}
P^{\infty}(z) = d^{ \sigma_{3}} P_{0}^{\infty}(z)  D(z)^{ - \sigma_{3}} \ ,
\end{eqnarray}
where
\begin{eqnarray}
P_{0}^{\infty}(z) = \begin{pmatrix}\frac{\gamma(z) + \gamma(z)^{-1}}{2}&
{\frac{\gamma(z) - \gamma(z)^{-1}}{2i}}\\
{\frac{\gamma(z) - \gamma(z)^{-1}}{-2i}}&
{\frac{\gamma(z) + \gamma(z)^{-1}}{2}}\end{pmatrix} \ \ ,
\end{eqnarray}
and
\begin{eqnarray}
\gamma(z)  = \frac{ (z-b)^{1/4}}{(z-a)^{1/4}} \ .
\end{eqnarray}
The quantity $D(z)$ is given by
\begin{eqnarray}
D(z)= \exp{\left\{ \frac{ \alpha\sqrt{ z - b} \sqrt{ z - a} }{2 \pi i} \int_{a}^{b} \frac{ \log{x}}{\sqrt{ x - b} \sqrt{ x - a } } \frac{dx}{x -z} \right\}} ,
\end{eqnarray}
and can be computed explicitly as follows:
\begin{equation}\label{D}
D(z) =  \left(
\frac{\sqrt{b(z-a)} + \sqrt{a (z-b)}}{\sqrt{z-a} + \sqrt{z-b}}
\right)^{\alpha}
.
\end{equation}
The constant $d$ is the large $z$ limit of $D(z)$:
\begin{equation}
\label{eq:2.06}
d
= \left( \frac{\sqrt{b} + \sqrt{a}}{2} \right)^{\alpha} .
\end{equation}

It will be useful to have the expansion of $P^{\infty}$ for $z \to \infty$.  For that, we will also need the asymptotic behavior of $D(z)$ and $P^{\infty}(z)$ as $z\to\infty$: we have
\begin{eqnarray}\label{asympD}
D(z) = d\left( 1 + \frac{d_{1}}{z} + \frac{d_{2}}{z^{2}} + \bigO(z^{-3}) \right),
\end{eqnarray}
and
\begin{equation}
P_{0}^{\infty}(z) = I + \frac{1}{z} \begin{pmatrix}0&\frac{ i ( b - a)}{4}\\
{\frac{-i(b-a)}{4}}& {0}\end{pmatrix} + \frac{1}{z^{2}}
\begin{pmatrix}{\frac{(b-a)^{2}}{32}}&{\frac{i ( b^{2} - a^{2} )}{8}}\\
{\frac{ - i (b^{2} - a^{2} )}{8}}&{\frac{(b-a)^{2}}{32}}\end{pmatrix}
+ \bigO(z^{-3}).
\end{equation}

Hence, after straightforward calculations, we obtain
\begin{equation}\label{asympPinf}
P^{\infty}(z)=I+\frac{P_1(\alpha)}{z}+\frac{P_2(\alpha)}{z^2}+\bigO(z^{-3}),\qquad z\to\infty,
\end{equation}
with
\begin{align}
&\label{P1}P_1(\alpha)=\begin{pmatrix}{-d_{1}}&{\frac{id^2(b-a)}{4}}\\{\frac{-i(b-a)}{4d^2}}&{d_{1}}\end{pmatrix} ,\\
&\label{P2}P_2(\alpha)=\begin{pmatrix}{\frac{(b-a)^{2}}{32}+d_1^2-d_2}&{\frac{id^2 ( b-a) ( a + b + 2 d_{1})}{8}}\\
{\frac{- i ( b-a) ( a + b - 2 d_{1})}{8d^2}}&{\frac{(b-a)^{2}}{32} +  d_2}\end{pmatrix},
\end{align}
where $d$ is defined in (\ref{eq:2.06}) and $d_{1}$ and $d_{2}$ are defined in (\ref{asympD}).
We have
\begin{align}
&\label{d1}
d_1(\alpha)=\frac{\alpha}{4}(\sqrt{b}-\sqrt{a})^2,\\
&\label{d2}d_2(\alpha)=\frac{d_1(\alpha)}{8}\left(3a+2\sqrt{ab}+3b+\alpha(\sqrt{b}-\sqrt{a})^2\right).
\end{align}
Note that the outer parametrix $P^\infty$ depends on $\sigma$ through $a=a_{\sigma}$ and $b=b_{\sigma}$, and so do $D$, the constants $d,d_1,d_2$, and the matrices $P_1$ and $P_2$.

\subsection{The Airy Parametrices}
\label{Airy}
Recall that $a$ and $b$ depend on $r$ and $\sigma$; if $\sigma=s_\pm$ with $s_\pm$ given by (\ref{spm}), $a$ and $b$ will thus depend on $N$. However, for $N$ large, $\sigma$ will be close to $s$, and $a,b$ will be close to $\frac{1}{2}(r\mp 2\sqrt{s})$.
We will construct local parametrices in fixed sufficiently small disks $B_\pm:=B_\delta(\frac{1}{2}(r\pm 2\sqrt{s}))$ centered at $\frac{1}{2}(r\mp 2\sqrt{s})$. For $N$ sufficiently large, those neighborhoods will also contain the points $a$ and $b$. The parametrices will be built out of a well-known Airy model RH problem, which we recall here.
Denote
\begin{eqnarray}
\label{K5.7} \omega := e^{\frac{2 \pi i}{3}},
\end{eqnarray}
and let the complex $\zeta$ plane be divided into $4$ regions as shown in Figure \ref{figure: Airy1}.
\begin{figure}[htb!]\begin{center}
\input{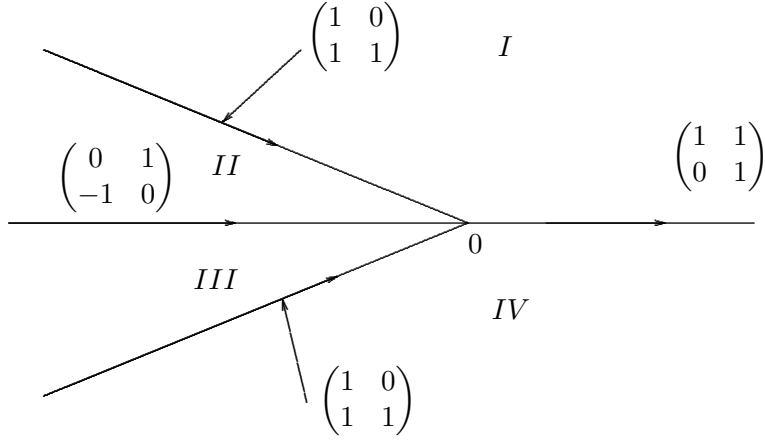}
\caption{\label{figure: Airy1}The jump contour $\Gamma$ for $\Psi$.}\end{center}
\end{figure}
Define
\begin{eqnarray}
\label{K5.8} \Psi: {\mathbb C} \setminus \Gamma
\rightarrow {\mathbb C}^{2 \times 2},
\end{eqnarray}
\begin{eqnarray}
\label{K5.9} & &
\Psi(\zeta) = \left\{
\begin{array}{ll}
\begin{pmatrix} \Ai(\zeta)&\Ai(\omega^2
\zeta)\\\Ai'(\zeta)&\omega^2 \Ai'(\omega^2 \zeta)
\end{pmatrix}
e^{- \frac{ \pi i}{6} \sigma_3}, &  \mbox{ for } \zeta \in I,
\\
\begin{pmatrix}
\Ai(\zeta)&  \Ai(\omega^2 \zeta)\\
\Ai'(\zeta)& \omega^2 \Ai'(\omega^2 \zeta)
\end{pmatrix}
e^{-\frac{\pi i}{6} \sigma_3} \begin{pmatrix} 1&0\\-1&1
\end{pmatrix}
,&  \mbox{ for } \zeta \in II,
\\
\begin{pmatrix}
\Ai(\zeta)&-\omega^2 \Ai(\omega \zeta)\\
\Ai'(\zeta)&-\Ai'(\omega \zeta)
\end{pmatrix}
e^{-\frac{\pi i}{6} \sigma_3} \begin{pmatrix} 1&0\\1&1
\end{pmatrix}
,&  \mbox{ for } \zeta \in III,
\\
\begin{pmatrix}
\Ai(\zeta)& -\omega^2 \Ai(\omega \zeta)\\
\Ai'(\zeta)& -\Ai'(\omega \zeta)
\end{pmatrix}
e^{-\frac{\pi i}{6} \sigma_3}, &  \mbox{ for } \zeta \in IV.
\end{array}
\right.
\end{eqnarray}

Then $\Psi$ satisfies the following model RH problem.
\subsubsection*{RH problem for $\Psi$}
\begin{itemize}
\item[(a)] $\Psi$ is analytic in the complex plane, except on the jump contour $\Gamma$ shown in Figure \ref{figure: Airy1}.
\item[(b)] For $\zeta$ on one of the four rays in the jump contour, we have $\Psi_+(\zeta)=\Psi_-(\zeta)J(\zeta)$, with $J(\zeta)$ the piecewise constant matrix indicated in Figure \ref{figure: Airy1}.
\item[(b)] As $\zeta\to\infty$, we have
\begin{eqnarray}
\Psi(\zeta) e^{ \frac{2}{3}\zeta^{3/2} \sigma_{3}} = \frac{ e^{ \frac{\pi i}{12}}}{ 2 \sqrt{\pi}}
\begin{pmatrix}{\zeta^{-1/4}}&{0}\\{0}&{\zeta^{1/4}}\end{pmatrix} \sum_{k=0}^{\infty} \begin{pmatrix}{(-1)^{k}s_{k}}&{s_{k}}\\{-(-1)^{k}r_{k}}&{r_{k}}\end{pmatrix} \ e^{ - \frac{\pi i}{4} \sigma_{3}} \left( \frac{2}{3} \zeta^{3/2} \right)^{-k} \ ,\ \
\end{eqnarray}
where
\begin{eqnarray}
\label{sr}
s_{k} = \frac{ \Gamma(3k + 1/2)}{54^{k}k! \Gamma(k+1/2)}, \ \ \ \ r_{k} = - \frac{ 6 k + 1 }{ 6 k - 1 } s_{k} \ \ \ \ \mbox{ for } k \ge 1 \ .
\end{eqnarray}
For future reference:
\begin{eqnarray}
&&s_{1} = \frac{5}{72}, \ \ r_{1} =-  \frac{7}{72}, \\
&& s_{2} = \frac{ 385}{10368}, \ \ r_{2} = - \frac{454}{10368}  \ .
\end{eqnarray}
\end{itemize}

The following alternative version of the asymptotic behavior of $\Psi$ as $\zeta\to\infty$ will be useful,
\begin{eqnarray}
\hspace{-0.3in}&&\Psi(\zeta) = \zeta^{- \frac{\sigma_{3}}{4}} \frac{e^{\frac{i \pi}{12}}}{2\sqrt{\pi}} \begin{pmatrix}{1}&{1}\\{-1}&{1}\end{pmatrix} \nonumber\\
\hspace{-0.3in}&& \times \left[ I + \sum_{k=1}^{\infty} \frac{1}{2} \left( \frac{2 }{3} \zeta^{3/2} \right)^{-k} \begin{pmatrix}{(-1)^{k}\left( s_{k} + r_{k} \right)}&{s_{k}-r_{k}}\\{(-1)^{k}(s_{k}-r_{k})}&{s_{k} + r_{k}}\end{pmatrix} \right] e^{ - \frac{ \pi i}{4} \sigma_{3} } e^{ - \frac{2}{3} \zeta^{3/2} \sigma_{3}}  \\
\hspace{-0.3in} &&=\frac{e^{ \frac{i \pi}{12}} \zeta^{-\frac{\sigma_{3}}{4}}}{2\sqrt{\pi}} \pmtwo{1}{1}{-1}{1} \nonumber\\
\hspace{-0.3in} && \times \left[
I + \frac{1}{8 \zeta^{3/2}} \pmtwo{\frac{1}{6}}{1}{-1}{-\frac{1}{6}} + \frac{ 35}{384\zeta^{3}} \pmtwo{\frac{1}{12}}{1}{-1}{-\frac{1}{12}} + \cdots
\right] e^{ - \frac{ \pi i}{4} \sigma_{3} } e^{ - \frac{2}{3} \zeta^{3/2} \sigma_{3}}  \ .
\end{eqnarray}

We define the function
\begin{eqnarray}
\label{eq:3.12}
f_b(z) = \left( \frac{ 3}{4} \right)^{2/3} \left( - \phi(z) \right)^{2/3} \ ,\qquad z\in B_+,
\end{eqnarray}
where $\phi$ is defined in (\ref{eq:1.09}). By (\ref{phi}) we obtain that $\phi$ behaves like
\begin{equation}
\label{phib}
\phi(z)\sim -\frac{8\sqrt{b-a}}{3\sigma} (z-b)^{3/2},\qquad z\to b,
\end{equation}
so that we can take $f_b$ to be analytic near $b$, and particularly in the disk $B_+$. We then have
\begin{equation}
\label{Phibl1b} f_b(z) = \left(
   \frac{2\sqrt{b-a}}{\sigma}\right)^{2/3}(z - b)\left(1 +\frac{z-b}{5(b-a)}+ \bigO\left(\left( z - b
\right)^2 \right)\right), \qquad z \to b. \
\end{equation}

The next auxiliary quantity which we define is the analytic matrix-valued function $E_{b}(z)$ for $z\in B_+$, defined as follows,
\begin{eqnarray}
\label{2.018}
 & &  E_{b}(z) =  \sqrt{\pi} e^{\frac{-i \pi}{12}}
 P^{(\infty)}(z)  z^{\alpha \sigma_{3} / 2}e^{ \frac{\pi i}{4} \sigma_{3}} \pmtwo{1}{-1}{1}{1}N^{\frac{\sigma_3}{6}}f_b(z)^{\sigma_{3}/4} \ .
 \end{eqnarray}
This function is analytic near $b$, we can see this for example by observing that it has no jump across the interval $(b-\delta, b)$:
 \begin{eqnarray} \nonumber
&&\left(E_{b} \right)_{-}^{-1} \left(E_{b} \right)_{+} = \left( f_b^{-\sigma_{3}/4} \right)_{-}\frac{1}{2} \pmtwo{1}{1}{-1}{1} e^{-\frac{\pi i}{4} \sigma_{3}}z^{- \alpha \sigma_{3}/2}  \times \\
\nonumber
&& \times \left( P^{\infty} \right)_{-}^{-1} \left( P^{\infty} \right)_{+} z^{ \alpha \sigma_{3}/2} e^{ \frac{ \pi i}{4} \sigma_{3}} \pmtwo{1}{-1}{1}{1} \left( f_b^{\sigma_{3}/4} \right)_{+} \\
\nonumber
&& = \left( f_b^{-\sigma_{3}/4} \right)_{-}\frac{1}{2} \pmtwo{1}{1}{-1}{1} e^{-\frac{\pi i}{4} \sigma_{3}} \pmtwo{0}{1}{-1}{0} e^{ \frac{ \pi i}{4} \sigma_{3}} \pmtwo{1}{-1}{1}{1} \left( f_b^{\sigma_{3}/4} \right)_{+} \nonumber\\
&& = \left( f_b^{-\sigma_{3}/4} \right)_{-}\frac{1}{2} \pmtwo{1}{1}{-1}{1} \pmtwo{0}{-i}{-i}{0}  \pmtwo{1}{-1}{1}{1} \left( f_b^{\sigma_{3}/4} \right)_{+} \nonumber \\
&& =  \left( f_b^{-\sigma_{3}/4} \right)_{-}
\pmtwo{-i}{0}{0}{i}
 \left( f_b^{\sigma_{3}/4} \right)_{+}  \ = \ \pmtwo{1}{0}{0}{1} \ .\nonumber
\end{eqnarray}
Within the neighborhood $B_+$ of $b$, we define
\begin{eqnarray}
\label{II.014} P^{(b)}(z) = E_{b}(z) \
\Psi(N^{2/3}f_b(z) ) \ e^{2N \sigma_{3}
f_b^{3/2}(z)/3 }z^{ - \frac{\alpha}{2} \sigma_{3}} \ , \ \ \mbox{ for } z \in B_+ \ .
\end{eqnarray}
The local parametrix $P^{(b)}$ then satisfies the following conditions.
\subsubsection*{RH problem for $P^{(b)}$}
\begin{itemize}
\item[(a)] $P^{(b)}$ is analytic in $\overline{B_+}\setminus \Sigma_S$,
\item[(b)] $P_+^{(b)}(z)=P_-^{(b)}(z)J_S(z)$ for $z\in B_+\cap\Sigma_S$,
\item[(c)] for $z\in\partial B_+$, we have
\begin{multline}\label{matching Pb}
P^{(b)}(z)=P^{\infty}(z)
 z^{\alpha \sigma_{3} / 2}e^{ \frac{\pi i}{4} \sigma_{3}}   \\
  \times
\left[ I + \sum_{k=1}^{\infty} \frac{1}{2} \left( \frac{2N }{3}f_b(z) ^{3/2} \right)^{-k} \pmtwo{(-1)^{k}\left( s_{k} + r_{k} \right)}{s_{k}-r_{k}}{(-1)^{k}(s_{k}-r_{k})}{s_{k} + r_{k}} \right]  \hspace{0.2in} \\
\times e^{-\frac{\pi i}{4} \sigma_{3}} z^{- \frac{\alpha}{2} \sigma_{3}},
\end{multline}
as $N\to\infty$.
\end{itemize}

\vskip 0.2in

The construction of the parametrix in a neighborhood of the point $a$ is similar.  First, we define the
transformation $f_a(z)$ for $z\in B_-$, analogous to $f_b$
defined in (\ref{eq:3.12}) above,
\begin{eqnarray}
\label{Phial1a} & & f_a(z) \equiv \left(\frac{3
}{4}\right)^{2/3} \left(
   \phi(a)
- \phi(z)   \right)^{2/3}, \ \ \ z \in B_{\delta}(a_s).  
\end{eqnarray}
As $z\to a$, by (\ref{phi})
we may choose the branch so that
$f_a(z) $ is analytic in a neighborhood of $a$, and
\begin{equation}
\label{Phial1b} f_a(z) = -\left(2
   \frac{\sqrt{b-a}}{\sigma}\right)^{2/3}(z - a)\left(1-\frac{z-a}{5(b-a)} + \bigO\left(\left( z - a
\right)^2 \right)\right), \qquad z \to a. \
\end{equation}
Next we define $E_{a}(z)$ as follows:
\begin{eqnarray}
& &  E_{a}(z) = e^{\frac{- \pi i }{12}}\sqrt{\pi}P^{\infty}(z)
\sigma_{3}  z^{ \frac{\alpha}{2} \sigma_{3}} e^{ \frac{\pi i}{4} \sigma_{3}} \pmtwo{1}{-1}{1}{1} N^{\frac{\sigma_3}{6}}f_a(z)^{\sigma_{3}/4}.
\end{eqnarray}
We now define the approximation $P^{(a)}$ within
$B_-$:
\begin{eqnarray}
\label{II.018} P^{(a)}(z) = E_{a}(z) \ \Psi(N^{2/3}f_a(z)) \ e^{2 n\sigma_{3} f_a^{3/2}(z) /3}
z^{ - \frac{\alpha}{2} \sigma_{3}}
 \sigma_{3},
\ \ \mbox{ for } z \in B_-.
\end{eqnarray}
We then have
\subsubsection*{RH problem for $P^{(a)}$}
\begin{itemize}
\item[(a)] $P^{(a)}$ is analytic in $\overline{B_-}\setminus \Sigma_S$,
\item[(b)] $P_+^{(a)}(z)=P_-^{(a)}(z)J_S(z)$ for $z\in B_-\cap\Sigma_S$,
\item[(c)] for $z\in\partial B_-$, we have
\begin{multline}\label{matching Pa}
P^{(a)}(z)=P^{\infty}(z)
 z^{\alpha \sigma_{3} / 2}e^{ \frac{\pi i}{4} \sigma_{3}}   \\
  \times
\left[ I + \sum_{k=1}^{\infty} \frac{1}{2} \left( \frac{2N }{3}f_a(z)^{3/2} \right)^{-k} \pmtwo{(-1)^{k}\left( s_{k} + r_{k} \right)}{-(s_{k}-r_{k})}{-(-1)^{k}(s_{k}-r_{k})}{s_{k} + r_{k}} \right] \times \hspace{0.2in} \\
\times e^{-\frac{\pi i}{4} \sigma_{3}} z^{- \frac{\alpha}{2} \sigma_{3}}
\end{multline}
as $N\to\infty$.
\end{itemize}

\subsection{The parametrix near $0$}

In a fixed neighborhood $B_0$ of $0$, we would like to construct a local parametrix $P^{(0)}$ which satisfies the following RH conditions.

\subsubsection*{RH problem for $P^{(0)}$}
\begin{itemize}
\item[(a)] $P^{(0)}$ is analytic in $\overline{B_0}\setminus \Sigma_S$,
\item[(b)] $P_+^{(0)}(z)=P_-^{(0)}(z)J_S(z)$ for $z\in B_0\cap\Sigma_S$,
\item[(c)] for $z\in \partial B_0$, we have
\begin{equation}\label{matching P0}
P^{(0)}(z)=P^{\infty}(z)\left(I+\bigO(e^{-cN})\right),\qquad N\to\infty.
\end{equation}
\end{itemize}
The jump relation should in other words be
\begin{eqnarray}
P^{(0)}_{+}(x) = P^{(0)}_{-}(x) \pmtwo{1}{x^{\alpha}  e^{ N \phi(x)}}{0}{1} \ \qquad 0<x<a.
\end{eqnarray}
A solution to this problem is given by
\begin{eqnarray}\label{Pinf}
P^{(0)}(z) = P^{\infty}(z) \ \pmtwo{1}{\frac{(z^{\alpha} - 1 )}{1 - e^{ 2 \pi i \alpha}} e^{  N \phi(z)}}{0}{1} \ ,
\end{eqnarray}
where the function $z^{\alpha}$ is chosen with branch cut on the {\it positive} half-line, $z^{\alpha} = |z|^{\alpha} e^{i \alpha \theta}$, $0 \le \theta < 2 \pi$. It is important to note that as $\alpha\to 0$, the limit of $P^{(0)}(z)$ exists for $z\neq 0$, and that it is given by
\begin{eqnarray}
P^{(0)}(z) = P^{\infty}(z) \ \pmtwo{1}{-\frac{\log z}{2\pi i} e^{  N \phi(z)}}{0}{1} \ .
\end{eqnarray}

\subsection{RH problem for the error, $R$}\label{section: R}
We now define $R$ in four regions of the plane, using our approximations to $S$.  Set
\begin{equation}\label{def R}
R(z) =\begin{cases} S(z) \left( P^{(0)}(z) \right)^{-1} ,& z \in B_0\ ,\\
 S(z) \left( P^{(a)}(z) \right)^{-1},& z \in B_- \ ,\\
 S(z) \left( P^{(b)}(z) \right)^{-1},& z \in B_+ \ ,\\
 S(z) \left( P^{\infty}(z) \right)^{-1},& \mbox{ everywhere else.}
\end{cases}
\end{equation}
The quantity $R$ is piecewise analytic, with jumps across the contours $\Sigma_{j}$, $j=1, \ldots, 7$, as shown in Figure \ref{figure: R}.

\subsubsection*{RH problem for $R$}
\begin{itemize}
\item[(a)] $R$ is analytic in $\mathbb C\setminus \cup_{j=1}^7\Sigma_j$,
\item[(b)] For $z\in\cup_{j=1}^7\Sigma_j$, we have $R_+(z)=R_-(z)J_R(z)$, with 
\begin{align}
&J_R(z)=P^{\infty}(z)J_S(z)\left(P^{\infty}(z)\right)^{-1},&z\in\Sigma_1\cup\Sigma_2\cup\Sigma_3\cup\Sigma_4,\\
&J_R(z)=P^{(0)}(z)\left(P^{\infty}(z)\right)^{-1},&z\in\Sigma_5,\\
&J_R(z)=P^{(a)}(z)\left(P^{\infty}(z)\right)^{-1},&z\in\Sigma_6,\\
&J_R(z)=P^{(b)}(z)\left(P^{\infty}(z)\right)^{-1},&z\in\Sigma_7.
\end{align}
\item[(c)] As $z\to\infty$, we have
\begin{equation}
R(z)=I+\bigO(z^{-1}).
\end{equation}
\end{itemize}
\begin{figure}[htb!]\begin{center}
\input{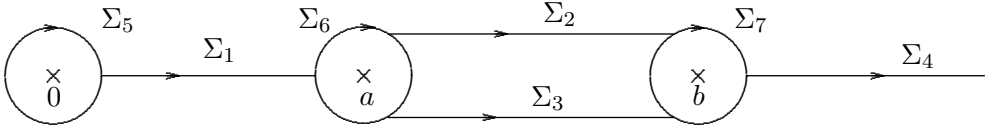}
\caption{\label{figure: R}The jump contour for $R$.}
\end{center}\end{figure}

The jump matrices across the contours $\Sigma_{1}$, $\Sigma_{2}$, $\Sigma_{3}$, and $\Sigma_{4}$ are all exponentially close to $I$ for large $N$ because $J_S$ is exponentially small and $P^{\infty}$ and its inverse are bounded on the contour, which does not contain the points $a$ and $b$.  The jump across $\Sigma_{5}=\partial B_0$ is also shown to be exponentially small by (\ref{matching P0}), since $e^{ N\phi(z)}$ is exponentially small for $z$ near $0$.
The only jumps that are not exponentially small are the ones on $\Sigma_6=\partial B_-$ and $\Sigma_7=\partial B_+$.
There we have for any $K$, by (\ref{matching Pa}) and (\ref{matching Pb}),
\begin{align}\label{asymp jump}
&J_R(z)=I+\sum_{k=1}^{K}J^{(k)}(z,\alpha)N^{-k}+\bigO(N^{-K-1}),\qquad N\to\infty.
\end{align}
By (\ref{matching Pa}) and (\ref{matching Pb}), we have
\begin{multline}\label{Jka}
J^{(k)}(z,\alpha)=\frac{1}{2} \left( \frac{2 }{3}f_a(z) ^{3/2} \right)^{-k}
\\
\times
 P^{\infty}(z)\pmtwo{(-1)^{k}\left( s_{k} + t_{k} \right)}{-i z^{\alpha}(s_{k}-r_{k})}{(-1)^{k}i z^{-\alpha}(s_{k}-r_{k})}{s_{k} + r_{k}} \left(P^{\infty}(z)\right)^{-1},\end{multline}
 for $z\in\Sigma_6$, and
\begin{multline}\label{Jkb}J^{(k)}(z,\alpha)=\frac{1}{2} \left( \frac{2 }{3}f_b(z) ^{3/2} \right)^{-k}
   \\
  \times
 P^{\infty}(z)\pmtwo{(-1)^{k}\left( s_{k} + r_{k} \right)}{i z^{\alpha}(s_{k}-r_{k})}{-i(-1)^{k} z^{-\alpha}(s_{k}-r_{k})}{s_{k} + r_{k}} \left(P^{\infty}(z)\right)^{-1},
\end{multline}
for $z\in\Sigma_7$.
Note that $J^{(1)}$ is a meromorphic function in $B_-\cup B_+$ with double poles at $a$ and $b$.

Using standard estimates, it can then be proved that
\begin{equation}\label{asymp R}
R(z)=I+\sum_{k=1}^{K}R^{(k)}(z,\alpha)N^{-k}+\bigO(N^{-K-1}),\qquad N\to\infty,
\end{equation}
uniformly for $z$ off the jump contour, where $R^{(k)}$ is analytic except on $\Sigma_6\cup\Sigma_7$.
Substituting (\ref{asymp jump}) and (\ref{asymp R}) into the jump relation $R_+=R_-J_R$ leads us to the following jump relation for $R^{(1)}$:
\begin{equation}
R_+^{(1)}(z,\alpha)-R_-^{(1)}(z,\alpha)=J^{(1)}(z,\alpha),\qquad \mbox{ for $z\in\Sigma_6\cup\Sigma_7$.}
\end{equation}
In addition, since $R\to I$ at infinity, $R^{(1)}(z,\alpha)\to 0$ as $z\to\infty$. These conditions determine $R^{(1)}$ uniquely, and we obtain
\begin{align}
&\label{R1}R^{(1)}(z)=Q(z,\alpha),& z\in\mathbb{C}\setminus \left(\overline{B_-}\cup \overline{B_+}\right),\\
&
R^{(1)}(z)=-J^{(1)}(z,\alpha)+Q(z,\alpha),& z\in B_-,
\\
&R^{(1)}(z)=-J^{(1)}(z,\alpha)+Q(z,\alpha),& z\in B_+,
\end{align}
where
\begin{multline}
Q(z,\alpha)=\dfrac{1}{z-a}\res[z=a]J^{(1)}(z,\alpha)+\dfrac{1}{(z-a)^2}\res[z=a]\left[(z-a)J^{(1)}(z,\alpha)\right]\\
+\dfrac{1}{z-b}\res[z=b]J^{(1)}(z,\alpha)
+\dfrac{1}{(z-b)^2}\res[z=b]\left[(z-b)J^{(1)}(z,\alpha)\right].
\end{multline}

As $z\to\infty$, we have an asymptotic expansion of the form
\begin{equation}\label{asymp Rz}
R(z)=I+\sum_{k=1}^{K}R_k(N,\alpha)z^{-k}+\bigO(z^{-K-1}),\qquad z\to\infty.
\end{equation}
From the asymptotics (\ref{asymp R}) for $R$ and the formula (\ref{R1}) for $R^{(1)}$, we obtain
\begin{align}
&\label{R1as}R_1(N,\alpha)=\frac{1}{N}\left(\res[z=a]J^{(1)}(z,\alpha)+\res[z=b]J^{(1)}(z,\alpha)\right)+\bigO(N^{-2}),\\
&R_2(N,\alpha)=\frac{1}{N}\left(\res[z=a]\left[(z-a)J^{(1)}(z,\alpha)\right]+\res[z=b]\left[(z-b)J^{(1)}(z,\alpha)\right]\right)\nonumber \\
&\quad +\frac{1}{N}\left(a\res[z=a]J^{(1)}(z,\alpha)+b\res[z=b]J^{(1)}(z,\alpha)\right) +\bigO(N^{-2}),
\end{align}
as $N\to\infty$.
Since the jump matrix $J_R$ depends analytically on the parameters $s>0$ and $r>2\sqrt{s}$, one can deduce that the functions $R_1$ and $R_2$, and all of the coefficients in their large $N$ asymptotic expansions, depend analytically on $s$ and $r$.

After a straightforward but long calculation, one obtains, writing $c_1=\sqrt{a}-\sqrt{b}$ and $d$ as in (\ref{eq:2.06}),
\begin{align}
&R_{1,11}(N,\alpha)=-R_{1,22}(N,\alpha)\\&\qquad\qquad=\left(\frac{\sigma}{32(b-a)}-\frac{\alpha^2\sigma}{16b(b-a)}(\sqrt{a}-\sqrt{b})^2\right)\frac{1}{N}+\bigO(N^{-2})\label{R111as},\\
&R_{1,12}(N,\alpha)\\&\qquad\qquad=\frac{id(\alpha)^2\sigma}{48b(b-a)} (3\alpha^2(\sqrt{a}-\sqrt{b})^2-6\alpha \sqrt{ab}+2b+6\alpha b)\frac{1}{N}+\bigO(N^{-2}),\\
&R_{1,21}(N,\alpha) =\frac{i\sigma}{48b(b-a)d(\alpha)^{2}} \left(3\alpha^2(\sqrt{a}-\sqrt{b})^2+6\alpha\sqrt{ab}+2b-6\alpha b\right)\frac{1}{N}\\&\hspace{10cm}+\bigO(N^{-2}),\label{R121as}\\
&R_{2,11}(N,\alpha)=\frac{\sigma}{32}\left(-\frac{5}{6}+\frac{-2\alpha^2(\sqrt{a}-\sqrt{b})^2+b}{b-a}\right)\frac{1}{N}+\bigO(N^{-2}),\label{R211as}
\end{align}
as $N\to\infty$.

\section{Proof of Theorem~\ref{theorem 1}}\label{section: as G}

We will now use the large $N$ asymptotics obtained from the RH analysis to obtain asymptotics for the function $\mathcal G_{N}(\alpha;r,\sigma)$ in (\ref{GN}).

Let us first consider the quantities $\tilde Y(z)$ for $z\in B_0$, which we need to compute the last two terms in (\ref{GN}).
From (\ref{def tilde Y}), (\ref{def T}), (\ref{def S}), (\ref{def R}), and (\ref{Pinf}), it follows that
\begin{eqnarray*}
\tilde Y(z)&=&e^{N\ell\sigma_3/2}R(z)P^{(0)}(z)e^{N(g(z)-\ell/2)\sigma_3}\begin{pmatrix}1&-\frac{w(z;\alpha)}{1-e^{2\pi i\alpha}}\\0&1\end{pmatrix}\\
&=&e^{N\ell\sigma_3/2}R(z)P^{\infty}(z) \ \pmtwo{1}{\frac{(z^{\alpha} - 1 )}{1 - e^{ 2 \pi i \alpha}} e^{  N \phi(z)}}{0}{1}e^{N(g(z)-\ell/2)\sigma_3}\begin{pmatrix}1&-\frac{w(z;\alpha)}{1-e^{2\pi i\alpha}}\\0&1\end{pmatrix}.
\end{eqnarray*}
By (\ref{eq:1.09}), we have
\begin{equation}
\tilde Y(z)=e^{N\ell\sigma_3/2}R(z)P^{\infty}(z) \ \pmtwo{1}{\frac{-1}{1 - e^{ 2 \pi i \alpha}} e^{  N \phi(z)}}{0}{1}e^{N(g(z)-\ell/2)\sigma_3},
\end{equation}
and in the large $N$ limit
\begin{align}
&\label{Y120}\tilde Y_{12}(0)=e^{-Ng(0)}e^{N\ell}(1+\bigO(N^{-1}))P_{12}^{\infty}(0),\\
&\label{Y220}\tilde Y_{22}(0)=e^{-Ng(0)}(1+\bigO(N^{-1}))P_{22}^{\infty}(0).
\end{align}
For the first column of $Y$ near $0$, we have
\begin{align}
&\label{Y110}Y_{11}(z)=(1+\bigO(N^{-1}))P_{11}^{\infty}(z)e^{Ng(z)},\\
&\label{Y210}Y_{21}(z)=(1+\bigO(N^{-1}))P_{21}^{\infty}(z)e^{Ng(z)}e^{-N\ell}.
\end{align}
Let us now compute the moments at infinity $Y_1$ and $Y_2$ defined by (\ref{Yinfty}).
For $z$ outside the lenses, we have
\begin{eqnarray}
Y(z) &=& e^{ N \ell \sigma_{3} / 2} S(z) e^{ N ( g(z) - \ell/2)\sigma_{3}}  \nonumber\\
&=& e^{ N \ell \sigma_{3} / 2} R(z)  P^{\infty}(z) e^{ N ( g(z) - \ell/2)\sigma_{3}}.
\end{eqnarray}
By (\ref{egasymp}), (\ref{asympPinf}), and (\ref{asymp Rz}), it follows that
\begin{equation}\label{Y1}
Y_{1}(\alpha) = e^{ N \ell \sigma_{3}/2} \left[
P_{1}(\alpha) + G_{1} + R_{1}(\alpha)
\right]e^{- N \ell \sigma_{3}/2},
\end{equation}
and
\begin{multline}\label{Y2}
Y_{2}(\alpha) = e^{ N \ell \sigma_{3}/2} \left[
P_{2}(\alpha) + G_{2} + R_{2}(\alpha)+(P_1(\alpha)+R_1(\alpha))G_1 +R_1(\alpha)P_1(\alpha)
\right]\\ \times e^{- N \ell \sigma_{3}/2}.
\end{multline}
Entry-wise, we obtain as $N\to\infty$, by (\ref{R111as})-(\ref{R211as}),
\begin{align}
&\label{Y111}Y_{1,11}(\alpha)=-Nm_1-d_1+\frac{(b-2\alpha^2(\sqrt{a}-\sqrt{b})^2)\sigma}{32b(b-a)N}+\bigO(N^{-2}),\\
&e^{-N\ell}Y_{1,12}(\alpha)=id^2\frac{b-a}{4}\\&\qquad+\frac{i d^2\sigma}{48ab(b-a)N}(3\alpha^2 (\sqrt{a}-\sqrt{b})^2-6\alpha\sqrt{ab}+2b+6\alpha b)+\bigO(N^{-2}),\\
&e^{N\ell}Y_{1,21}(\alpha)=-i\frac{b-a}{4d^2}\\&\qquad+\frac{i\sigma}{48b(b-a)d^2}(3\alpha^2(\sqrt{a}-\sqrt{b})^2+6\alpha\sqrt{ab}+2b-6\alpha b)+\bigO(N^{-2}),\\
&Y_{1,22}(\alpha)=-Y_{1,11}(\alpha).
\end{align}
For $Y_{2,11}$, we have
\begin{multline}
Y_{2,11}(\alpha)=\frac{m_1^2}{2}N^2-\left(\frac{m_2}{2}-m_1d_1\right)N+\frac{(b-a)^2}{32}+d_1^2-d_2\\
\qquad-\frac{m_1\sigma(-2\alpha^2(\sqrt{a}-\sqrt{b})^2+b)}{32b(b-a)}+\bigO(N^{-1}).\label{Y222}
\end{multline}
Recall from (\ref{GN}) that \begin{equation}
\mathcal G_{N}(\alpha)=T_1(\alpha)+T_2(\alpha)+T_3(\alpha)+T_4(\alpha)+T_5(\alpha),
\end{equation}
with
\begin{align}
&\label{T1}T_1(\alpha)=\frac{\alpha}{2}\frac{ Y_{1,12}'}{Y_{1,12}},\\
&T_2(\alpha)=2\frac{N}{\sigma}\left(2Y_{2,11}'-rY_{1,11}'+(\det Y_1)'\right),\\
&T_3(\alpha)=\frac{N}{2}\left(\frac{Y_{1,12}'}{Y_{1,12}}-\frac{Y_{1,21}'}{Y_{1,21}}\right),\\
&T_4(\alpha)=\sqrt{Y_{1,12}Y_{1,21}} \left(\frac{Y_{11}(0;\alpha)}{\sqrt{-iY_{1,12}}} \right)'\dfrac{\alpha}{\sqrt{iY_{1,21}}}\tilde Y_{22}(0;\alpha),\\
&T_5(\alpha)=-\sqrt{Y_{1,12}Y_{1,21}} \left(\frac{Y_{21}(0;\alpha)}{\sqrt{ i Y_{1,21}}}\right)'\frac{\alpha}{\sqrt{- iY_{1,12}}}\tilde Y_{12}(0;\alpha).\label{T5}
\end{align}
We can now substitute (\ref{Y120})-(\ref{Y210}) and (\ref{Y111})-(\ref{Y222}) into (\ref{T1})-(\ref{T5}) to obtain asymptotics for $T_1(\alpha),\ldots, T_5(\alpha)$. Their expressions can be computed further by the formulas (\ref{R1as})-(\ref{R121as}). We used Mathematica for this lengthy but straightforward calculation and obtain
\begin{align*}
&T_1=\alpha\frac{d'(\alpha)}{d(\alpha)}
+\bigO(N^{-1}),
\\
&T_2=\left(rd_1'(\alpha)+2d_1(\alpha)d_1'(\alpha)-2d_2'(\alpha)
\right)
\frac{2N}{\sigma}-\frac{\alpha(\sqrt{a}-\sqrt{b})^2}{8b}+\bigO(N^{-1}),\\&
T_3=2\frac{d'(\alpha)}{d(\alpha)}N
+\frac{\alpha\sigma}{2b(\sqrt{a}+\sqrt{b})^2}+\bigO(N^{-1}),\\
&T_4=\frac{\alpha}{4}\left(\gamma+\gamma^{-1}\right)^2\left(\frac{d'(\alpha)}{d(\alpha)}-\frac{1}{2}\log(ab)\right) +\bigO(N^{-1}),\qquad \gamma=b^{1/4}a^{-1/4},\\
&T_5=-\frac{\alpha}{4}\left(\gamma-\gamma^{-1}\right)^2\left(\frac{d'(\alpha)}{d(\alpha)}-\frac{1}{2}\log(ab)\right) +\bigO(N^{-1}).
\end{align*}
Substituting (\ref{eq:2.06}), (\ref{d1}), (\ref{d2}), and $r=a+b$ (see (\ref{a b})),
and summing up the expressions for $T_1,\ldots, T_5$, we obtain
\begin{equation}
\mathcal G_{N}(\alpha;r,\sigma)=\mathcal G^{(0)}(r,\sigma)N
+\alpha \mathcal G^{(1)}(r,\sigma)+\bigO(N^{-1}),
\end{equation}
where
\begin{align}
&\label{G0}\mathcal G^{(0)}(r,\sigma)=\frac{(\sqrt{b}-\sqrt{a})^2}{2(\sqrt{b}+\sqrt{a})^2}+2\log\left(\frac{\sqrt{a}+\sqrt{b}}{2}\right),\\
&\mathcal G^{(1)}(r,\sigma)=
2\log\frac{\sqrt{a}+\sqrt{b}}{2(ab)^{1/4}}.\label{G1}
\end{align}

It should be noted that all of the error terms we used above can be expanded in asymptotic series in negative integer powers of $N$, and this can be used to obtain a full asymptotic expansion for $\mathcal G_{N}(\alpha;r,\sigma)$:
\begin{equation}\label{GNM}
\mathcal G_{N}(\alpha;r,\sigma)=\mathcal G^{(0)}(r,\sigma)N
+\alpha \mathcal G^{(1)}(r,\sigma)
+\sum_{j=1}^k f^{(j)}(\alpha;r,\sigma)N^{-j}+  \bigO(N^{-k-1}),
\end{equation}
for any $k\in\mathbb N$. The constants $f^{(j)}(\alpha;r,\sigma)$ may depend on $\alpha$ and are real analytic functions in $r,\sigma$.

\medskip

Let us now compute the asymptotics for $\log Z_{2N}$ using (\ref{partition rec6}).
Since the $\bigO(N)$-term in the expansion of $\mathcal G_{N}$ is independent of $\alpha$, the leading order terms of the two integrals in (\ref{partition rec6}) cancel out against each other, and we obtain
\begin{equation}
\log \left[\frac{1}{(2N)!}Z_{2N}\right]=2\log \left[\frac{1}{N!}Z_{N,s}^{GUE}\right]+\frac{N^2r^2}{s}+\frac{1}{4}\mathcal G^{(1)}(r,s)+ \bigO(N^{-1}),
\end{equation}
as $N\to\infty$.
Replacing $2N$ by $N$, we get the following large $N$ asymptotics for $N$ even:
\begin{equation}
\log \left[\frac{1}{N!}Z_N\right]=2\log \left[\frac{1}{(N/2)!}Z_{N/2,s}^{GUE}\right]+\frac{N^2r^2}{4s}+\frac{1}{4}\mathcal G^{(1)}(r,s)+ \bigO(N^{-1}).
\end{equation}
In order to have a more compact formula, we substract $2\log \left[\frac{1}{(N/2)!}Z_{N/2,\sigma^*}^{GUE}\right]$, defined in (\ref{ZGUE-intro}), at both sides of this formula, with $\sigma^*=4e^{3/2}$, and obtain by (\ref{ZGUE-intro}),
\begin{eqnarray}
\log \left[\frac{(N/2)!^2 Z_N}{N!\left(Z_{N/2,4e^{3/2}}^{GUE}\right)^2}\right]&=&2\log \left[\frac{Z_{N/2,s}^{GUE}}{Z_{N/2,4e^{3/2}}^{GUE}}\right]+\frac{N^2r^2}{4s}+\frac{1}{4}\mathcal G^{(1)}(r,s)+ \bigO(N^{-1})\nonumber\\
&=&-N^2 F_0+\frac{1}{2}\log\frac{\sqrt{a}+\sqrt{b}}{2(ab)^{1/4}}+ \bigO(N^{-1}),\label{asZNeven1}
\end{eqnarray}
where 
\begin{equation}
F_0=-\frac{1}{4}\log\frac{s}{4}+\frac{3}{8}-\frac{r^2}{4s}.
\end{equation}
Note that the choice of $\sigma^*=4e^{3/2}$ in the denominator at the left hand side of (\ref{asZNeven1}) is such that 
\[\lim_{N\to\infty}\frac{1}{N^2}\log\left[\frac{1}{(N/2)!}Z_{N/2,4e^{3/2}}^{GUE}\right]=0,\] by (\ref{expansion GUE1}), so that $F_0$ is the free energy defined by (\ref{energy}).
The error term $\bigO(N^{-1})$ in (\ref{asZNeven1}) can be expanded as an asymptotic series in negative integer powers of $N$ of the form
\[\sum_{j=1}^kc^{(j)}(r,s)N^{-j}+\bigO(N^{-k-1}),\]
for any $k\in\mathbb N$, where the coefficients $c^{(j)}(r,s)$ are analytic functions of $r$ and $s$.

\medskip

For the computation of $\log Z_{2N+1}$, the integrals in (\ref{partition rec7}) do not cancel out exactly, but they contribute to the $\bigO(1)$-term. Substituting  (\ref{G0}),  (\ref{G1}) and (\ref{GNM}), into (\ref{partition rec7}), we obtain
\begin{multline}\label{finalZNodd0}
\log Z_{2N+1}=\log(2N+1)!+\log \left[\frac{1}{(N+1)!}Z_{N+1,s_+}^{GUE}\right]+\log \left[\frac{1}{N!}Z_{N,s_-}^{GUE}\right]\\+\frac{(2N+1)^2r^2}{4s}
-\frac{1}{2}\log\left((\sqrt{a}+\sqrt{b})(ab)^{1/4}\right)+\frac{1}{2}\log 2+\bigO(N^{-1}),
\end{multline}
as $N\to\infty$.
Replacing $2N+1$ by $N$, this becomes
\begin{multline}\label{finalZNodd01}
\log Z_{N}=\log N!+\log \left[\frac{1}{\left(\frac{N+1}{2}\right)!}Z_{(N+1)/2,\widehat s_+}^{GUE}\right]+\log \left[\frac{1}{\left(\frac{N-1}{2}\right)!}Z_{(N-1)/2,\widehat s_-}^{GUE}\right]\\+\frac{N^2r^2}{4s}
-\frac{1}{2}\log\left((\sqrt{a}+\sqrt{b})(ab)^{1/4}\right)+\frac{1}{2}\log 2+ \bigO(N^{-1}),
\end{multline}
for $N$ odd,
where $\widehat s_\pm=s\left(1\pm \frac{1}{N}\right)$,
and substracting $\log \left[\frac{Z_{(N-1)/2,4e^{3/2}}^{GUE} Z_{(N+1)/2,4e^{3/2}}^{GUE}}{\left(\frac{N-1}{2}\right)! \left(\frac{N+1}{2}\right)!}\right]$ on both sides, we find  that
\begin{multline}\label{finalZNodd02}
\log \left[\frac{\left(\frac{N-1}{2}\right)! \left(\frac{N+1}{2}\right)!Z_{N}}{N!Z_{(N-1)/2,4e^{3/2}}^{GUE} Z_{(N+1)/2,4e^{3/2}}^{GUE}}\right]=\log \left[\frac{Z_{(N+1)/2,\widehat s_+}^{GUE}}{Z_{(N+1)/2,4e^{3/2}}^{GUE}}\right]+\log \left[\frac{Z_{(N-1)/2,\widehat s_-}^{GUE}}{Z_{(N-1)/2,4e^{3/2}}^{GUE}}\right]\\+\frac{N^2r^2}{4s}
-\frac{1}{2}\log\left((\sqrt{a}+\sqrt{b})(ab)^{1/4}\right)+\frac{1}{2}\log 2+ \bigO(N^{-1}).
\end{multline}
By (\ref{ZGUE-intro}) and using the fact that $s=(b-a)^2/4$, we finally obtain
\begin{multline}\label{finalZNodd04}
\log \left[\frac{\left(\frac{N-1}{2}\right)! \left(\frac{N+1}{2}\right)!Z_{N}}{N!Z_{(N-1)/2,4e^{3/2}}^{GUE} Z_{(N+1)/2,4e^{3/2}}^{GUE}}\right]\\=-N^2F_0
+\frac{1}{2}\log\left(\frac{\sqrt{b}-\sqrt{a}}{2(ab)^{1/4}}\right)
+ \bigO(N^{-1}).
\end{multline}
The error term $\bigO(N^{-1})$ in (\ref{finalZNodd04}) can be expanded as an asymptotic series in negative integer powers of $N$ of the form
\[\sum_{j=1}^k\widetilde c^{(j)}(r,s)N^{-j}+\bigO(N^{-k-1})\]
for any $k\in\mathbb N$, where the coefficients $\widetilde c^{(j)}(r,s)$ are analytic functions of $r$ and $s$.
Theorem \ref{theorem 1} is proven by (\ref{asZNeven1}) and (\ref{finalZNodd04}).

\section{Deformation of one-cut regular $V$: proof of Theorem \ref{theorem: 1} \label{Sec1-cut}}
We assume in this section that $V=V_{\vec t}\in P_m^{(1)}$ is $1$-cut regular.
In general, the equilibrium measure associated to $V\in P_m^{(1)}$ can be written in the form \cite{DKM} \begin{equation}\label{densitymu1}
d\mu_V(x)=\frac{1}{ic}\mathcal{R}_+^{1/2}(x)h(x)dx, \qquad x\in[a,b],\qquad c>0,
\end{equation}
where \begin{equation}\mathcal{R}(z)=(z-a)(z-b),\end{equation}
with $\mathcal R^{1/2}(z)$ analytic in $\mathbb C\setminus[a,b]$ and positive for $z>b$, and $\mathcal R_+^{1/2}(x)$, $x\in(a,b)$, the boundary value from the upper half plane. The function 
$h$ is a monic polynomial of degree ${\rm deg\,}V-2$,
and can be expressed as follows in terms of its zeros, \begin{equation}
h(z)=\prod_{j=1}^{n_1}(z-\gamma_j). \prod_{j=1}^{n_2}(z-\eta_j).\prod_{j=1}^{n_3}(z-z_j)(z-\bar z_j),\label{h}
\end{equation}
where
\begin{equation}\gamma_{n_1}<\ldots<\gamma_2<\gamma_1<a<b<\eta_1<\eta_2<\ldots
<\eta_{n_2},\ \Im z_j\geq 0,\ z_j\notin
[a,b].\label{zeros}\end{equation} 
The number of zeros $n_1$ at the left of $[a,b]$ and the number of zeros $n_2$ at the right of $[a,b]$ are always even, the other zeros come in complex conjugate pairs.
 Using this notation, it is possible that some
of the $z_j$'s are real or coincide with each other or with one
of the real zeros $\gamma_j,\eta_j$. 
Note that
in the case of a quadratic external field $V$ as in Example
\ref{exampel: 1}, we have $n_1=n_2=n_3=0$.

A probability measure of the form (\ref{densitymu1})
is not necessarily the equilibrium measure for an external field $V$, or in other words it does not necessarily
belong to the image of the map $V\in P_m^{(1)} \mapsto \mu_V\in\mathcal P$. However, if a measure of the form
(\ref{densitymu1})
is an equilibrium measure, it follows from (\ref{var eq}), after a straightforward calculation,
that $V'$ has to be equal to
\begin{equation}\label{Vx}
V'(x)=\frac{2}{c}\PVint_a^b \frac{|\mathcal{R}_+^{1/2}(y)h(y)|}{x-y}dy,\qquad x\in[a,b],
\end{equation}
and consequently
\begin{equation}\label{V}
V'(z)=\frac{1}{ic}\int_\gamma \frac{\mathcal{R}^{1/2}(y)h(y)}{z-y}dy,
\end{equation}
with $\gamma$ a clock-wise oriented contour surrounding $[a,b]$ and $z$.
 The potential $V$ is uniquely fixed by the above condition and by $V(0)=0$.

Conversely, for any $\mu$ of the form
(\ref{densitymu1}), if we define $V'$ as above, the variational
equality (\ref{var eq}) always holds. Nevertheless, it is still
possible that the variational inequality (\ref{var ineq}) is
violated, and in that case $\mu$ is not an equilibrium measure. The
following lemma gives a convenient characterization of the image of
the map $V\in P_m^{(1)} \mapsto \mu_V\in\mathcal P$.

Define
\begin{align}
&\label{mj}m_j=\int_{\eta_j}^{\eta_{j+1}}|\mathcal{R}^{1/2}(x)h(x)|dx\geq 0, &j=0,1,\ldots, n_2-1,\\
&\label{tilde mj}\tilde
m_j=\int_{\gamma_{j+1}}^{\gamma_{j}}|\mathcal{R}^{1/2}(x)h(x)|dx\geq 0, &
j=0,1,\ldots, n_1-1,
\end{align}
where we use the convention that $\gamma_0=a$, $\eta_0=b$.

\begin{lemma}\label{lemma: m}
Let $h$ be of the form (\ref{h})-(\ref{zeros}) with $n_1, n_2$ even, and let $\mu$ be defined by \begin{equation}\label{mu1}
d\mu(x)=\frac{1}{c}|\mathcal{R}^{1/2}(x)h(x)|dx,\qquad x\in[a,b],\qquad c=\int_a^b|\mathcal{R}^{1/2}(x)h(x)|dx.
\end{equation}
Let $V$ be given by (\ref{V}) with integration constant fixed by the condition $V(0)=0$.
Then, $V$ is one-cut regular ($V\in P_m^{(1)}$) and $\mu$ is equal to the equilibrium measure $\mu_V$ if and only if
\begin{align}
&\label{sum1}\sum_{j=0}^{k-1} (-1)^jm_j>0,&0<k \leq n_2,\ k\mbox{ even,}\\
&\label{sum2}\sum_{j=0}^{k-1} (-1)^j\tilde m_j>0,&0<k \leq n_1,\
k\mbox{ even.}
\end{align}
\end{lemma}
\begin{proof}
Let us consider the function 
\begin{equation}
\label{G}
G(z)=\frac{1}{\pi i}\int_a^b\frac{d\mu(y)}{z-y}=-\dfrac{1}{\pi c}\int_{a}^b \dfrac{\mathcal{R}_+^{1/2}(y)h(y)}{z-y}dy.
\end{equation}
By Cauchy's theorem for $z\notin[a,b]$,
\[
G(z)=-\frac{1}{ic}\mathcal{R}^{1/2}(z)h(z)-\dfrac{1}{2\pi c}\int_{\gamma} \dfrac{\mathcal{R}^{1/2}(y)h(y)}{z-y}dy,
\]
where $\gamma $ is a clock-wise contour surrounding the interval $[a,b]$ and $z$.
By (\ref{V}), we have 
\begin{equation}
\label{G1c}
G(z)=-\frac{1}{ic}\mathcal{R}^{1/2}(z)h(z)+\dfrac{1}{2\pi i}V'(z),\quad z\notin[a,b].
\end{equation}
It follows that 
\[G_+(x)+G_-(x)=\frac{1}{\pi i}V'(x),\qquad x\in[a,b],\]
and upon integrating, we have that there exists a constant $\ell_V$ such that
\begin{equation}\label{vareq2}
2\int_a^b\log|x-y|d\mu(y)-V(x)=\ell_V,\qquad x\in[a,b]. 
\end{equation}
Evaluating at $x=b$ and $x=a$, we obtain
\begin{equation}
\label{lV}
\ell_V=2\int\log|b-y|d\mu(y)-V(b)=2\int\log|a-y|d\mu(y)-V(a).
\end{equation}
Recall that the equilibrium measure $\mu_V$ is characterized by the variational equality (\ref{var eq}) and the variational inequality (\ref{var ineq}). The measure $\mu$ already satisfies the equality by (\ref{vareq2}), and this means that $\mu$ is equal to the equilibrium measure $\mu_V$ if and only if 
\begin{equation}\label{varineq2}
2\int_a^b\log|x-y|d\mu(y)-V(x)\leq \ell_V,\qquad x\in\mathbb R. 
\end{equation}
Moreover, since $h$ has no zeros on $[a,b]$, $V$ is one-cut regular and $\mu=\mu_V$ if and only if 
\begin{equation}\label{varineq3}
2\int_a^b\log|x-y|d\mu(y)-V(x)<\ell_V,\qquad x\in\mathbb R \setminus [a,b], 
\end{equation}
which is by (\ref{lV}) equivalent to
\[
2\int_a^b\log\left|\frac{x-y}{b-y}\right|d\mu(y)-\int_{b}^xV'(\xi)d\xi< 0,\quad x\in\mathbb{R}\backslash[a,b].
\]
For $x>b$, by (\ref{G}), this inequality can be written as\[
\int_b^x\left[2\pi iG(\xi)-V'(\xi)\right]d\xi <0,
\]
or, by (\ref{G1c}), 
\begin{equation}
\label{in1}
H_b(x):=\dfrac{1}{c}\int_{b}^x\mathcal{R}^{1/2}(\xi)h(\xi) d\xi >0,\quad x>b.
\end{equation}
Repeating the same computation for $x<a$, we get
\begin{equation}
\label{in2}
H_a(x):=\dfrac{1}{c}\int_{x}^a\mathcal{R}^{1/2}(\xi)h(\xi) d\xi >0,\quad x<a.
\end{equation}
\begin{itemize}
\item[(i)] Assume first that $\mu=\mu_V$, $V\in P_m^{(1)}$.
Then we have (\ref{varineq3}), which implies (\ref{in1}) and (\ref{in2}) by the above discussion, and in particular we have 
\begin{align}
\label{H1}
&0<H_b(\eta_k)=\frac{1}{c} \int_b^{\eta_k} \mathcal{R}^{1/2}(\xi)h(\xi)d\xi=\frac{1}{c}(m_0-m_1+\ldots - m_{k-1}),\\
&\label{H2}0<H_a(\gamma_k)=\frac{1}{c} \int_{\gamma_k}^a\mathcal{R}^{1/2}(\xi)h(\xi)d\xi=\frac{1}{c}(\tilde m_0-\tilde m_1+\ldots - \tilde m_{k-1}),
\end{align}
for $k>0$ even.
This implies (\ref{sum1})-(\ref{sum2}). 
\item[(ii)] Conversely, assume that (\ref{sum1})-(\ref{sum2}) hold. In order to show that $\mu$ is the equilibrium measure, we need to prove (\ref{in1}) and (\ref{in2}). From the fact that $h$ is positive on $[a,b]$ and changes sign at $\eta_1, \ldots, \eta_{n_2}$, $\gamma_1,\ldots, \gamma_{n_1}$,
it is clear that $H_b(x)$ is increasing for $\eta_{2j}<x<\eta_{2j+1}$ and  decreasing for $\eta_{2j+1}<x<\eta_{2j+2}$, while   $H_a(x)$ is decreasing for $\gamma_{2j+1}<x<\gamma_{2j}$, and increasing for  $\gamma_{2j+2}<x<\gamma_{2j+1}$. Consequently, it is sufficient to prove that
$H_a(\gamma_k)>0$ and $H_b(\eta_k)>0$ for $k>0$ even. This follows from the equalities in (\ref{H1})-(\ref{H2}).
\end{itemize}
\end{proof}

We will prove Theorem \ref{theorem: 1} by induction on $n_1+n_2$.
i.e.  on the number of real zeros of $h$ with odd multiplicity.

\subsection{Basis $n_1+n_2=0$} We assume here that
$V=V_{\vec t}\in P_m^{(1)}$, and that $\mu$ is of the form
(\ref{densitymu1})-(\ref{h}) with $n_1+n_2=0$. We will now define a
continuous deformation $V=V_{\vec t(s)}$, $s\in[0,1]$, of $V$ such
that $V_{\vec t(s)}\in P_m^{(1)}$ and such that the following three
conditions hold:
\begin{itemize}
\item[(i)] $V_{\vec t(1)}(x)=V_{\vec t}(x)$,
\item[(ii)] $V_{\vec t(0)}(x)=\frac{x^2}{2}$,
\item[(iii)] for all $s\in[0,1]$, $V_{\vec t(s)}\in P_m^{(1)}$.
\end{itemize}
We will define this deformation implicitly by deforming the
equilibrium measure $\mu_V$ to a measure $\mu_s$ for $s\in [0,1]$
such that
\begin{equation}
\mu_1=\mu_V,\qquad \mu_0 =\mu_{x^2/2}.
\end{equation}
We will then use Lemma \ref{lemma: m} to prove that $\mu_s$ is an
equilibrium measure for all $s\in (0,1]$. The deformation of
$\mu$ defines a deformation of $V$ by (\ref{V}).

\medskip

 If $n_1+n_2=0$, $h$
has no real zeros of odd multiplicity and can be written as
\begin{equation}
h(z)=\prod_{j=1}^{n_3}(z-z_j)(z-\bar z_j).
\end{equation}
Let us define, for $s\in(0,1]$,
\begin{align}
&h(z;s)=\prod_{j=1}^{n_3}(z-z_j(s))(z-\bar z_j(s)),&
z_j(s)=z_j+i\frac{1-s}{s}, \\
&\mathcal{R}(z;s)=(z-a(s))(z-b(s)),\\& a(s)=sa-2(1-s),& b(s)=bs+2(1-s),
\end{align}
and
\begin{equation}
d\mu_s(x)=\frac{1}{ic(s)}\mathcal{R}_+^{1/2}(x;s)\,h(x;s)dx, \qquad x\in[a(s),b(s)],
\end{equation}
where \begin{equation}\label{cs}c(s)=-i\int_{a(s)}^{b(s)}\mathcal{R}_+^{1/2}(x;s)\,h(x;s)dx\end{equation} is chosen such that $\mu_s$ has total mass $1$.

For all $s\in(0,1)$, $h$ has no sign changes on the real line and thus satisfies
automatically the conditions of Lemma \ref{lemma: m}. We have that
$\mu_s$ is the equilibrium measure corresponding to the external
field (\ref{V}) for all $s\in (0,1)$, given by
\begin{equation}
V_{\vec t(s)}'(z)=\frac{1}{ic(s)}\int_\gamma
\frac{\mathcal{R}^{1/2}(y;s)h(y;s)}{z-y}dy.
 \end{equation}
Since the dependence of $V$ on the coefficients of $h$ and $\mathcal R$ is
continuous, $\vec t(s)$ describes a continuous path for $0< s\leq 1$.
Moreover, it is straightforward to verify by residue calculus that $V_{\vec t(s)}$ is a polynomial of degree ${\rm deg\, } h+2={\rm deg\, } V$ for $0<s<1$.

When taking the limit $s\to 0$, the zeros of $h$ tend to infinity
and $c(s)\to\infty$, and one verifies that $\frac{1}{c(s)}h(z;s)\to
\frac{1}{2\pi}$ uniformly on compacts. Indeed, as $s\to\infty$, we have
\begin{eqnarray*}
\frac{1}{c(s)}h(z;s)&=&\frac{\prod_{j=1}^{n_3}(z-z_j(s))(z-\bar z_j(s))}{\int_{a(s)}^{b(s)}|\mathcal{R}_+^{1/2}(x;s)|\,h(x;s)dx}\\
&\sim & \frac{s^{-2n_3}}{s^{-2n_3}\int_{-2}^2\sqrt{4-x^2} dx}=\frac{1}{2\pi}.
\end{eqnarray*}

We thus obtain
\begin{equation}
d\mu_0(x)=\frac{1}{2\pi}\sqrt{4-x^2}dx, \qquad x\in[-2,2],
\end{equation}
and \begin{equation}V_{\vec t(0)}(z)= \frac{1}{2\pi i}\int_\gamma
\frac{[(y-2)(y+2)]^{1/2}}{z-y}dy=z^2/2.
\end{equation}
This completes the construction of the deformation
of $V$ in the case $n_1+n_2=0$.

\subsection{Inductive step} We assume that Theorem
\ref{theorem: 1} holds true for any external field $\tilde{V}\in P_m^{(1)}$
which is such that $\tilde{h}$ in (\ref{mu1}) is of the form
(\ref{h})-(\ref{zeros}) with $n_1+n_2<N_1+N_2$, i.e.  $\tilde{h}$ has
strictly less than $N_1+N_2$ real zeros of odd multiplicity, with
$N_1+N_2>0$. Next, let us consider $h$ of the form
\begin{equation}
h(z)=\prod_{j=1}^{N_1}(z-\gamma_j). \prod_{j=1}^{N_2}(z-\eta_j).\prod_{j=1}^{N_3}(z-z_j)(z-\bar z_j),\quad \Im z_j\geq 0,\ z_j\notin [a,b],
\label{h1}
\end{equation}
where
\begin{equation}
\gamma_{N_1}<\ldots<\gamma_2<\gamma_1<\gamma_0=a<b=\eta_0<\eta_1<\eta_2<\ldots
<\eta_{N_2}\ .\label{zeros1}
\end{equation}
At least one of the integers $N_1, N_2$ is stictly positive; we will assume $N_2>0$ here, if
$N_2=0$, the case $N_1>0$ can be handled similarly. We will
construct a path $V_{\vec t(s)}$ which connects $V_{\vec t(1)}$ to an
external field $\tilde{V}=V_{\vec t(0)}$ which satisfies the condition
$n_1+n_2<N_1+N_2$. The composition of this path with the
 path connecting $V_{\vec t(0)}$ to the Gaussian
$z^2/2$ (which exists because of the induction hypothesis), completes the construction of the path required for the
proof of Theorem \ref{theorem: 1}.

\medskip

Let us deform $h(z)$ to $h(z;s)$ for $s\in[0,1]$ in such a way that
$h(z;1)=h(z)$, and $h(z;0)$ has only $N_1+N_2-2$ real zeros of odd
multiplicity. Define
\begin{equation}
h(z;s)=(z-\eta_1(s))\prod_{j=1}^{N_1}(z-\gamma_j).
\prod_{j=2}^{N_2}(z-\eta_j)\,\prod_{j=1}^{N_3}(z-z_j)(z-\bar z_j),
\label{h20}
\end{equation}
where we deform only one zero $\eta_1$ as follows:
\begin{equation}
\eta_1(s)=s\eta_1+(1-s)\eta_2.
\end{equation}
We then clearly have $h(z;1)=h(z)$, and $h(z;0)$ has only
$N_1+N_2-2$ zeros of odd multiplicity, since $\eta_1$ merged with
$\eta_2$, which has now become a zero of even multiplicity. Let us
consider the measure
\[d\mu_s(x)=\frac{1}{ic(s)}\mathcal{R}_+^{1/2}(x) \,h(x;s)dx, \qquad x\in[a,b],\]
where $c(s)$ is chosen such that $\mu_s$ is a probability measure.
The deformation of $\mu$ induces as before a deformation of $V$ to
$V_{\vec t(s)}$ by (\ref{V}), which preserves the degree of $V$ since the degree of $h$ is unchanged. What remains to prove, is that
$\mu_s$ is equal to the equilibrium measure in external field
$V_{\vec t(s)}$. To that end, by Lemma \ref{lemma: m}, we need to
prove only that (\ref{sum1})-(\ref{sum2}) hold for $0<s<1$.

\medskip
Recall the definition (\ref{mj})-(\ref{tilde mj}) of $m_j$ and
$\tilde m_j$, which now depend on $s$. We have
\begin{align}
&m_j(s)=\int_{\eta_j}^{\eta_{j+1}}|\mathcal{R}^{1/2}(x)h(x;s)|dx=\int_{\eta_j}^{\eta_{j+1}}|\mathcal{R}^{1/2}(x)h(x)f(x;s)|dx,\\
&\tilde
m_j(s)=\int_{\gamma_{j+1}}^{\gamma_{j}}|\mathcal{R}^{1/2}(x)h(x)f(x;s)|dx,
\end{align}
where
\begin{equation}
f(x;s)=\frac{x-\eta_1(s)}{x-\eta_1}.
\end{equation}
Note that $f(x;s)$ is positive, increasing, and smaller than $1$ for
$x>\eta_1(s)$, and that $f(x;s)$ is bigger than $1$ and increasing
for $x<\eta_1$. As a consequence, we have the inequalities
$m_0(s)\geq m_0$ and
\begin{align}
&m_j(s)\geq m_jf(\eta_j;s), &&m_j(s)\leq m_jf(\eta_{j+1};s), &&&1\leq j\leq N_2,\\
&\tilde m_j(s)\geq \tilde m_jf(\gamma_{j+1};s), &&\tilde m_j(s)\leq
\tilde m_jf(\gamma_{j};s), &&&0\leq j\leq N_1.
\end{align}
This implies that
\begin{multline}
m_{2j-1}(s)-m_{2j}(s)\leq
f(\eta_{2j};s)(m_{2j-1}-m_{2j})\\<(m_{2j-1}-m_{2j}),\qquad
j=1,\ldots, \frac{n_2}{2}-1,
\end{multline}
and $m_{k-1}(s)\leq m_{k-1}$ for $k>0$ even. We obtain
\begin{align*}
&m_0(s)-m_1(s)+m_2(s)- \ldots - m_{k-1}(s)\\
&\qquad \geq  m_0 -(m_1-m_2)-(m_3-m_4)-\ldots - (m_{k-3}-m_{k-2}) - m_{k-1}\\
&\qquad = \sum_{j=0}^{k-1}(-1)^jm_j>0,
\end{align*}
by Lemma \ref{lemma: m}. Similarly,
\begin{align*}
&\tilde m_0(s)-\tilde m_1(s)+\tilde m_2(s)- \ldots - \tilde m_{\ell}(s)\\
&\qquad\geq f(\gamma_1;s)(\tilde m_0-\tilde m_1) + f(\gamma_{3};s)(\tilde m_2-\tilde m_3)+\ldots + f(\gamma_{\ell};s)
(\tilde m_{\ell-1}-\tilde m_{\ell})\\
&\qquad \geq (\tilde m_0-\tilde m_1) + (\tilde m_2-\tilde m_3)+\ldots + (\tilde m_{\ell-1}-\tilde m_{\ell})\\
&\qquad =\sum_{j=0}^{\ell}(-1)^j\tilde m_j >0.
\end{align*}
This proves (\ref{sum1})-(\ref{sum2}) for $0<s<1$ and completes the proof of
Theorem \ref{theorem: 1}.

\section{Deformation of two-cut regular $V$: proof of Theorem \ref{theorem: 2}}\label{section: 2-cut}
We proceed in a similar way as for the proof in the one-cut case.
Assume now that $V=V_{\vec t}\in P_m^{(2)}$ is $2$-cut regular. The
equilibrium measure then has the form \cite{DKM}\begin{equation}
d\mu_V(x)=\frac{1}{ic}\mathcal{R}_+^{1/2}(x)h(x)dx, \qquad x\in [a_1,a_2]\cup
[a_3,a_4],
\end{equation}
with
\begin{equation}\mathcal{R}(z)=(z-a_1)(z-a_2)(z-a_3)(z-a_4),\qquad a_1<a_2<a_3<a_4,\end{equation}
where $\mathcal R^{1/2}(z)$ analytic in $\mathbb C\setminus([a_1,a_2]\cup[a_3,a_4])$ and positive for $z>a_4$.
The function $h$ is a monic polynomial of degree ${\rm deg\,}V-3$, which can be written in terms of its zeros as follows, \begin{equation}
h(z)=\prod_{j=1}^{n_1}(z-\gamma_j). \prod_{j=1}^{n_2}(z-\xi_j). \prod_{j=1}^{n_3}(z-\eta_j).\prod_{j=1}^{n_4}(z-z_j)(z-\bar z_j),\label{h2cut}
\end{equation}
with $\Im z_j\geq 0, z_j\notin[a_1,b_1]\cup[a_2,b_2]$ and
\begin{multline}
\gamma_{n_1}<\ldots <\gamma_2<\gamma_1<\gamma_0=a_1<a_2=\xi_0<\xi_1<\ldots <\xi_{n_2}<\xi_{n_2+1}=a_3\\
<a_4=\eta_0<\eta_1<\eta_2<\ldots <\eta_{n_3},\label{zeros2cut}\end{multline} with
$n_1,n_3$ even, $n_2$ odd. In other words, the $\gamma_j$'s are the real zeros of odd multiplicity at the left of $a_1$, the $\xi_j$'s are the zeros of odd multiplicity in the gap $(a_2,a_3)$, and the $\eta_j$'s are the real zeros of odd multiplicity at the right of $a_4$. Again, it is possible that some of the
$z_j$'s are real and coincide with each other or with one of the
real zeros $\gamma_j,\xi_j,\eta_j$. Define
\begin{align}
&\label{mj2}m_j=\int_{\eta_j}^{\eta_{j+1}}|\mathcal{R}^{1/2}(x)h(x)|dx,&j=0,\ldots,
n_3-1;\\
&\label{tilde mj2}\tilde
m_j=\int_{\gamma_{j+1}}^{\gamma_{j}}|\mathcal{R}^{1/2}(x)h(x)|dx,&\qquad
j=0,\ldots, n_1-1,\\&\label{hat mj} \widehat
m_j=\int_{\xi_{j}}^{\xi_{j+1}}|\mathcal{R}^{1/2}(x)h(x)|dx,& j=0,\ldots, n_2.
\end{align}
%

\begin{lemma}\label{lemma: m2}
Let $h$ be of the form (\ref{h2cut})-(\ref{zeros2cut}) with $n_1, n_3$ even and $n_2$ odd, and let
$\mu$ be defined by
 \begin{equation}\label{mu2}
\begin{split}
&d\mu(x)=\frac{1}{ic}\mathcal{R}_+^{1/2}(x)h(x)dx,\quad
x\in[a_1,a_2]\cup[a_3,a_4],\\
& c=\int_{J}|\mathcal{R}(x)^{1/2}\,h(x)|dx,\qquad J=[a_1,a_2]\cup[a_3,a_4].
\end{split}
\end{equation}
Let $V$ be given by
\begin{equation}
\label{V2}
V'(z)=\frac{1}{ic}\int_{\gamma}\frac{\mathcal{R}^{1/2}(y)h(y)}{z-y}dy,\qquad
z\in \mathbb C\backslash J,
\end{equation}
with $\gamma$ a clock-wise oriented contour surrounding $[a_1,a_2]\cup[a_3,a_4]$ and $z$. The constant of integration is fixed by $V(0)=0$.
Then, $V$ is two-cut regular ($V\in P_m^{(2)}$) and $\mu$ is equal to the equilibrium measure $\mu_V$ if and only if
\begin{align}
&\label{sum1-2}\sum_{j=0}^{k-1} (-1)^jm_j>0,&k=2,4,6,\ldots,n_1,\\
&\label{sum2-2}\sum_{j=0}^{k-1} (-1)^j\tilde m_j>0,&k=2,4,6,\ldots,n_3,\\
&\label{sum3-2}\sum_{j=0}^{k-1} (-1)^j\widehat m_j>0,&k=2,4,6,\ldots,n_2-1,\\
&\label{sum4-2}\sum_{j=0}^{n_2} (-1)^j\widehat m_j=0.
\end{align}
\end{lemma}
\begin{proof}
Let $\mu$ be of the form (\ref{mu2}), and $V$ as in (\ref{V2}). In a
similar way as in the proof of Lemma \ref{lemma: m},  one can show that the variational conditions (\ref{var eq})-(\ref{var ineq}) are equivalent to (see \cite{DKMVZ2})
\begin{align}
\label{var H1}
&H_{4}(x):=\frac{1}{c} \int^x_{a_4}
\mathcal{R}^{1/2}(\xi)h(\xi)d\xi>0,&& x>a_4,\\
\label{var H2}
&H_{2}(x):=\frac{1}{c} \int^x_{a_2}
\mathcal{R}^{1/2}(\xi)h(\xi)d\xi< 0,&&a_2<x<a_3,\\
&\label{var H2b}
H_2(a_3)=0,\\
\label{var H3}
&H_{1}(x):=\frac{1}{c} \int_x^{a_1}
\mathcal{R}^{1/2}(\xi)h(\xi)d\xi>0,&& x<a_1.
\end{align}
%
%
\begin{itemize}
\item[(i)] Assume that $\mu=\mu_V$, $V\in P_m^{(2)}$.  Then by 
 (\ref{var H1})-(\ref{var H3}) we have
 
\begin{align}
&\label{H3}0<\frac{1}{c} \int_{a_4}^{\eta_j} \mathcal{R}^{1/2}(\xi)h(\xi)d\xi=\frac{1}{c}(m_0-m_1+\ldots - m_j),\\
&\label{H4}0>\frac{1}{c} \int_{a_2}^{\xi_j} \mathcal{R}^{1/2}(\xi)h(\xi)d\xi=-\frac{1}{c}(\widehat m_0-\widehat m_1+\ldots - \widehat
m_j),\\
&\label{H5}0<\frac{1}{c} \int^{a_1}_{\gamma_j} \mathcal{R}^{1/2}(\xi)h(\xi)d\xi=\frac{1}{c}(\tilde
m_0-\tilde m_1+\ldots - \tilde m_j).
\end{align}
This implies (\ref{sum1-2})-(\ref{sum3-2}). By (\ref{var H2b}),
we also have (\ref{sum4-2}).
\item[(ii)] Conversely, assume that (\ref{sum1-2})-(\ref{sum4-2}) hold.
We need to prove the strict inequality  (\ref{var ineq})  for $x\in\mathbb R\setminus
\left([a_1,a_2]\cup[a_3,a_4]\right)$  which is equivalent to
$H_4(x)>0$ for  $x>a_4$, $H_2(x)<0$ for $x\in(a_2,a_3)$ and $H_1(x)>0$ for $x<a_1$.
From (\ref{H3}) we know that $H_4(x)$ is increasing for  $\eta_{2j}<x<\eta_{2j+1}$
and decreasing for $\eta_{2j+1}<x<\eta_{2j+2}$. Since $H_4(\eta_j)>0$ for all $j=1,\dots,n_1$ it follows that 
$H_4(x)$ is strictly positive for $x>a_4$.
The same reasoning can be applied to $H_1(x)$ and $H_2(x)$.
 In order to prove the variational equality  (\ref{var eq})  we observe that integrating the second equality in  (\ref{V2})   between $a_1$ and $x\in[a_1,a_2]$  we have 
 \[
 V(x)-V(a_1)=2\int_J\log|x-s|\psi(s)ds-2\int_J\log|a_1-s|\psi(s)ds,
 \]
 which shows that $V(x)-2\int_J\log|x-s|\psi(s)ds$ is constant on $[a_1,a_2]$. In the same way we can show that $V(x)-2\int_J\log|x-s|\psi(s)ds$ is constant on $[a_3,a_4]$.
 The condition (\ref{sum4-2}) guarantees that the constant is the same on the two intervals.
\end{itemize}
%
\end{proof}

We will prove Theorem \ref{theorem: 2} by induction on $n_1+n_2+n_3$.
Since $n_2$ is odd, we have $n_1+n_2+n_3\geq 1$.

\subsection{Basis $n_1+n_2+n_3=1$} If $n_1+n_2+n_3=1$,
$h$ has only one real odd multiplicity zero $\xi$, which has to lie
in $(a_2,a_3)$ since $n_2$ is odd and $n_1,n_3$ are even. We have
\begin{equation}
h(z)=(z-\xi)\prod_{j=1}^{n_4}(z-z_j)(z-\bar z_j).
\end{equation}
As before, we will deform the equilibrium measure $\mu$ to the
measure (\ref{mu}) in the case $s=1$, $r=4$, which will induce a deformation of $V(z)$ to the
quartic potential $z^4-4z^2$. Similarly as in the one-cut case, we
do this by sending all zeros $z_j$ of $h$ to $\infty$, but as
opposed to the one-cut case, where we were free to deform the zeros,
here we need to make sure that the deformation preserves the
constraint $\widehat m_0-\widehat m_1=0$. Otherwise, the deformed measure $\mu$
will not be an equilibrium measure.

\medskip

For $s\in[0,1]$, we define
\begin{equation}
\mathcal{R}(z;s)=(z-a_1(s))(z-a_2(s))(z-a_3(s))(z-a_4(s)),
\end{equation}
with
\begin{align}
&a_1(s)=sa_1-\sqrt{3}(1-s), &a_2(s)=sa_2-(1-s),\\
&a_3(s)=sa_3+(1-s), & a_4(s)=sa_4+\sqrt{3}(1-s),
\end{align}
and
\begin{equation}
h(z;s)=(z-\xi(s))\prod_{j=1}^{n_3}(z-z_j(s))(z-\bar z_j(s)),\qquad
z_j(s)=z_j+i\frac{1-s}{s},
\end{equation}
where $\xi(s)\in (a_2,a_3)$ is chosen in such a way that $\widehat m_0-\widehat m_1=0$. Note that there is a
unique value of $\xi(s)$ such that this is true, and that it depends continuously on $s$.
We let
\begin{equation}
d\mu_s(x)=\frac{1}{ic(s)}\mathcal{R}_+^{1/2}(x;s)\,h(x;s)dx, \qquad x\in[a_1,a_2]\cup[a_3,a_4],
\end{equation}
where $c(s)$ is chosen such that $\mu_s$ has mass $1$.

For all $s\in(0,1)$, $h$ has only one real zero and thus satisfies
automatically the conditions (\ref{sum1-2})-(\ref{sum3-2}). The
fourth condition (\ref{sum4-2}) holds by definition of $\xi(s)$.
When taking the limit $s\to 0$, the non-real zeros $z_j(s)$ tend to
infinity, and $d\mu_s$ converges to
\begin{equation}
d\mu_0(x)=\frac{2}{\pi}|x|\sqrt{(3-x^2)(x^2-1)}dx,\qquad
x\in[-\sqrt{3},-1]\cup[1,\sqrt{3}].
\end{equation}
By (\ref{V2}), we have that $\mu_s$ is the equilibrium measure
corresponding to external field \[ V_{\vec
t(s)}'(z)=\frac{1}{ic(s)}\int_{\gamma}\frac{\mathcal{R}^{1/2}(y,s)h(y,s)}{z-y}dy,\qquad
z\in \mathbb C\setminus J(s),\] for all $s\in[0,1]$. Since the
dependence of $V$ in (\ref{V2}) on the coefficients of $h$ and $\mathcal R$
is continuous, $\vec t(s)$ describes a continuous path. In addition,
we have $V_{\vec t(0)}(z)=z^4-4z^2$, see (\ref{def V})-(\ref{mu}), and the degree of $V_{\vec t(s)}$ is the same as the degree of $V$ for all $0<s<1$.

\subsection{Inductive step} We now suppose that Theorem
\ref{theorem: 1} holds for any $h$ of the form
(\ref{h})-(\ref{zeros}) with $n_1+n_2+n_3<N_1+N_2+N_3$, i.e.\ if $h$
has strictly less than $N_1+N_2+N_3$ real zeros of odd multiplicity,
with $N_1+N_2+N_3>1$. Let $V=V_{\vec t(1)}\in P_m^{(1)}$ be such that
$h$ has $N_1+N_2+N_3$ real zeros of odd multiplicity.  Our strategy
is as before to deform $V_{\vec t(1)}$ by means of a deformation of
$\mu$, to an external field $V_{\vec t(s_0)}$ for which $h$ has less
than $N_1+N_2+N_3$ zeros. Afterwards, in order to complete the
proof, the path $V_{\vec t(s)}$, $s\in[s_0,1]$, has to be composed
with the path connecting $V_{\vec t(s_0)}$ to the quartic external
field $V_{\vec t(0)}=z^4-4z^2$. The latter exists because of the induction
hypothesis. There are two cases which we will treat separately:
$N_2>1$ and $N_2=1$.

\subsubsection*{Case 1: $N_2>1$}

First consider the case where $N_2>1$. Define a deformation of $h$
as follows:
\begin{equation}
h(z)=(z-\xi_1(s))(z-\xi_{N_2}(s))\prod_{j=1}^{N_1}(z-\gamma_j).
\prod_{j=2}^{N_2-1}(z-\xi_j).\prod_{j=1}^{N_3}(z-\eta_j)
.\prod_{j=1}^{N_4}(z-z_j)(z-\bar z_j),\label{h21}
\end{equation}
where \begin{equation}\label{xi1}\xi_1(s)=\xi_1+(1-s), \end{equation}
and where we choose $\xi_{N_2}(s)$ in such a way that (\ref{sum4-2})
holds.

One verifies that $\xi_{N_2}(s)$ is uniquely defined and smooth for $s$ close to $1$. We claim that in addition $\xi_{N_2}(s)$ is increasing in $s$.

\begin{varproof}{\bf of the claim:} By
definition of $\xi_{N_2}(s)$, we have $\sum_{j=0}^{n_2} (-1)^j\widehat
m_j(s)=0$. Since the $m_j$'s depend on $\xi_1, \xi_{N_2}$ by
(\ref{hat mj}), we have
\begin{multline}\label{chain}0=\frac{\partial}{\partial s}\sum_{j=0}^{n_2} (-1)^j\widehat
m_j(s)\\=\frac{\partial \xi_1(s)}{\partial
s}\frac{\partial}{\partial \xi_1}\sum_{j=0}^{n_2} (-1)^j\widehat
m_j(\xi_1,\xi_{N_2}) + \frac{\partial \xi_{N_2}(s)}{\partial
s}\frac{\partial}{\partial \xi_{N_2}}\sum_{j=0}^{n_2} (-1)^j\widehat
m_j(\xi_1,\xi_{N_2}).\end{multline} We write $\widehat m_j(\xi_1,\xi_{N_2})$
here since all the other $\xi_j$'s are constant. If we prove that
\begin{align*}
&\frac{\partial}{\partial \xi_1}\sum_{j=0}^{n_2} (-1)^j \widehat
m_j(\xi_1,\xi_{N_2})>0,\\
&\frac{\partial}{\partial \xi_{N_2}}\sum_{j=0}^{n_2} (-1)^j\widehat
m_j(\xi_1,\xi_{N_2})>0,
\end{align*} it follows from $\frac{\partial
\xi_1(s)}{\partial s}<0$  that $\frac{\partial
\xi_{N_2}(s)}{\partial s}>0$.

\medskip

For $\xi_1'>\xi_1$, we have
\[\widehat
m_j(\xi_1',\xi_{N_2})=
\int_{\xi_{j}}^{\xi_{j+1}}|\mathcal{R}^{1/2}(x)h(x)F(x;\xi_1',\xi_1)|dx,
\]
with
\[F(x;\xi_1',\xi_1)=\frac{x-\xi_1'}{x-\xi_1}.\]
It follows immediately that
\[\widehat m_0(\xi_1',\xi_{N_2})>\widehat m_0(\xi_1,\xi_{N_2}),\qquad \widehat m_{n_2}(\xi_1',\xi_{N_2})<\widehat m_{n_2}(\xi_1,\xi_{N_2}),\]
and, since $F$ is increasing and smaller than $1$ for $x>\xi_1$, we
also have
\begin{multline}\widehat m_{2j-1}(\xi_1',\xi_{N_2})-\widehat m_{2j}(\xi_1',\xi_{N_2})\leq
F(\xi_{2j})(\widehat m_{2j-1}(\xi_1,\xi_{N_2})-\widehat m_{2j}(\xi_1,\xi_{N_2}))\\
<\widehat m_{2j-1}(\xi_1,\xi_{N_2})-\widehat
m_{2j}(\xi_1,\xi_{N_2})).\end{multline} Consequently,
\[\sum_{j=0}^{n_2} (-1)^j
\widehat m_j(\xi_1',\xi_{N_2})>\sum_{j=0}^{n_2} (-1)^j \widehat
m_j(\xi_1,\xi_{N_2}),\] and
\[\frac{\partial}{\partial \xi_1}\sum_{j=0}^{n_2} (-1)^j
\widehat m_j(\xi_1,\xi_{N_2})>0.\] Similarly, one obtains
\[\frac{\partial}{\partial \xi_{N_2}}\sum_{j=0}^{n_2} (-1)^j\widehat
m_j(\xi_1,\xi_{N_2})>0,\] and this completes the proof of the claim.
\end{varproof}

Since $\xi_1(s)$ increases at constant speed as $s$ decreases, there will be a point
$s_0>0$ where two $\xi_j$'s coincide: either $\xi_1(s_0)=\xi_2(s_0)$, or $\xi_{N_2}(s_0)=\xi_{N_2-1}(s_0)$. We define
$s_0$ as the first point where one of these collisions occurs.
Setting
\begin{equation}
d\mu_s(x)=\frac{1}{c(s)}|\mathcal{R}^{1/2}(x)h(x;s)|dx,\qquad
x\in[a_1,a_2]\cup[a_3,a_4],
\end{equation}
we have $\mu_1=\mu_V$ and $\mu_{s_0}$ corresponds to $h(.;s_0)$ with
at most $N_1+N_2+N_3-2$ real zeros of odd multiplicity. It remains to
prove that $\mu_s$ is an equilibrium measure for all $s\in [s_0,1]$.
As before, we will use Lemma \ref{lemma: m2} to prove this.

\medskip

We have
\begin{align}
&m_j(s)=\int_{\eta_j}^{\eta_{j+1}}|\mathcal{R}^{1/2}(x)h(x;s)|dx=\int_{\eta_j}^{\eta_{j+1}}|\mathcal{R}^{1/2}(x)h(x)f(x;s)|dx,\\
&\tilde
m_j(s)=\int_{\gamma_{j+1}}^{\gamma_{j}}|\mathcal{R}^{1/2}(x)h(x)f(x;s)|dx,\\
&\widehat m_j(s)=\int_{\xi_{j}}^{\xi_{j+1}}|\mathcal{R}^{1/2}(x)h(x)f(x;s)|dx
\end{align}
where
\begin{equation}
f(x;s)=\frac{x-\xi_1(s)}{x-\xi_1}\frac{x-\xi_{N_2}(s)}{x-\xi_{N_2}},\qquad
\xi_1<\xi_1(s)<\xi_{N_2}(s)<\xi_{N_2}.
\end{equation}
This function is bigger than $1$ and increasing for $x<\xi_1$, and
decreasing and bigger than $1$ for $x>\xi_2$.  This implies the
inequalities
\begin{align}
&m_j(s)\geq m_jf(\eta_{j+1};s), &&m_j(s)\leq m_jf(\eta_{j};s), &&&0\leq j\leq N_3-1,\\
&\tilde m_j(s)\geq \tilde m_jf(\gamma_{j+1};s), &&\tilde m_j(s)\leq
\tilde m_jf(\gamma_{j};s), &&&0\leq j\leq N_1-1.
\end{align}
For $k$ even, we have
\begin{align*}
&m_0(s)- m_1(s)+m_2(s)- \ldots - m_{k-1}(s)\\
&\qquad\geq f(\eta_1;s)( m_0- m_1) + f(\eta_{3};s)(m_2- m_3)+\ldots
+ f(\eta_{k-1};s)
(m_{k-2}- m_{k-1})\\
&\qquad \geq (m_0- m_1) + (m_2- m_3)+\ldots + (m_{k-2}- m_{k-1})\\
&\qquad =\sum_{j=0}^{k-1}(-1)^j m_j >0,
\end{align*}
by Lemma \ref{lemma: m2}, and
\begin{align*}
&\tilde m_0(s)-\tilde m_1(s)+\tilde m_2(s)- \ldots - \tilde m_{k-1}(s)\\
&\qquad\geq f(\gamma_1;s)(\tilde m_0-\tilde m_1) +
f(\gamma_{3};s)(\tilde m_2-\tilde m_3)+\ldots +
f(\gamma_{k-1};s)
(\tilde m_{k-2}-\tilde m_{k-1})\\
&\qquad \geq (\tilde m_0-\tilde m_1) + (\tilde m_2-\tilde m_3)+\ldots + (\tilde m_{k-2}-\tilde m_{k-1})\\
&\qquad =\sum_{j=0}^{k-1}(-1)^j\tilde m_j >0.
\end{align*}

For $\xi_1<x<\xi_{N_2}$, $f(x)<1$ and $f$ has a unique local maximum
$\xi^*$. Suppose $\xi^*$ lies in between $\xi_\ell$ and $\xi_{\ell+2}$
with $\ell$ odd. Then, consider for $1\leq d\leq \ell$ and $d$ odd
the alternating sum
\[\sum_{j=0}^{d}(-1)^j\widehat m_j(s)=\widehat m_0(s) - (\widehat m_1(s)-\widehat
m_2(s))-\ldots -(\widehat m_{d-2}(s)-\widehat m_{d -1}(s))-\widehat m_{d}(s).\] We
have
\begin{multline}
\widehat m_{2j+1}(s)-\widehat m_{2j+2}(s)\leq f(\xi_{2j+2};s)(\widehat
m_{2j+1}-\widehat m_{2j+2})\\<\widehat
m_{2j+1}-\widehat m_{2j+2},\qquad j\geq 0, \xi_{2j+2}\leq \xi^*.
\end{multline}Furthermore,
since $f$ is bigger than $1$ outside $[\xi_1,\xi_{N_2}]$ and smaller
inside, $\widehat m_0(s)\geq \widehat m_0$ and $\widehat m_{d}(s)\leq \widehat
m_{d}$. Combining those inequalities, we obtain
\[\sum_{j=0}^{d}(-1)^j\widehat m_j(s)\geq \sum_{j=0}^{d}(-1)^j\widehat m_j>0.\]
For $d$ odd and $\ell+2\leq d\leq n_2-1$, we have by
(\ref{sum4-2}),
\[\sum_{j=0}^{d}(-1)^j\widehat m_j(s)= \sum_{j=d +1}^{n_2}(-1)^{j+1}\widehat m_j(s).\]
But \begin{eqnarray*}\sum_{j=d +1}^{n_2}(-1)^{j+1}\widehat
m_j(s)&=&\widehat m_{n_2}(s) - (\widehat m_{n_2-1}(s)-\widehat
m_{n_2-2}(s))-\ldots -
\widehat m_{d +1}(s)\\
&\geq &\widehat m_{n_2} - (\widehat m_{n_2-1}-\widehat m_{n_2-2})-\ldots - \widehat
m_{d +1}\\
&=&\sum_{j=0}^{d}(-1)^j\widehat m_j>0.
\end{eqnarray*}

\subsubsection*{Case 2: $N_2=1$}

We assume $N_1+N_2+N_3>1$ and $N_2=1$. Consequently, either $N_1$ or
$N_3$ is strictly positive. We assume $N_3>0$, the other case is
similar. Define a deformation of $h$ as follows:
\begin{equation}
h(z)=(z-\xi_1(s))(z-\eta_{1}(s))\prod_{j=1}^{N_1}(z-\gamma_j).\prod_{j=2}^{N_3}(z-\eta_j)
.\prod_{j=1}^{N_4}(z-z_j)(z-\bar z_j),\label{h3}
\end{equation}
where \begin{equation}\label{xi2}\eta_1(s)=s\eta_1+(1-s)\eta_2,\quad s\in[0,1],
\end{equation} and where we choose $\xi_1(s)$ to be the unique value
in $(a_2,a_3)$ such that (\ref{sum4-2}) holds. Similarly as
in case 1, one proves that $\xi_1(s)\geq \xi_1$ for $s\in[0,1]$.
We will prove that
\begin{equation}
d\mu_s(x)=\frac{1}{ic(s)}\mathcal{R}_+^{1/2}(x)h(x;s)dx,\qquad
x\in[a_1,a_2]\cup[a_3,a_4],
\end{equation}
is an equilibrium measure for any $s\in[0,1]$, by Lemma \ref{lemma: m2}.

\medskip

Note first that equation (\ref{sum4-2}) holds for any $s\in [0,1]$ by definition of $\xi_1(s)$, and that (\ref{sum3-2}) is trivially satisfied since $N_2=1$.
We still need to prove (\ref{sum1-2})-(\ref{sum2-2}) for $s\in[0,1]$.
Note that\begin{align}
&m_j(s)=\int_{\eta_j}^{\eta_{j+1}}|\mathcal{R}^{1/2}(x)h(x;s)|dx=\int_{\eta_j}^{\eta_{j+1}}|\mathcal{R}^{1/2}(x)h(x)f(x;s)|dx,\\
&\tilde
m_j(s)=\int_{\gamma_{j+1}}^{\gamma_{j}}|\mathcal{R}^{1/2}(x)h(x)f(x;s)|dx,
\end{align}
with
\begin{equation}
f(x;s)=\frac{x-\eta_1(s)}{x-\eta_1}\frac{x-\xi_1(s)}{x-\xi_1}.
\end{equation}
First, $f$ is increasing and bigger than $1$ for $x<\xi_1$, which means that, for $k$ even,
\begin{multline}
\sum_{j=0}^{k-1}(-1^j)\tilde m_j(s)\geq f(\gamma_{1})(\tilde m_0-\tilde m_1)+f(\gamma_{3})(\tilde m_2-\tilde m_3)+\ldots + f(\gamma_{k-1})(\tilde m_{k-2}-\tilde m_{k-1})\\
\geq \sum_{j=0}^{k-1}(-1^j)\tilde m_j>0,
\end{multline}
where $\tilde m_j=\tilde m_j(1)$.
Note that $\widehat m_0(s)\geq \widehat m_0$ for $s\in[0,1]$, which implies by (\ref{sum4-2}) that $\widehat m_1(s)\geq \widehat m_1$. But $\widehat m_1(s)\leq f(a_3;s)\widehat m_1$, which implies that $f(a_3;s)>1$. We thus have that $f$ is increasing and bigger than $1$ for $a_3<x<\eta_1$, which implies that $m_0(s)> m_0$ for $0<s<1$. Now, since for $x>\eta_1$, $f$ is increasing and smaller than $1$, we have
\begin{multline}
\sum_{j=0}^{k-1}(-1^j)m_j(s)\geq m_0-f(\eta_{2})(m_1-m_2)-f(\eta_{4})(m_3-m_4)-\ldots - f(\eta_{k}) m_{k-1}\\
 \geq \sum_{j=0}^{k-1}(-1^j) m_j>0.
\end{multline}
This completes the proof of Theorem \ref{theorem: 2}.


\section{Time-derivatives of $\log Z_{N}$ in the two-cut regime}

In this section, we will extend the asymptotic results for the partition function $Z_N=Z_N(V)$ to general two-cut regular polynomial external fields $V$.  In order to accomplish this, we will first establish  formula  (\ref{diff id t3}) for a logarithmic derivative (with respect to the parameters $\{t_{k}\}$) of the partition function, in terms of the solution $X$ to the RH problem for orthogonal polynomials on the full real line.  The formula is valid for any potential, whether it is symmetric or not (or even regardless of the number of intervals in the support of the equilibrium measure).  The asymptotic behavior of the solution $X$ as described in Section \ref{sec:LargeNAS} is well understood.  The real challenge is to obtain formulae that are simple and intrinsic, and this is aim of Sections \ref{sec:ellipRSP} through \ref{sec:FinalProof}.

 Let 
\[
P_j(z)=\kappa_jz^j+\ldots,\quad j=0,1,2,\ldots,\]
 be the orthonormal polynomials with respect to the weight $e^{-NV_{\vec t}}$ on $\mathbb R$, where, as before,
\begin{equation}\label{V0tdef}
V(z)=V_{\vec t}(z)=V_0(z)+\sum_{j=1}^{2d}t_jz^j,\qquad V_0(z)=z^4-4z^2.
\end{equation}
These polynomials satisfy a three term recurrence relation of the form
\[
zP_k(z)=a_kP_{k+1}(z)+b_kP_k(z)+a_{k-1}P_{k-1}(z).
\]
Recall that $Z_N(V)$ is given by (\ref{def Zn}).

\subsection{Differential identity}
For $z\in\mathbb{C}\setminus\mathbb R$, define the matrix-valued function
\begin{equation}\label{X2}
X(z)=X^{(N)}(z;\vec t)=\begin{pmatrix} \kappa_N^{-1}P_N(z)&
\frac{1}{2\pi i}\kappa_N^{-1}\int_{\mathbb R}P_N(s)\frac{e^{-NV_{\vec t}(s)}ds}{s-z}\\
-2\pi i\kappa_{N-1}P_{N-1}(z)&
-\kappa_{N-1}\int_{\mathbb R}P_{N-1}(s)\frac{e^{-NV_{\vec t}(s)}ds}{s-z}
\end{pmatrix},
\end{equation}
which solves the standard RH problem for orthogonal polynomials \cite{FIK}. In particular $X$ satisfies the jump relation
\begin{equation}
\label{jumpX}
 X_+(x)=X_-(x)\begin{pmatrix}1&e^{-NV_{\vec t}(x)}\\0&1\end{pmatrix},\quad \mbox{ for $x\in\mathbb R$}.
 \end{equation}
Following \cite{BEH, MezzadriMo}, we can derive an identity for
$\frac{\partial}{\partial t_k}\log Z_N$. Therefore, recall the derivation of (\ref{diff id alpha}) in Section \ref{section: diff id}. In exactly the same way as we did there for the $\alpha$-derivative, we can prove that
\begin{multline}\label{diff id t}
\frac{d}{d t_k}\log Z_{N}(\vec t)
=-\int_\mathbb R\frac{\partial}{\partial t_k}\left(a_{N-1}(\vec t)P_{N}'(x;\vec t)P_{N-1}(x;\vec t)\right)e^{-NV_{\vec t}(x)}dx
\\ \qquad +\int_\mathbb R\frac{\partial}{\partial t_k}\left(a_{N-1}(\vec t)P_{N}(x;\vec t)P_{N-1}'(x;\vec t)\right)
e^{-NV_{\vec t}(x)}dx,
\end{multline}
where $'$ denotes derivative with respect to $x$.
Integrating by parts, we obtain
\begin{multline}\label{diff id t2}
\frac{d}{d t_k}\log Z_{N}(\vec t)
=-N\int_\mathbb R a_{N-1}(\vec t)\left(P_{N}'(x;\vec t)P_{N-1}(x;\vec t)\right.
\\ \qquad \left. -P_{N}(x;\vec t)P_{N-1}'(x;\vec t)\right)
x^k e^{-NV_{\vec t}(x)}dx.
\end{multline}
By (\ref{X2}) and the jump relation (\ref{jumpX}), it is straightforward to see that
\begin{eqnarray}\label{diff id t3}
\frac{d}{d t_k}\log Z_{N}(\vec t)
&=&-\frac{N}{2\pi i}\int_\mathbb R \left(X_\pm^{-1}X_\pm'\right)_{21}(x)
x^k e^{-NV_{\vec t}(x)}dx\\
&=&\frac{N}{4\pi i}\int_\mathbb R \left[\Tr\left(X_+^{-1}X_+'\sigma_3\right)(x)-\Tr\left(X_-^{-1}X_-'\sigma_3\right)(x)\right]
x^k dx,\nonumber
\end{eqnarray}
where $'$ denotes derivative with respect to $x$.
Using a contour deformation argument, we can write the integral over the real line as the limit of the integral of $\Tr\left(X^{-1}(z)X'(z)\sigma_3\right)$ over a large clockwise oriented circle in the complex plane, as the radius tends to infinity. 
Although $z^k\Tr\left(X^{-1}(z)X'(z)\sigma_3\right)$ is not analytic at $\infty$, it has a large $z$ expansion
given by
\begin{equation}
z^k\Tr\left(X^{-1}(z)X'(z)\sigma_3\right)=\sum_{j=-1}^{k-2}c_jz^{j}+\bigO(z^{-2}),\qquad z\to\infty,
\end{equation}
for any $k\in\mathbb N$, where $-c_{-1}$ is the residue at infinity of $z^k\Tr\left(X^{-1}(z)X'(z)\sigma_3\right)$. By residue calculus, we obtain \cite{JMU}
\begin{equation}\label{diff id t res}
\frac{d}{d t_k}\log Z_{N}(\vec t)=\frac{N}{2}\Res_{z=\infty}\left[\Tr\left(X^{-1}(z)X'(z)\sigma_3\right)z^k\right].
\end{equation}
This identity is not new; similar identities can be found for example in \cite{BEH, JMU, MezzadriMo}.

\subsection{Large $N$ asymptotics for $X$}
\label{sec:LargeNAS}

We will now derive large $N$ asymptotics for the right hand side of (\ref{diff id t res}), and integrate it afterwards along a two-cut regular path in the $\vec t$-space as constructed in Section \ref{section: 2-cut}. Large $N$ asymptotics for the matrix  $X(z)$ have been obtained in \cite{DKMVZ2, DKMVZ1} for general $k$-cut regular polynomials $V=V_{\vec t}$. In this section, we recall the results from \cite{DKMVZ2, DKMVZ1} in the two-cut regular case, where the equilibrium measure $\mu_V$ is supported on $[a_1,a_2]\cup[a_3,a_4]$ with $a_1<a_2<a_3<a_4$. We write 
\begin{equation}
\label{Omega}
\Omega=\int_{a_3}^{a_4}d\mu_V(x),
\end{equation}
i.e.\ $\Omega$ is the expected portion of random matrix eigenvalues in the rightmost interval of the support as $N\to\infty$,
 and
\begin{equation}
g(z)=\int\log(z-s)d\mu_V(s).
\end{equation}
The equilibrium measure $\mu_V$ can be written in the form \cite{DKM}
\begin{equation}
\label{psi}
d\mu_V(x)=\psi(x)dx={\mathcal{R}(x)_+^{1/2} \over   \pi i }h (x)dx,\qquad \mbox {for  }x\in[a_1,a_2]\cup[a_3,a_4],\end{equation}
with $h$ given by
\begin{equation}\label{h11} h(z)=-\dfrac{1 }{2\pi i } \PVint_{J} {V'(\eta) \over
\mathcal{R}(\eta)_+^{1/2} (\eta -z)} d \eta,\quad J=[a_1,a_2]\cup[a_3,a_4],
\end{equation}
where $V(z)=V_{\vec t}(z)$ is the polynomial external field and
\begin{equation}\label{R}
\mathcal{R}(z)=\prod_{j=1}^4(z-a_j).
\end{equation}
The square root  $\mathcal{R}(z)^{1/2}$ is as usual chosen in such a way  that it is positive for $z>a_4$ and analytic in $\mathbb{C}\backslash 
[a_1,a_2]\cup[a_3,a_4]$. $\mathcal{R}_+(z)^{1/2}$ denotes the boundary value from the upper half plane on $[a_1,a_2]\cup[a_3,a_4]$.
\begin{remark}
There is a difference in notation between the function $h$ here and the one in Section \ref{section: 2-cut}: they differ by a constant factor. The function $\frac{1}{c}h$ from Section \ref{section: 2-cut}, where it was convenient to consider a monic polynomial $h$, is equal to $\frac{1}{\pi}h$ using the notation of this section, where we prefer to have $h$ defined by (\ref{h11}).
\end{remark}

The function  $h(z)$ is a polynomial, which also has the property of being 
 the polynomial part of the function 
$\frac{V'(z)}{\mathcal{R}(z)^{1/2}}$ at infinity, i.e.
\begin{equation}\label{equi12}
\frac{V'(z)}{\mathcal{R}(z)^{1/2}}=h(z)+\bigO(z^{-1}),\qquad z\to \infty.
\end{equation} 

The endpoints $a_1<a_2<a_3<a_4$ are  determined by 
\begin{equation}
\label{moment}
\int_\gamma\dfrac{V'(\lb)\lb^k}{\mathcal{R}(\lb)^{1/2}}d\lb=0,\quad k=0,1, \qquad -\dfrac{1}{4\pi i}\int_\gamma\dfrac{V'(\lb)\lb^{2}}{\mathcal{R}(\lb)^{1/2}}d\lb=1,
\end{equation}
where $\gamma$ is a clock-wise oriented contour surrounding the support $J$,
and by the normalization condition
\begin{equation}
\label{loop}
\int_{a_{2}}^{a_{3}} \mathcal{R}(\lb)^{1/2}h(\lb)d\lb=0.
\end{equation}
These $4$ equations determine  $a_1<a_2<a_3< a_4$ uniquely as a function of $t_1,\dots, t_{2d}$.

For $z$ outside small disks surrounding the endpoints $a_1,\ldots, a_4$, the matrix $X(z) $ defined in (\ref{X2})  can be expressed as \cite{DKMVZ2}
\begin{equation} \label{as X}
X(z)=R(z)P^{\infty}(z)e^{Ng(z)\sigma_3}.
\end{equation}
Here:
\begin{itemize}
\item[(1)] $R(z)$ is a $2\times 2$ matrix  that solves the  so-called error RH problem and admits an asymptotic expansion of the form
\begin{equation}\label{as R 2-cut}
R(z)=I+\sum_{j=1}^kN^{-j}R^{(j)}(z)+\bigO(N^{-k-1}),\qquad N\to\infty,
\end{equation}
where $R^{(j)}(z)$, $j=1,2,\ldots $ can be computed by a recursive procedure and remains bounded as $N\to\infty$, and 
\item[(2)]$P^{\infty}(z)$ is the unique solution to the following RH problem:
\subsubsection*{RH problem for $P^{\infty}(z)$ }
\vspace{-0,5cm}
\begin{align}
\label{RHPT1}
&({\rm a}) \mbox{ $P^{\infty}(z)$ is a $2\times2$ matrix  analytic in $\mathbb C\setminus[a_1,a_4].$}\\
\nonumber
& ({\rm b})\mbox{  $P^\infty$ has the jumps}\\
\label{RHPT2}
&\quad\quad
P_{+}^{\infty}(z)=P_{-}^{\infty}(z)\pmtwo{0}{1}
{-1}{0},\qquad z\in(a_1,a_2)\cup(a_3,a_4),\\
\label{RHPT3}
&\quad\quad P_{+}^{\infty}(z)=P_{-}^{\infty}(z)e^{-2\pi iN\Omega\sigma_3},\qquad\quad z\in(a_2,a_3).\\
\label{RHPT4}
&({\rm c}) \mbox{ $P^{\infty}(z)=I+\bigO(z^{-1})$ as $z\to\infty$.}\\
\label{RHPT5}
&({\rm d}) \mbox{ As $z\to a_j$, $j=1,\ldots, 4$, we have $P^\infty(z)=\bigO(|z-a_j|^{-1/4})$.}
\end{align}
This RH problem can be solved using elliptic $\theta$-functions (see e.g. \cite{KK},\cite{DKI},\cite{DKMVZ1}), we will give more details about this construction later on.
\end{itemize}
Using (\ref{diff id t res}), (\ref{as X}), and the identity \cite{KM} 
\begin{equation}
\frac{\partial}{\partial t_k}F_0=-\Res_{z=\infty}\left(z^kg'(z)\right),
\end{equation} 
where 
\begin{equation}
\label{F0V}
F_0=\iint \log|x-y|^{-1}d\mu_V(x)d\mu_V(y)+\int V(x)d\mu_V(x),
\end{equation}
we obtain
\begin{multline}\label{identity BEH3}
\frac{\partial}{\partial t_k}\log Z_N=-N^2\frac{\partial}{\partial t_k}F_0+\frac{N}{2}\res[z=\infty]\Tr\left(P^\infty(z)^{-1}P^\infty(z)'\sigma_3z^k \right)\\
+\frac{N}{2}\res[z=\infty] \Tr\left(P^\infty(z)^{-1}R^{-1}(z)R'(z)P^\infty(z)\sigma_3z^k\right).
\end{multline}
Substituting (\ref{as R 2-cut}), we obtain large $N$ asymptotics for the right hand side of this identity:
\begin{multline}\label{identity BEH4}
\frac{\partial}{\partial t_k}\log Z_N=-N^2\frac{\partial}{\partial t_k}F_0+\frac{N}{2}\res[z=\infty] \Tr\left(P^\infty(z)^{-1}P^\infty(z)'\sigma_3z^k\right)\\
+\dfrac{1}{2}\res[z=\infty]\Tr\left(P^\infty(z)^{-1}\dfrac{dR^{(1)}(z)}{dz}P^\infty(z)\sigma_3z^k\right)
+\bigO(N^{-1}),\qquad N\to\infty.
\end{multline}
Here we used the fact that $P^\infty(z)$ is uniformly bounded for $N$ and $z$ large. This will follow from our construction of $P^\infty$ below.
We observe  that the second term at the right hand side of the above formula is invariant  under the transformation 
\[
P^{\infty}(z)\mapsto C P^{\infty}(z), 
\]
  where $C$ is an invertible matrix.  If in addition we apply the transformation
  \[
  R^{(1)}(z)\mapsto C  R^{(1)}(z)C^{-1},
\]
the third term on the right hand side of (\ref{identity BEH4}) is also unchanged.
Let 
$C=P^{\infty}(z_0)^{-1}$, and define 
\begin{equation}\label{Pinfbase}
P^{\infty}(z,z_0)\equiv  P^{\infty}(z_0)^{-1} P^{\infty}(z),\quad z_0\neq a_j,\;\;j=1,2,3,4.
\end{equation}
Now, $P^{\infty}(z,z_0)$ solves the RH conditions
(\ref{RHPT1}), (\ref{RHPT2}), (\ref{RHPT3}) and (\ref{RHPT5}), with the asymptotic condition (\ref{RHPT4}) replaced by
\begin{equation}
\label{RHPT4b}
P^{\infty}(z,z_0)= I+\bigO(z-z_0),\qquad \mbox{ as $z\to z_0$.}
\end{equation}
$P^{\infty}(z,z_0)$ is moreover uniquely determined by those conditions. It will be convenient for us to construct the solution to this more general RH problem for $P^\infty(z,z_0)$, normalized at $z_0$, instead of the particular case of $P^\infty(z)=P^\infty(z,\infty)$.
Writing
\begin{equation}\label{Rbase}
R(z,z_0)=  P^{\infty}(z_0)^{-1}R(z)P^{\infty}(z_0)       ,\qquad R^{(1)}(z,z_0)\equiv P^{\infty}(z_0)^{-1} R^{(1)}(z)P^{\infty}(z_0),
\end{equation}
\begin{multline}
\res[z=\infty]\Tr\left(P^\infty(z)^{-1}\dfrac{dR^{(1)}(z)}{dz}P^\infty(z)\sigma_3 z^k\right)\\=
\res[z=\infty]
 \mbox{Tr} \left(P^\infty(z,z_0)^{-1}\dfrac{dR^{(1)}(z,z_0)}{dz}P^{\infty}(z,z_0)\sigma_3 z^k\right).
\end{multline}
Since the left hand side in this equation is independent of $z_0$, we can let $z_0\to z$ at the right hand side, and obtain
\begin{multline}
\res[z=\infty]\Tr\left(P^\infty(z)^{-1}\dfrac{dR^{(1)}(z)}{dz}P^\infty(z)\sigma_3 z^k\right)\\=
\res[z=\infty]
 \mbox{Tr} \left(\left.\dfrac{dR^{(1)}(z,z_0)}{dz}\right|_{z_0=z}\sigma_3 z^k\right).
\end{multline}
Therefore, the relation (\ref{identity BEH4}) can be rewritten in the form
\begin{lemma}
\label{lemmaZ}
We have
\begin{multline}\label{identity BEH5}
\frac{\partial}{\partial t_k}\log Z_N(V_{\vec t})=-N^2\frac{\partial}{\partial t_k}F_0+\frac{N}{2}\res[z=\infty] \Tr\left(\left.\dfrac{dP^\infty(z,z_0)}{dz}\right|_{z_0=z}\sigma_3 z^k\right)\\
+\dfrac{1}{2}\res[z=\infty] 
 \Tr \left(\left.\dfrac{d R^{(1)}(z,z_0)}{dz}\right|_{z_0=z}\sigma_3 z^k\right)+\bigO(N^{-1}),\qquad N\to\infty,
\end{multline}
where $P^{\infty}(z,z_0)$  satisfies the RH conditions (\ref{RHPT1})-(\ref{RHPT3}), (\ref{RHPT5}), and the normalization (\ref{RHPT4b}), and $R^{(1)}(z,z_0)$ is defined in (\ref{Rbase}).
\end{lemma}
The above Lemma demonstrates the asymptotic behavior of the partition function for general two-cut potentials.  Moreover, it is clear from the Riemann-Hilbert analysis described that the asymptotic behavior can be determined to any desired order.  While a representation in terms of quantities derived from Riemann-Hilbert analysis is to some the end of the road, the goal of a representation in terms of other intrinsic quantities is fundamental.  In Subsections \ref{sec:ellipRSP} and \ref{Section R}, we will construct  $P^{\infty}(z,z_0)$ and $R^{(1)}(z,z_0)$ explicitly, and in the remaining Subsections we will throw the full force of the theory of Riemann surfaces at these formulae so as to arrive at a canonical expansion, which has not appeared in the literature before.  

\subsection{The elliptic Riemann surface and the construction of $P^\infty(z,z_0)$}
\label{sec:ellipRSP}
Instead of constructing $P^\infty(z)$ with normalization at $\infty$, we will directly construct the unique matrix function $P^{\infty}(z,z_0)$ which satisfies the RH conditions (\ref{RHPT1})-(\ref{RHPT3}), (\ref{RHPT5}), and the normalization
\begin{equation}
\label{normz-0}
P^{\infty}(z,z_0)=I+\bigO(z-z_0),\qquad\mbox{ as $z\to z_0$.}
\end{equation}
If $z_0= \infty$, the solution $P^\infty(z,z_0)$ is equal to $P^\infty(z)$. 
In order  to construct $P^{\infty}(z,z_0)$, we will make use of the elliptic Riemann surface 
\be
\label{X}
\mathcal{S}\equiv \{(z,y)\in \mathbb{C}^2\cup\{(\infty,\pm\infty)\}:\, f(z,y)=y^2-\prod_{j=1}^{4}(z-a_j)=0\}.
\ee
This surface consists of two sheets glued together along the cuts $[a_1,a_2]$ and $[a_3,a_4]$.
Given a point $Q=(z,y)$ on $\mathcal S$, we write $z=z(Q)$ for its projection on the complex plane.
For each  $z\in\mathbb C\cup\{\infty\}\backslash\{a_1,a_2,a_3,a_4\}$, the pre-image  of the projection map   gives two  points $(z,y)$ and $(z,-y)$, one on each sheet of the Riemann surface, which we will denote below by $z^{(1)}$ and $z^{(2)}$ and which we call conjugate points on the surface. We adopt the convention that $\infty^{(1)}=(\infty,+\infty)$, $\infty^{(2)}=(\infty,-\infty)$.


Let us introduce a basis of canonical cycles $\{\al,\beta\}$ as in Figure~\ref{fig1}.
The  holomorphic 1-form on $\mathcal{S}$ 
is given by $\frac{ dz}{y}$, and we define the  $\alpha$- and $\beta$-periods $\mathcal A$ and $\mathcal B$
of this 1-form by
\begin{equation}
{\cal A}=\oint_{\alpha}\frac{ dz}{y},\;\;\;\;\;\;\;
{\cal B}=\oint_{\beta}\frac{ dz}{y}.
\label{AB}
\end{equation}
Then the normalized  holomorphic 1-form is given by 
\begin{equation}
v=\frac{dz}{\mathcal{A}y},
\label{holo}
\end{equation}
and satisfies the  condition
$\oint_{\alpha} v=1$. 
Writing $
B= \frac{\mathcal B}{\mathcal A}$, we have $\Im B>0$.
The $\theta$-function with characteristics $\delta$, $\epsilon\in \mathbb{R}$ is given by (\ref{theta car}) 
and satisfies the periodicity properties
\begin{equation}
\begin{split}
\label{periodtheta}
&\theta\pq(z+1;B)=\theta\pq(z;B)e^{2\pi i \delta},\\ 
&\theta\pq(z+B ;B)=\theta\pq(z;B)e^{-i\pi( B+2z+2\epsilon)},
\end{split}
\end{equation}
and the heat equation
\begin{equation}
\label{heat}
\dfrac{\partial^2\theta\pq(z;B)}{\partial z^2}=4\pi i \dfrac{\partial \theta\pq(z;B)}{\partial B}.
\end{equation}
The RH solution $P^\infty(z,z_0)$ will be built out of the following object which is related to the Szeg\H{o} kernel 
\cite{fay}:
\begin{equation}
\label{szego}
\widehat S\pq(Q,Q_0)=\dfrac{1}{2}\left(\dfrac{\gamma(Q)}{\gamma(Q_0)}+
\dfrac{\gamma(Q_0)}{\gamma(Q)}
\right)
\dfrac{\theta\pq\left(\int\limits_{Q_0}^{Q}v
;B\right)}
{\theta\left(\int\limits_{Q_0}^{Q} v;B\right)}
\dfrac{\theta(0;B)}{\theta\pq(0;B)},
\end{equation}
where $Q_0,Q\in \mathcal{S}$, and where $\gamma(Q)$, $Q\in \mathcal{S}$, is defined as 
\begin{equation}
\label{gamma}
\gamma(Q)=\left(\dfrac{(z(Q)-a_{2})(z(Q)-a_4) }{(z(Q)-a_{1})(z(Q)-a_3)}\right)^{1/4},
\end{equation}
where $z(Q)$ is the projection of $Q$ to the complex plane. The function $\gamma$ is a multivalued function   on the surface  $\mathcal S$ and we define it  in such a way that 
\[
\gamma(z^{(2)})=-i\gamma(z^{(1)}),\qquad \gamma(\infty^{(1)})=1.
\]
On the cuts $(a_{1},a_{2}) $ and $(a_3,a_4)$ oriented from left to right, we have
\[
\gamma(z^{(2)})_+=\gamma(z^{(1)})_-,\qquad \gamma(z^{(2)})_-=-\gamma(z^{(1)})_+.
\]
\begin{remark}
The Szeg\"o kernel with characteristics $\pq$ is defined as 
\[
S\pq(Q,Q_0)=\widehat S\pq(Q,Q_0)\dfrac{\sqrt{dz(Q)dz(Q_0)}}{z(Q)-Z(Q_0)}
\]
and it is a $(\frac{1}{2},\frac{1}{2})$ form on $\mathcal{S}\times\mathcal{S}$ having only a simple pole on the  diagonal as $Q\to Q_0$ \cite{fay}. It follows that 
$\widehat S\pq(Q,Q_0)$ has only singularities of branching type at the points $(a_j,0)$, $j=1,2,3,4$.
\end{remark}
 We have the following properties:
\begin{itemize}
\item[(i)] As $Q_0\rightarrow Q$, 
\begin{multline}
\label{szegoexp}
\widehat S\pq(Q,Q_0)=1+
(\log \theta\pq(0;B))'\dfrac{v(Q)}{dz(Q)}(z(Q)-z(Q_0)) \\+\bigO\left((z(Q)-z(Q_0))^2\right),
\end{multline}
where $v$ is the  normalized holomorphic 1-form.
In the relation above and in the rest of the paper, we denote
\[
(\log\theta\pq(0;B))'=
\dfrac{\partial }{\partial z} \log\theta\pq(z;B)|_{z=0},
\]
and similarly for higher order derivatives with respect to $z$.
\item[(ii)] As $Q$ goes around $\al$ and $\bt$-cycles one has
\begin{align}
\label{Szego1}
&T_{\al}\widehat S\pq(\,Q, Q_0)=\e^{2\pi i \delta}\widehat S\pq(Q,Q_0), \\
\label{Szego2}
&T_{\beta}\widehat S\pq(Q,\,Q_0)=\e^{-2\pi i \epsilon}
\widehat S\pq(Q,Q_0),
\end{align}
where  $T_{\alpha}$ and $T_{\beta}$ are  the monodromy operators.
\end{itemize}

\medskip

The solution $P^{\infty}(z,z_0)$ of the matrix RH problem (\ref{RHPT1}), (\ref{RHPT2}), (\ref{RHPT3}), (\ref{RHPT5}), normalized to the identity  at $z_0$, can be expressed in terms of $\widehat S\pq(Q,Q_0)$. 
For $j,k=1,2$ and $z\in\mathbb C\setminus[a_1,a_4]$, $z_0\in \mathbb C\cup \{\infty\}$, 
we have  \cite{KK, DKI, korotkin}
\begin{equation}
\label{solRHP}
P_{jk}^{\infty}(z,z_0)=\widehat{S}\pq(z^{(k)},z_0^{(j)}),\qquad\mbox{ with $\epsilon=-N\Omega$,\quad $\delta=0$,}
\end{equation}
where $z^{(1)}\equiv (z,y)$ and $z_0^{(1)}\equiv (z_0,y_0)$ lie on the first sheet of the surface, and $z^{(2)}\equiv (z,-y)$ and $z_0^{(2)}\equiv (z_0,-y_0)$ on the second sheet.
The integrals of the form 
\[
\int\limits_{z_0^{(1)}}^{z^{(2)}}v\quad \mbox{or}\quad  \int\limits_{z_0^{(2)}}^{z^{(1)}}v
\] in (\ref{szego})
should be understood in the following way:
\[
\int\limits_{z_0^{(1)}}^{z^{(2)}}v=\int\limits_{z_0^{(1)}}^{a_{4}}v-\int\limits_{a_{4}}^{z^{(1)}}v,
\qquad
\int\limits_{z_0^{(2)}}^{z^{(1)}} v=-\int\limits_{z_0^{(1)}}^{a_{4}}v+\int\limits_{a_{4}}^{z^{(1)}} v,
\]
i.e.\ the point $a_4$ connects the parts of the integration path on the first and second sheet of the surface.
The explicit formula for $P^\infty(z,z_0)$ thus takes the form
\begin{equation}
\label{Pinfinity2cut}
\begin{split}
&
P^{\infty}(z,z_0)=\dfrac{\theta(0;B)}{2\theta(N\Omega;B)}\times\\
&\hskip -0.5cm\begin{pmatrix}
\left( \dfrac{\gamma(z)}{\gamma(z_0)}+\dfrac{\gamma(z_0)}{\gamma(z)}\right)
\dfrac{\theta\left(\int\limits_{z_0}^{z}v-N\Omega;B\right)}{\theta\left(\int\limits_{z_0}^{z}v;B\right)}
&
\dfrac{\dfrac{\gamma(z)}{\gamma(z_0)}-\dfrac{\gamma(z_0)}{\gamma(z)}  }{i}\dfrac{\theta\left(\int\limits_{z_0}^{a_4}v-\int\limits_{a_4}^{z}v
-N\Omega;B\right)} {\theta\left(\int\limits_{z_0}^{a_4}v-\int\limits_{z_0}^zv;B\right)}\\
\dfrac{\dfrac{\gamma(z_0)}{\gamma(z)}-\dfrac{\gamma(z)}{\gamma(z_0)}}{i}\dfrac{\theta\left(-\int\limits_{z_0}^{a_4}v+\int\limits_{a_4}^{z}v
-N\Omega;B\right)}
{\theta\left(\int\limits_{z_0}^{a_4}v-\int\limits_{z_0}^{z}v;B\right)}
&
\left(\dfrac{\gamma(z)}{\gamma(z_0)}+\dfrac{\gamma(z_0)}{\gamma(z)}\right)\dfrac{\theta\left(\int\limits_{z_0}^{z}v+N\Omega;B\right)}
{\theta\left(\int\limits_{z_0}^{z}v;B\right)}
\end{pmatrix},
\end{split}
\end{equation}
where we now   simply write    $z$ and $z_0$ as complex variables instead of points on the Riemann surface. The complex function $\gamma$ is now defined as
\begin{equation}
\label{gammaz}
\gamma(z)=\left(\dfrac{(z-a_{2})(z-a_4) }{(z-a_{1})(z-a_3)}\right)^{1/4},
\end{equation}
with branch cuts on $[a_1,a_2]\cup[a_3,a_4]$ and such that $\gamma(\infty)=1$.
%
When $z_0= \infty$, this formula for $P^\infty$ reduces to the one derived in \cite{DKMVZ2}{\footnote{ The equivalence of the above formula with the one appearing in \cite{DKMVZ2} is obtained by observing that the quantity $u_+(\infty)+d$ appearing in \cite{DKMVZ2} on page 1367, formula (4.45), is identically zero. }}.
From Fay's trisecant formula \cite{fay}, or alternatively from the RH conditions, it follows that 
\[
\det P^{\infty}(z,z_0)\equiv 1.
\]
Furthermore one verifies easily that
\be
\label{inverse}
P^{\infty}(z,z_0)^{-1}=P^{\infty}(z_0,z),
\ee
and that,
changing the base point, one has the relation
\be
\label{gauge}
P^{\infty}(z,z_1)=P^{\infty}(z_1,z_0)^{-1}P^{\infty}(z,z_0)=P^{\infty}(z_0,z_1)P^{\infty}(z,z_0).
\ee

\subsection{Construction of $R^{(1)}(z,z_0)$ \label{Section R}}
It was shown in \cite{DKMVZ2} that the matrix function $R=R(z)$ appearing in (\ref{as X}) is the solution to a small-norm RH problem similar to the one in Section \ref{section: R}, with jumps on a contour $\Sigma_R$ as shown in Figure \ref{JumpError}. 
\begin{figure}[htb!]
    \includegraphics[width=1.3\textwidth]{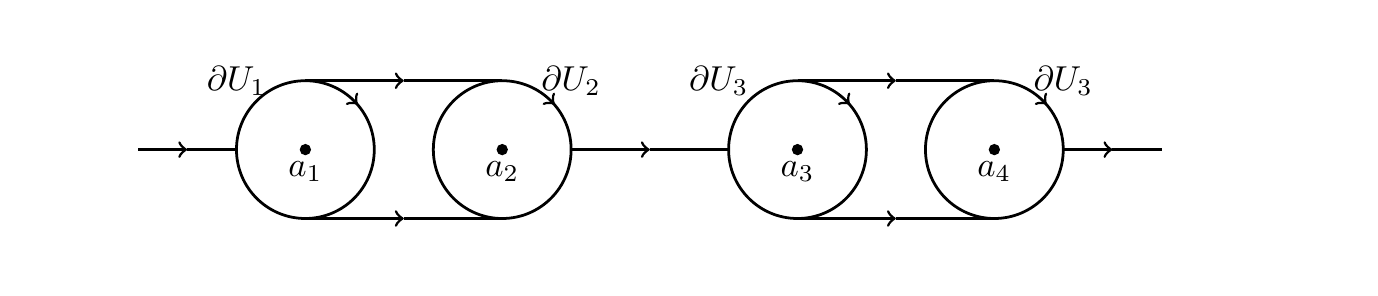}
 \caption{The jump contour $\Sigma_R$ for the error matrix  $R$.}
 \label{JumpError}
\end{figure}

The only jumps for $R$ that are not exponentially close to the identity matrix, are the ones on clockwise oriented circles $\partial U_1, \partial U_2, \partial U_3, \partial U_4$ surrounding the endpoints $a_1,a_2,a_3,a_4$.

\subsubsection*{RH problem for $R$ }
\begin{itemize}
\item[(a)] $R(z)$ is analytic for  $z\in \mathbb C\setminus \Sigma_R$.
\item[(b)] For $z\in \Sigma_R\setminus(\cup_{m=1}^4\partial U_m)$, we have $R_+(z)=R_-(z)(I+\bigO(e^{-cN}))$ as $N\to\infty$ with $c>0$; for $z\in\cup_{m=1}^4\partial U_m$, we have
\begin{equation}
\label{RHR}
R_+(z)=R_-(z)J_R(z),
\end{equation}
where $J_R$ has an asymptotic expansion of the form
\begin{equation}\label{JRexpansion}
J_R(z)=I+\sum_{j=1}^k\frac{\Delta^{(j)}(z)}{N^j}+\bigO(N^{-k-1}),\qquad N\to\infty,
\end{equation}
for any $k\in\mathbb N$,
and where $\Delta^{(j)}(z)$ depends on $N$ but is uniformly bounded as $N\to\infty$.
\item[(c)] As $z\to\infty$, we have
\begin{equation}\label{Rc}
R(z)=I+\bigO(z^{-1}).
\end{equation}
\end{itemize}
Explicit formulas for $\Delta^{(j)}$ can be found. 
For that purpose, we need to define auxiliary objects  that are related to the Airy RH problem as in Section~\ref{Airy}. Let 
\begin{equation}
\label{Ak}
A_j=\frac{1}{2} \left( \frac{2 }{3} \right)^{-j} \pmtwo{(-1)^{j}\left( s_{j} + r_{j} \right)}{i(s_{j}-r_{j})}{-i(-1)^{j}(s_{j}-r_{j})}{s_{j} + r_{j}} ,\qquad j=1,2,\ldots,
\end{equation}
where $s_j$ and $r_j$ are  defined in (\ref{sr}). For $j=1$, we have
\begin{equation}
\label{A1}
A_1=\frac{1}{8}\begin{pmatrix}\frac{1}{6}&i\\i&-\frac{1}{6}\
\end{pmatrix}.
\end{equation}
The  correction terms   $\Delta^{(j)}$,  $j=1,2,\ldots$  for the jump matrix $J_R$ can be expressed in terms of $A_j$ and are given by
\begin{equation}
\label{Delta1}
\Delta^{(j)}(z)=\begin{cases}
(-f_1(z))^{-\frac{3}{2}j}P^{(\infty)}(z)\sigma_3A_j\sigma_3P^{(\infty)}(z)^{-1},&z\in\partial U_1,\\
f_2(z)^{-\frac{3}{2}j}P^{(\infty)}(z)\mathcal{Q}(z)A_j \mathcal{Q}(z)^{-1}P^{(\infty)}(z)^{-1},&z\in\partial U_2,\\
(-f_3(z))^{-\frac{3}{2}j}P^{(\infty)}(z)\mathcal{Q}(z)\sigma_3A_j\sigma_3\mathcal{Q}(z)^{-1}P^{(\infty)}(z)^{-1},&z\in\partial U_3,\\
f_4(z)^{-\frac{3}{2}j}P^{(\infty)}(z)A_jP^{(\infty)}(z)^{-1},&z\in\partial U_4,
\end{cases}\end{equation}

where 
\begin{equation}
\label{Q}
 \mathcal{Q}(z)=e^{\pm \pi iN\Omega\sigma_3}, \qquad \mbox{ for $\pm\Im z>0$},
 \end{equation}
and where  $f_m$, $m=1,2,3,4$,  is a conformal map from a neighborhood $U_m$ of $a_m$ to a neighborhood of $0$, with $f_m'(a_m)>0$.
As $z\to a_m$, we have
\begin{equation}
\label{fm}
((-1)^mf_m(z))^{3/2}=\widehat\psi(a_m)(z-a_m)^{3/2}
+\dfrac{3}{2}\sum_{k=1}^M\widehat\psi^{(k)}(a_m)\dfrac{(z-a_m)^{k+\frac{3}{2}}}{k!(k+\frac{3}{2})}+ \bigO((z-a_m)^{M+\frac{5}{2}}),
\end{equation}
where $\widehat\psi(a_m)$,  $\widehat\psi'(a_m)$  and  $\widehat\psi^{(k)}(a_m)$ are related to the equilibrium density $\psi$ and the function $h$, defined in (\ref{psi}), in the following way:  
\begin{align}
\label{psihat}
&\widehat\psi(a_m)=\lim_{z\to a_m^+}h(z)\prod_{j\neq m}(z-a_j)^{1/2},
\\
&\widehat\psi^{(k)}(a_m)=\lim_{z\to a_m^+}\dfrac{\partial ^k}{\partial z^k}\left( h(z)\prod_{j\neq m}(z-a_j)^{1/2}\right),\quad m=1,2,3,4.
\end{align}
Here $z\to a_m^+$ denotes the limit as $z\to a_m$ from the upper half plane, and $(z-a_j)^{1/2}$ is the principal branch of the square root with branch cut on $(-\infty,a_j]$ and positive for $z>a_j$.
Note that 
\[\widehat\psi(a_2),\widehat\psi(a_4)>0,\qquad \widehat\psi(a_1), \widehat\psi(a_3)\in i\mathbb R^+.\]

The matrix 
$R$ has a large $N$ expansion of the form
\begin{equation}
R(z)=I+\sum_{j=1}^k N^{-j}R^{(j)}(z)+\bigO(N^{-k-1}),\qquad N\to\infty,
\end{equation}
for any $k$, where the coefficients $R^{(j)}$ depend on $N$ but remain 
uniformly bounded for $N$ large.
Combining (\ref{as R 2-cut}) with (\ref{JRexpansion}), we obtain jump relations for the coefficients $R^{(j)}$: we have
 \[
 R^{(k)}_+(z)=R^{(k)}_-(z)+\Delta^{(k)}+\sum_{j=1}^{k-1}R_-^{(j)}(z)\Delta^{(k-j)},\quad z\in\cup_{j=1}^4 \partial U_j.
 \]
 In particular, for $k=1$, this gives 
 \begin{equation}
R^{(1)}_+(z)-R_-^{(1)}(z)=\Delta^{(1)}(z)\qquad z\in\partial U_j,\quad j=1,\ldots, 4,
\end{equation}
where  $\Delta^{(1)}$ is given by (\ref{Delta1}). It is easily seen that $\Delta^{(1)}$ is a meromorphic function in $U_1,\ldots, U_4$ with double poles at $a_1,\ldots, a_4$.
In addition, by (\ref{Rc}), $R^{(1)}(z)$ tends to $0$ as $z\to\infty$.

The analyticity of $R^{(1)}$, its jump relations, and the vanishing at infinity define a scalar RH problem for $R^{(1)}$ with a unique solution
given by \begin{equation}
R^{(1)}(z)=\frac{1}{2\pi i}\sum_{m=1}^4\int_{\partial U_m}\frac{\Delta^{(1)}(\lb)}{\lb-z}d\lb.
\end{equation}
Since $\Delta^{(1)}(z)$ has double poles at $a_1,\ldots, a_4$, this yields
\begin{equation}
\label{R12cut}
R^{(1)}(z)=\sum_{m=1}^4\frac{1}{z-a_m}\res[\lb=a_m]\Delta^{(1)}(\lb)+\sum_{m=1}^4\frac{1}{(z-a_m)^2}\res[\lb=a_m]\left((\lb-a_m)\Delta^{(1)}(\lb)\right),
\end{equation}
for $z$ outside the disks $U_m$.
By (\ref{Rbase}) and (\ref{Pinfbase}), we can write
\begin{equation}
\label{R1b}
R^{(1)}(z,z_0)=\sum_{m=1}^4\frac{1}{z-a_m}\res[\lb=a_m]\Delta^{(1)}(\lb, z_0)+\sum_{m=1}^4\frac{1}{(z-a_m)^2}\res[\lb=a_m]\left((\lb-a_m)\Delta^{(1)}(\lb,z_0)\right),
\end{equation}
where we denote
\begin{eqnarray}
\label{Delta1b}
\Delta^{(j)}(\lambda,z_0)=\begin{cases}
\frac{1}{(-f_1(\lambda))^{\frac{3}{2}j}}P^{(\infty)}(\lambda,z_0)\sigma_3A_j\sigma_3P^{(\infty)}(\lambda,z_0)^{-1},&\lambda\in\partial U_1,\\
&\\
\frac{1}{f_2(\lambda)^{\frac{3}{2}j}}P^{(\infty)}(\lambda,z_0)\mathcal{Q}(\lambda)A_j \mathcal{Q}(\lambda)^{-1}P^{(\infty)}(\lambda,z_0)^{-1},&\lambda\in\partial U_2,\\
&\\
\frac{1}{(-f_3(\lambda))^{\frac{3}{2}j}}P^{(\infty)}(\lambda,z_0)\mathcal{Q}(\lambda)\sigma_3A_j\sigma_3\mathcal{Q}(\lambda)^{-1}P^{(\infty)}(\lambda,z_0)^{-1},&\lambda\in\partial U_3,\\
&\\
\frac{1}{f_4(\lambda)^{\frac{3}{2}j}}P^{(\infty)}(\lambda,z_0)A_jP^{(\infty)}(\lambda,z_0)^{-1},&\lambda\in\partial U_4.
\end{cases}
\end{eqnarray}
Although we will not need this, it is worth noting that a similar procedure applies to the higher order corrections $R^{(j)}(z)$, and that one can recursively find formulas for
\[
R^{(j)}(z,z_0)= P^{\infty}(z_0)^{-1}R^{(j)}(z)P^{\infty}(z_0),\qquad j>1.
\]

\subsection{Evaluation  of (\ref{identity BEH5})}
\label{sec:FirstEvalSec}
In this section we will evaluate the quantities appearing at the right hand side of the differential identity (\ref{identity BEH5}), in terms of other, more canonical quantities, namely the equilibrium measure, and objects from the theory of Riemann surfaces.  The final formula is described in Theorem \ref{theorem86}.
We start with the second term.
\begin{lemma}
\label{lemmaP}
The following identity is satisfied,
\begin{equation}
\label{expfirst}
\begin{split}
\res[z=\infty] \Tr\left(\left.\dfrac{dP^\infty(z,z_0)}{dz}\right|_{z_0=z}\sigma_3 z^k\right)&=-2\left(\log \theta(N\Omega;B)\right)' \res[z=\infty^1](z^k v(z^{1})/dz)\\
&=-2\left(\log \theta(N\Omega;B)\right)' \res[z=\infty^1](z^k \frac{1}{\mathcal A \sqrt{\mathcal{R}(z)}}),
\end{split}
 \end{equation}
 where $\mathcal A$ is the $\alpha$-period given in (\ref{AB}) and where $\mathcal{R}(z)$ is defined by (\ref{R}) and 
 \[
 \left(\log \theta(N\Omega;B)\right)':=\left.\dfrac{d}{dz}\log \theta(z+N\Omega;B)\right|_{z=0}.
 \]
 \end{lemma}
 \begin{proof}
To prove this formula, we only need to consider the diagonal entries of $P^{\infty}(z,z_0)$. As 
 $z_0\rightarrow z$, by (\ref{szegoexp}),
\[
P^{\infty}_{jj}(z,z_0)=1-(z-z_0)\left(\log \theta(N\Omega;B)\right)'\dfrac{v(z^{(j)})}{d z}+\bigO((z-z_0)^2,\;\;j=1,2,
\]
where $v(z^{(2)})=-v(z^{(1)})$ is the holomorphic differential defined in (\ref{holo}).

 It follows that 
 \[
 \mbox{Tr}\left(\left.\dfrac{dP^\infty(z,z_0)}{dz}\right|_{z_0=z}\sigma_3 z^k\right)=-2\left(\log \theta(N\Omega;B)\right)' z^k\frac{v(z^{(1)})}{dz},
 \]
 and this immediately implies (\ref{expfirst}).
\end{proof}
The task of evaluating the term 
\[
\mbox{Tr} \left(\left.\dfrac{d R^{(1)}(z,z_0)}{dz}\right|_{z_0=z}\sigma_3 z^k\right)
\]
appearing in (\ref{identity BEH5}) is considerably more involved.  Th goal is to establish an expression in terms of canonical quantities from the theory of Riemann surfaces, so that comparisons to formulae appearing elsewhere in the literature can be validated (or invalidated).  Our work relies on two so-called Fay identities, (\ref{szegoBergmann}) and (\ref{fay2}).  In order to explain these identities, we need to introduce some classical quantities and  identities in the theory of Riemann surfaces.

The fundamental  symmetric bi-differential,  introduced in \cite{fay} and also called Bergman kernel in some literature, of our Riemann surface,  has an explicit expression which is, 
for \[Q,Q_0\in\mathcal{S}\setminus\{\infty^1,\infty^2,(a_1,0),(a_2,0),(a_3,0),(a_4,0)\},\] given by 
\begin{equation}
\label{bergmann}
w(Q,Q_0)=\dfrac{1}{4}\left(\dfrac{\gamma(Q)}{\gamma(Q_0)}+
\dfrac{\gamma(Q_0)}{\gamma(Q)}\right)^2
\dfrac{dz(Q)dz(Q_0)}{(z(Q)-z(Q_0))^2}-\left(\log \theta(\boldsymbol{0};B)\right)''v(Q)v(Q_0),
\end{equation}
where $\gamma$ has been defined in (\ref{gamma}).
The bi-differential $w(Q,Q_0)$ is analytic for all $Q,Q_0$  except for a  double pole on the diagonal 
as $Q\to Q_0$.

We have the well-known relations \cite{fay}
\begin{align}
&\label{intbeta}
\int_{\beta}w(Q,Q_0)=2\pi i v(Q),\\&
\label{intQ1Q2}
\int_{Q_2}^{Q_1}w(Q,Q_0)=\omega_{Q_1Q_2}(Q),
\end{align}
where $\omega_{Q_1Q_2}(Q)$ is the normalized  third kind differential with simple poles at $Q_1$ and $Q_2$ with residues $\pm 1$ respectively.

The  normalized second kind differential on the elliptic curve $\mathcal{S}$ with poles at $(\infty,\pm\infty)$ of order $k+1$  is given by
\begin{equation}
\label{sigmak}
\sigma_k(z,y)=\dfrac{P_k(z)dz}{y},
\end{equation}
where $P_k(z)$ is a monic polynomial of  degree $k+1$  in $z$ which is  determined uniquely by the conditions 
\[
\sigma_k(z,y)=\pm(z^{k-1}+\bigO(z^{-2}))dz,\quad \mbox{as}\quad (z,y)\rightarrow (\infty,\pm\infty),
\]
and
\[
\int\limits_{a_{2}}^{a_{3}}\sigma_k=0.
\]
In order to write $\sigma_k(z,y)$ using the fundamental symmetric bi-differential  it is  sufficient to observe that 
\begin{equation}
\label{sigmaberg}
\dfrac{1}{2}\sigma_k(z,y)+\dfrac{1}{2}z^{k-1}dz=-\dfrac{1}{k}\res[(\lb,\eta)=(\infty,+\infty)]\left(\lb^{k}w((z,y), (\lb,\eta))\right).
\end{equation}

 The following important identities connect the product of Szeg\"o kernels to the fundamental symmetric bi-differential. First we have \cite[Corollary 2.12]{fay} 
\begin{multline}
\label{szegoBergmann}
\widehat S\pq(Q,Q_0)\widehat S\pq(Q_0,Q)\frac{dz(Q_0)dz(Q)}{(z(Q)-z(Q_0))^2}\\=w(Q,Q_0)+\left(\log \theta\pq(\boldsymbol{0};B)\right)''v(Q)v(Q_0),
\end{multline}
and secondly \cite[Proposition 2.10]{fay} 
\begin{multline}
\label{fay2}
-\widehat{S}\pq(Q,P_0)\widehat{S}\pq(Q_0,Q)\dfrac{dz(Q)}{(z(Q)-z(P_0))(z(Q_0)-z(Q))}=\dfrac{\widehat{S}\pq(Q_0,P_0)}{z(Q_0)-z(P_0)}\\
\times \left(\omega_{Q_0P_0}(Q)+\left[\left(\log\theta\pq\left(\int\limits_{P_0}^{Q_0}v;B\right)\right)'-\left(\log\theta\pq(0;B)\right)'\right]v(Q)\right),
\end{multline}
where $\omega_{Q_0,P_0}(P)$ is the third kind normalized differential  with simple poles at the points $P_0$ and $Q_0$ with residues $\pm 1 $.
In the confluent limit $Q_0\rightarrow P_0$ the second identity reduces to the first one.
\begin{remark}
In (\ref{fay2}), we observe that the quantity
\begin{equation}
\label{pidentity}
\omega_{Q_0P_0}(Q)+ \left[\left(\log\theta\pq\left(\int\limits_{P_0}^{Q_0}v;B\right)\right)'-\left(\log\theta\pq(0;B)\right)'\right]v(Q)
\end{equation}
is  a single-valued function of $Q_0$ and $P_0$. Indeed, defining $d_{Q_0}$ as the differentiation with respect to $Q_0$,  by (\ref{intQ1Q2}),
one has 
\[
d_{Q_0}(\omega_{Q_0,P_0}(Q))=d_{Q_0}\int\limits_{P_0}^{Q_0}w(Q,\tilde{Q})=w(Q,Q_0),
\]
so that, by  (\ref{intbeta}), the integral in $Q_0$ of the above identity along the  $\alpha$- and $\beta$-cycles gives
\[
\int_{\beta}d_{Q_0}(\omega_{Q_0P_0}(Q))=\int_{\beta}w(Q,Q_0)=2\pi i v(Q),\quad \int_{\alpha}d_{Q_0}(\omega_{Q_0P_0}(Q))=\int_{\alpha}w(Q,Q_0)=0.
\]
Therefore as $Q_0\mapsto Q_0+\beta$, we have
\[
\omega_{Q_0,P_0}(Q)\mapsto \omega_{Q_0,P_0}(Q)+2\pi i v(Q),
\]
and as $Q_0\mapsto Q_0+\alpha$, we have
\[
\omega_{Q_0,P_0}(Q)\mapsto \omega_{Q_0,P_0}(Q).
\]
For the second term in (\ref{pidentity}), as $Q_0\mapsto Q_0+\alpha$  the  logarithm of the  $\theta$-function does not change while as $Q_0\mapsto Q_0+\beta$, by (\ref{periodtheta}),  
\[
\left(\log\theta\pq\left(\int\limits_{P_0}^{Q_0}v;B\right)\right)'\mapsto \left(\log\theta\pq\left(\int\limits_{P_0}^{Q_0}v;B\right)\right)'-2\pi i. 
\]
It follows that (\ref{pidentity})  is a single-valued function of $Q_0$. The same consideration applies to the point $P_0$.
\end{remark}

Introduce the following quantities which contain information about the fundamental bi-differential kernel near the branch points:
\begin{equation}
\begin{split}
\label{Bevaluation}
&\widehat{w}(a_j,z)=\lim_{\lambda\to a_j}\dfrac{(\lambda-a_j)^{1/2}w(\lambda,z^{(1)})}{d\lambda dz},\\
& \widehat{w}'(a_j,z)=\lim_{\lambda\to a_j}\dfrac{\partial }{\partial \lambda}\left[(\lambda-a_j)^{1/2}\left(\dfrac{w(\lambda,z^{(1)})}{d\lambda dz }-\dfrac{1}{2(z-\lambda)^2}\right)\right].
\end{split}
\end{equation}
Similarly, for the normalized holomorphic differential, we define
\be
\label{hatv}
\widehat{v}(a_j)=\lim_{\lambda\to a_j}(\lambda-a_j)^{1/2}\dfrac{v(\lambda)}{d\lambda}, \qquad \widehat{v}'(a_j)=\lim_{\lambda\to a_j}\dfrac{\partial }{\partial \lambda}\left[(\lambda-a_j)^{1/2}\dfrac{v(\lambda)}{d\lambda }\right].
\end{equation}
Recall also the definition of $\widehat\psi(a_m)$ in (\ref{psihat}).
A last quantity we need to introduce is the
 so-called  projective connection $S_B(a_j)$ given by
\begin{eqnarray}
\label{BPC}
S_B(a_j)
&
=&
\dfrac{3}{2}\res[\lambda=a_j]\left(\dfrac{d}{d\lambda}\log\prod_{i=1}^{2}\dfrac{\lambda-a_{2i}}{\lambda-a_{2i-1}}\right)^2-24\left(\log \theta(\boldsymbol{0};B)\right)''
\res[\lambda=a_j]v(\lambda)^2\\
&=&\label{BPC2}3\sum_{i\neq j}(-1)^{i+j}\frac{1}{a_j-a_i}
-24\left(\log \theta(\boldsymbol{0};B)\right)''\widehat v(a_j)^2.
\end{eqnarray}

We are now ready to evaluate the second term at the right hand side of  (\ref{identity BEH5}).
\begin{theorem}
\label{theo82}
The following identity is satisfied,
 \begin{multline}
 \label{finalexp}
 \dfrac{1}{2}\res[z=\infty]\left(\Tr \left(\left.\dfrac{d R^{(1)}(z,\lambda_0)}{dz}\right|_{\lambda_0=z}\sigma_3\right)z^k\right)\\=
 \dfrac{1}{24}\sum_{j=1}^4\dfrac{\res[z=\infty^1][z^k\widehat{w}'(a_j,z)]}{\widehat{\psi}(a_j))}\\
\quad -\dfrac{1}{2}\sum_{j=1}^4\left(\dfrac{1}{4}\dfrac{\widehat{\psi}'(a_j)}{\widehat{\psi}(a_j)}-\dfrac{1}{6}S_B(a_j)-4\dfrac{\theta''\left( N \Omega;B\right)}{\theta\left( N \Omega;B\right)}\widehat{v}(a_j)^2\right)\dfrac{\res[z=\infty^1][z^k\widehat{w}(a_j,z)]}
{\widehat{\psi}(a_j)}\\
\quad +\dfrac{1}{2}\left(\log\theta\left( N \Omega;B\right)\right)''\res[z=\infty][z^k\frac{v(z^{(1)})}{dz}]
 \sum_{j=1}^4\dfrac{\frac{1}{6}S_B(a_j)\widehat{v}(a_j)+\frac{1}{12}\widehat{v}'(a_j)-\dfrac{1}{4}\dfrac{\widehat{\psi}'(a_j)}{\widehat{\psi}(a_j)}\widehat{v}(a_j)}
{\widehat{\psi}(a_j)}\\
+\dfrac{2}{3}\left(\dfrac{\left(\theta\left( N \Omega;B\right)\right)'''}{\theta\left( N \Omega;B\right)}\right)'\res[z=\infty][z^k\frac{v(z^{(1)})}{dz}]
\sum_{j=1}^4\dfrac{\widehat{v}(a_j)^3}{\widehat{\psi}(a_j)}.
\end{multline}
\end{theorem}
\begin{proof}
Observe by (\ref{R1b}) that 
\begin{multline}
\label{eqR12}
\frac{1}{2}\res[z=\infty]\Tr\left(\left.\dfrac{dR^{(1)}(z,z_0)}{dz}\right|_{z_0=z}z^k\sigma_3\right)\\=
-\frac{1}{2}\res[z=\infty]
\left(\sum_{j=1}^4\frac{z^k}{(z-a_j)^2}\res[\lb=a_j]\left[\left(1+2\dfrac{\lambda-a_j}{z-a_j}\right)\Tr\left(\Delta^{(1)}(\lb, z)\sigma_3\right)\right)\right].
\end{multline}

From the definitions of $\Delta^{(1)}$ and $A_1$ in (\ref{A1})-(\ref{Delta1}),  one has 
\begin{equation}
\label{TrDelta11}
\mbox{Tr}(\Delta^{(1)}(\lb,z)\sigma_3)=\Sigma_1(\lb,z)+\Sigma_2(\lb,z),
\end{equation}
with 
\begin{multline}
\label{Sigma1}
\Sigma_1(\lb,z)=\dfrac{1}{48}\dfrac{1}{((-1)^jf_j(\lb))^{\frac{3}{2}}}\left(P^{\infty}_{11}(\lb,z)P^{\infty}_{11}(z,\lb)-P^{\infty}_{12}(\lb,z)P^{\infty}_{21}(z,\lb)\right.\\
+\left. P^{\infty}_{22}(\lb,z)P^{\infty}_{22}(z,\lb)-P^{\infty}_{21}(\lb,z)P^{\infty}_{12}(z,\lb)\right),
\end{multline} for $z$ near $a_j$,
and for $\Im\lambda\geq 0$, we have
\begin{equation}
\label{Sigma2}
\Sigma_2(\lb,z)=\begin{cases}-\dfrac{i}{4}\dfrac{1}{(-f_1(\lb))^{\frac{3}{2}}}\left(P^{\infty}_{12}(\lb,z)P^{\infty}_{11}(z,\lb)-P^{\infty}_{21}(\lb,z)P^{\infty}_{22}(z,\lb)\right),\\
\dfrac{i}{4}\dfrac{1}{f_2(\lb)^{\frac{3}{2}}}\left(P^{\infty}_{12}(\lb,z)P^{\infty}_{11}(z,\lb)e^{-2\pi iN\Omega}-P^{\infty}_{21}(\lb,z)P^{\infty}_{22}(z,\lb)e^{2\pi iN\Omega}\right),\\
-\dfrac{i}{4}\dfrac{1}{(-f_3(\lb))^{\frac{3}{2}}}\left(P^{\infty}_{12}(\lb,z)P^{\infty}_{11}(z,\lb)e^{-2\pi iN\Omega}-P^{\infty}_{21}(\lb,z)P^{\infty}_{22}(z,\lb)e^{2\pi iN\Omega}\right),\\
\dfrac{i}{4}\dfrac{1}{f_4(\lb)^{\frac{3}{2}}}\left(P^{\infty}_{12}(\lb,z)P^{\infty}_{11}(z,\lb)-P^{\infty}_{21}(\lb,z)P^{\infty}_{22}(z,\lb)\right),
\end{cases}
\end{equation}
where we used the fact that $P^{\infty}(\lb,z)^{-1}=P^{\infty}(z,\lb)$, and where the first line is valid for $\lambda$ near $a_1$, the second near $a_2$, etc.

For each of the matrix products at the right hand side of  (\ref{Sigma1}), the identity (\ref{szegoBergmann}) which connects the Szeg\"o kernel and the fundamental bi-differential. For instance,
\begin{multline*}
P^{\infty}_{kj}(\lb,z)P^{\infty}_{jk}(z,\lb)\\=\left(w(\lambda^{(j)},z^{(k)})+\left(\log\theta(N\Omega)\right)''v(\lambda^{(j)})v(z^{(k)})\right)\dfrac{(\lambda-z)^2}{d\lambda dz}, \;\;j,k=1,2.
\end{multline*}

Using the expression of the fundamental bi-differential in (\ref{bergmann}), by (\ref{Bevaluation}), one obtains expansions of the following form as $\lb\to a_j$:
\begin{equation}
 \label{expansionW}
 \begin{split}
&\dfrac{ w(\lambda^{(1)},z)}{d\lambda dz}=\dfrac{\widehat{w}(a_j,z)}{(\lambda-a_j)^{1/2}}+\dfrac{1}{2(z-a_j)^2}+\widehat{w}'(a_j,z)(\lambda-a_j)^{1/2}+\bigO(\lambda-a_j),\\
&\dfrac{ w(\lambda^{(2)},z)}{d\lambda dz}=-\dfrac{\widehat{w}(a_j,z)}{(\lambda-a_j)^{1/2}}+\dfrac{1}{2(z-a_j)^2}-\widehat{w}'(a_j,z)(\lambda-a_j)^{1/2}+\bigO(\lambda-a_j),
\end{split}
\end{equation}
and 
\begin{equation}
\label{expansionholo}
\dfrac{v(\lambda^{(1)})}{d\lambda}=-\dfrac{v(\lambda^{(2)})}{d\lambda}=\dfrac{\widehat{v}(a_j)}{(\lambda-a_j)^{1/2}}+\widehat{v}'(a_j)(\lambda-a_j)^{1/2}+\bigO(\lambda-a_j).
\end{equation}
Combining  (\ref{fm}),  (\ref{TrDelta11}), (\ref{expansionW}), and (\ref{expansionholo}), we obtain the following relation for the residues of $\Sigma_1(\lb,z)$ defined in (\ref{Sigma1}),
\begin{equation}
\label{expansionDelta1}
\begin{split}
&-\sum_{j=1}^4\frac{1}{(z-a_j)^2}\res[\lb=a_j]\Sigma_1(\lb, z)-2\sum_{j=1}^4\frac{1}{(z-a_j)^3}\res[\lb=a_j]\left((\lb-a_j)\Sigma_1(\lb,z)\right)\\
&\qquad =-\dfrac{1}{12}\dfrac{1}{\widehat{\psi}(a_j)}\left[\widehat{w}'(a_j,z^{(1)})+\left(\log\theta\left(N\Omega\right)\right)''\widehat{v}'(a_j)\frac{v(z^{(1)})}{dz}
\right]\\
&\qquad\qquad\qquad+\dfrac{1}{20}\dfrac{\widehat{\psi}'(a_j)}{\widehat{\psi}(a_j)^2}\left[\widehat{w}(a_j,z^{(1)})+\left(\log\theta\left(N\Omega\right)\right)''\widehat{v}(a_j)\frac{v(z^{(1)})}{dz}\right].
\end{split}
\end{equation}

In order  to evaluate the  residues of products of the form $P^{\infty}_{12}(\lb,z)P^{\infty}_{11}(z,\lb)$  in  $\Sigma_2(\lb,z)$ defined in (\ref{Sigma2}), we use (\ref{fay2}) which gives
 \begin{multline}
 \label{fay2a}
P^{\infty}_{12}(\lb,z)P^{\infty}_{11}(z,\lb)=\lim_{\eta\to\lambda}\dfrac{\widehat{S}[^{\;\;\;\;0}_{-N\Omega}](\eta^{(2)},\lambda^{(1)})}{\eta-\lambda}
\dfrac{(\lambda-z)^2}{dz}
\times\\
\left[v(z^{(1)})\left(\left(\log\theta\left(\int_{\lb^{(1)}}^{\lb^{(2)}}v-N\Omega\right)\right)'-(\log\theta\left(-N\Omega\right))'\right)
+\omega_{\lb^{(2)}\lb^{(1)}}(z^{(1)})\right],\end{multline}
and similarly
 \begin{multline}
 \label{fay2b}
P^{\infty}_{21}(\lb,z)P^{\infty}_{22}(z,\lb)=\lim_{\eta\to\lambda}\dfrac{\widehat{S}[^{\;\;\;\;0}_{-N\Omega}](\eta^{(1)},\lambda^{(2)})}{\eta-\lb}
\dfrac{(\lambda-z)^2}{dz}
\times\\
\left[v(z^{(2)})\left(\left(\log\theta\left(\int^{\lb^{(1)}}_{\lb^{(2)}}v-N\Omega\right)\right)'-(\log\theta\left(-N\Omega\right))'\right)+\omega_{\lb^{(1)}\lb^{(2)}}(z^{(2)})\right],
\end{multline}
where $\omega_{\lb^{(1)}\lb^{(2)}}(z^{(j)})$, $j=1,2$, stands as before for the normalized third kind differential, in this case with poles at two conjugate points $\lb^{(1)}$ and $\lb^{(2)}$.
Next, using (\ref{intQ1Q2}), we have the following identities for a third kind differential with poles at two conjugate points
 $\lambda^{(2)}$ and $\lambda^{(1)}$, 
 
\begin{equation}
\label{thirdkindsym}
\omega_{\lambda^{(1)}\lambda^{(2)}}(P)=-\omega_{\lambda^{(2)}\lambda^{(1)}}(P),\qquad \omega_{\lambda^{(2)}\lambda^{(1)}}(z^{(1)})= \omega_{\lambda^{(1)}\lambda^{(2)}}(z^{(2)}),
\end{equation} and
we have the following expansion as $\lambda\to a_j$,
\begin{equation}
\label{expansionthirdkind}
\omega_{\lambda^{(1)},\lambda^{(2)}}(z^{(1)})=4\widehat{w}(a_j,z^{(1)})(\lambda-a_j)^{1/2}+\dfrac{4}{3}\widehat{w}'(a_j,z^{(1)})(\lambda-a_j)^{3/2}+\bigO((\lambda-a_j)^{5/2}).
\end{equation}

In order to  proceed, we need to evaluate the expansions in $\lb$ of the  quantities  (\ref{fay2a}) and (\ref{fay2b})
as $\lambda\to a_j$, $j=1,2,3,4$.
First we observe that
\begin{equation}
\label{prime2}
\lim_{\eta\to\lambda}\dfrac{\widehat{S}[^{\;\;\;\;0}_{-N\Omega}](\eta^{(1)},\lambda^{(2)})}{\eta-\lb}
=\dfrac{i}{4}\dfrac{\theta(0;B)}{\theta(N\Omega;B)}
\dfrac{\theta\left(\int\limits_{\lambda^{(2)}}^{\lambda^{(1)}}v-N\Omega;B\right)}{\theta\left(\int\limits_{\lambda^{(2)}}^{\lambda^{(1)}}v;B\right)}
\sum_{k=1}^4\dfrac{(-1)^k}{\lambda-a_k},
\end{equation} 
and
\begin{equation}
\label{prime22}
\lim_{\eta\to\lambda}\dfrac{\widehat{S}[^{\;\;\;\;0}_{-N\Omega}](\eta^{(2)},\lambda^{(1)})}{\eta-\lb}
=-\dfrac{i}{4}\dfrac{\theta(0;B)}{\theta(N\Omega;B)}
\dfrac{\theta\left(\int\limits_{\lambda^{(2)}}^{\lambda^{(1)}}v+N\Omega;B\right)}{\theta\left(\int\limits_{\lambda^{(2)}}^{\lambda^{(1)}}v;B\right)}
\sum_{k=1}^4\dfrac{(-1)^k}{\lambda-a_k},
\end{equation} 
as follows from the definition (\ref{szego}) and the symmetry $\theta(z;B)=\theta(-z;B)$.
Therefore using the parity property of $\theta$-function, (\ref{thirdkindsym}),   (\ref{prime2}) and (\ref{prime22}), the   identities (\ref{fay2a}) and (\ref{fay2b}) take the form
 \begin{multline}
 \label{fay2aa}
P^{\infty}_{12}(\lb,z)P^{\infty}_{11}(z,\lb)=-\dfrac{i}{4}\dfrac{\theta(0;B)}{\theta(N\Omega;B)}
\dfrac{\theta\left(\int\limits_{\lambda^{(2)}}^{\lambda^{(1)}}v+N\Omega;B\right)}{\theta\left(\int\limits_{\lambda^{(2)}}^{\lambda^{(1)}}v;B\right)}
\sum_{k=1}^4\dfrac{(-1)^k}{\lambda-a_k}\dfrac{(\lambda-z)^2}{dz}
\times\\
\left[-v(z^{(1)})\left(\left(\log\theta\left(\int_{\lb^{(2)}}^{\lb^{(1)}}v+N\Omega\right)\right)'-(\log\theta\left(N\Omega\right))'\right)
-\omega_{\lb^{(1)}\lb^{(2)}}(z^{(1)})\right],\end{multline}
and similarly
 \begin{multline}
 \label{fay2bb}
P^{\infty}_{21}(\lb,z)P^{\infty}_{22}(z,\lb)=\dfrac{i}{4}\dfrac{\theta(0;B)}{\theta(N\Omega;B)}
\dfrac{\theta\left(\int\limits_{\lambda^{(2)}}^{\lambda^{(1)}}v-N\Omega;B\right)}{\theta\left(\int\limits_{\lambda^{(2)}}^{\lambda^{(1)}}v;B\right)}
\sum_{k=1}^4\dfrac{(-1)^k}{\lambda-a_k}
\dfrac{(\lambda-z)^2}{dz}
\times\\
\left[-v(z^{(1)})\left(\left(\log\theta\left(\int^{\lb^{(1)}}_{\lb^{(2)}}v-N\Omega\right)\right)'-(\log\theta\left(-N\Omega\right))'\right)-\omega_{\lb^{(1)}\lb^{(2)}}(z^{(1)})\right].
\end{multline}

For the integrals between $\lambda^{(2)}$ and $\lambda^{(1)}$, we follow a path which goes through $a_4$. Consequently,  for $\Im\lambda\geq 0$, one has
\begin{equation}
\label{extheta}
\int\limits_{\lambda^{(2)}}^{\lambda^{(1)}}v=\begin{cases}
2\int\limits_{a_4}^{\lambda^{(1)}}v,\\
2\int\limits_{a_3}^{\lambda^{(1)}}v-B,\\
2\int\limits_{a_2}^{\lambda^{(1)}}v-B-1,\\
2\int\limits_{a_1}^{\lambda^{(1)}}v-1,
\end{cases}
\end{equation}
and as $\lambda\to a_j$, $j=1,2,3,4$, we obtain
\begin{equation}
\label{expv}
2\int\limits_{a_j}^{\lambda^{(1)}}v=4\widehat{v}(a_j)(\lb-a_j)^{1/2}+\dfrac{4}{3}\widehat{v}'(a_j)(\lb-a_j)^{\frac{3}{2}}+\bigO((\lb-a_j)^{\frac{5}{2}}),
\end{equation}
with $\widehat{v}(a_j)$  defined  in (\ref{hatv}).
Using  (\ref{periodtheta}),  (\ref{prime2}), (\ref{prime22}), (\ref{extheta}),  (\ref{expv}), and the fact that $\theta(z;B)$ is an even function of $z$, we obtain,  as $\lb\to a_j$ from the upper half plane,
\begin{multline}
\dfrac{i}{4}\dfrac{\theta(0;B)}{\theta(N\Omega;B)}
\dfrac{\theta\left(\int\limits_{\lambda^{(2)}}^{\lambda^{(1)}}v\pm N\Omega;B\right)}{\theta\left(\int\limits_{\lambda^{(2)}}^{\lambda^{(1)}}v;B\right)}
\sum_{k=1}^4\dfrac{(-1)^k}{\lambda-a_k}
\\
=\dfrac{i(-1)^j}{4}\left(\dfrac{1}{\lambda-a_j}\pm 4\widehat{v}(a_j)\frac{\theta'}{\theta}(N\Omega;B)(\lambda-a_j)^{-1/2}+\sum_{k\neq j}\dfrac{(-1)^{k+j}}{a_j-a_k}\right.\\
\left.+8\widehat{v}(a_j)^2\left(\frac{\theta''}{\theta}(N\Omega;B)-\dfrac{\theta''}{\theta}(0;B)\right)
+o(1)\right)
\times
\begin{cases}
1\quad &j=1,4,\\
e^{\pm 2\pi i N\Omega}\quad &j=2,3.
\end{cases}
\end{multline}
Using the definition of the projective connection given in (\ref{BPC}), we can rewrite the above expansion as
\begin{multline}
\label{expansion14}
\dfrac{i}{4}\dfrac{\theta(0;B)}{\theta(N\Omega;B)}
\dfrac{\theta\left(\int\limits_{\lambda^{(2)}}^{\lambda^{(1)}}v\pm N\Omega;B\right)}{\theta\left(\int\limits_{\lambda^{(2)}}^{\lambda^{(1)}}v;B\right)}
\sum_{k=1}^4\dfrac{(-1)^k}{\lambda-a_k}
\\
=\dfrac{i(-1)^j}{4}\left(\dfrac{1}{\lambda-a_j}\pm 4\widehat{v}(a_j)\frac{\theta'}{\theta}(N\Omega;B)(\lambda-a_j)^{-1/2}+\frac{1}{3}S_B(a_j)\right.\\
\left.+8\widehat{v}(a_j)^2\frac{\theta''}{\theta}(N\Omega;B)\right)
\times
\begin{cases}
1\quad &j=1,4,\\
e^{\pm 2\pi i N\Omega}\quad &j=2,3.
\end{cases}
\end{multline}


Combining (\ref{fm}),  (\ref{Sigma2}),  (\ref{expansionthirdkind}),  (\ref{expv}), (\ref{fay2aa}),  (\ref{fay2bb}), and   (\ref{expansion14}), one arrives with the help of Maple to the relation 
\begin{multline}
\label{expansionDelta2}
-\sum_{j=1}^4\frac{\res[\lb=a_j]\Sigma_2(\lb, z)dz}{(z-a_j)^2}-2\sum_{j=1}^4\frac{\res[\lb=a_j]((\lb-a_j)\Sigma_2(\lb,z)dz)}{(z-a_j)^3}\\
\quad=\sum_{j=1}^4\dfrac{\widehat{w}'(a_j,z^{(1)})+\left(\log\theta\left(N\Omega\right)\right)''\widehat{v}'(a_j)v(z^{(1)})}{6\widehat{\psi}(a_j)}\\
\quad +\sum_{j=1}^4\dfrac{\left(\dfrac{1}{6}S_B(a_j)+4\dfrac{\theta''\left(N\Omega\right)}{\theta\left(N\Omega\right)}\widehat{v}(a_j)^2-\dfrac{3}{10}\dfrac{\widehat{\psi}'(a_j)}{\widehat{\psi}(a_j)}\right)}
{\widehat{\psi}(a_j)}
\left(\widehat{w}(a_j,z^{(1)})+\left(\log\theta\left(N\Omega\right)\right)''\widehat{v}(a_j)v(z^{(1)})\right)\\
+\dfrac{4}{3}\sum_{j=1}^4\dfrac{\widehat{v}(a_j)^3v(z^{(1)})}{\widehat{\psi}(a_j)}\left(\left(\log\theta\left(N\Omega\right)\right)''''+3\left(\log\theta\left(N\Omega\right)\right)'''\left(\log\theta\left(N\Omega\right)\right)'\right).
\end{multline}
Summing (\ref{expansionDelta1}) and (\ref{expansionDelta2}) and multiplying by $1/2$ we obtain (\ref{finalexp}) and the theorem is proven.
\end{proof}
Combining   the heat equation (\ref{heat}), Lemma~\ref{lemmaZ}, Lemma~\ref{lemmaP}, and Theorem~\ref{theo82}, we get to the following  identity.
\begin{theorem}
\label{theorem86}
For a two-cut regular potential $V_{\vec t}$,
the logarithmic derivative of the partition function with respect to the parameters $t_k$, has the following asymptotic expansion as $N\to \infty$,
\begin{multline}
\label{identity BEH6}
\frac{\partial}{\partial t_k}\log Z_N(V_{\vec t})=-N^2\frac{\partial}{\partial t_k}F_0-N(\log(\theta(N\Omega;B))'\res[z=\infty][z^kv(z^{(1)})] \\
\quad +\dfrac{1}{24}\sum_{j=1}^4\dfrac{\res[z=\infty^1][z^k\widehat{w}'(a_j,z)]}{\widehat{\psi}(a_j)}\\
\quad -\dfrac{1}{2}\sum_{j=1}^4\left(\dfrac{1}{4}\dfrac{\widehat{\psi}'(a_j)}{\widehat{\psi}(a_j)}-\dfrac{1}{6}S_B(a_j)-16\pi i\dfrac{\dfrac{\partial}{\partial B}\theta\left( N \Omega;B\right)}{\theta\left( N \Omega;B\right)}\widehat{v}(a_j)^2\right)\dfrac{\res[z=\infty^1][z^k\widehat{w}(a_j,z)]}
{\widehat{\psi}(a_j)}\\
+\quad\left(\dfrac{1}{6}\left(\dfrac{\left(\theta\left( N \Omega;B\right)\right)'''}{\theta\left( N \Omega;B\right)}\right)'F_0^{(3)}-\left(\log\theta\left( N \Omega;B\right)\right)'' F_1^{(1)}\right)\res[z=\infty][z^kv(z^{(1)})]+\bigO(N^{-1}),
\end{multline}
 where
 \begin{equation}
\label{F11}
F_1^{(1)}:=\dfrac{1}{2} \sum_{j=1}^4\dfrac{\widehat{v}(a_j)}{\widehat{\psi}(a_j)}\left( \dfrac{1}{4}\dfrac{\widehat{\psi}'(a_j)}{\widehat{\psi}(a_j)}-  \frac{S_B(a_j)}{6}-\frac{\widehat{v}'(a_j)}{12\widehat{v}(a_j)}\right),
\end{equation}
and
\begin{equation}
\label{F03}
F_0^{(3)}:=4\sum_{j=1}^4\dfrac{\widehat{v}(a_j)^3}{\widehat{\psi}(a_j)},
\end{equation}
and all the other quantities are defined in the statement of Theorem~\ref{theo82}.
\end{theorem}
We remark that  the definition of $F_1^{(1)}$  and $F_0^{(3)} $  given by Eynard in \cite{Eynard}  (see also \cite{EO}) coincides with our definition up to  multiplicative factors.

\subsection{Calculating antiderivatives}
\label{sec:LastEvalSec}
In this section, we will express the terms in (\ref{identity BEH6})  as derivatives.
We start with the projective connection $S_B(a_j)$.
\begin{lemma}\label{korotkin}
\cite{korotkin, KK1} 
\label{Delta}
Let $\Delta(\boldsymbol{a})=\prod_{1\leq j<k\leq 4}(a_k-a_j)$ be the Vandermonde determinant of the branch points of the surface $\mathcal{S}$ 
 and $\mathcal{A}$ the $\alpha$- period of the non-normalized holomorphic differential as defined in (\ref{AB}). Then the following relation is satisfied,
\begin{equation}
\label{delta3}
\dfrac{\partial }{\partial a_j}\log\left(\mathcal{A}^{12}\Delta(\boldsymbol{a})^3
\right)=-S_B(a_j),
\qquad j=1,2,3,4,
\end{equation}
where $S_B(a_j)$  at the branch point has been defined (\ref{BPC}).
\end{lemma}
This identity was derived in \cite{korotkin} for more general surfaces than our elliptic surface, namely for  generic  non-singular hyper-elliptic Riemann surfaces. 
Next we are going to evaluate the derivatives with respect to the times $t_k$ of all the relevant quantities appearing  in (\ref{expfirst}) and  (\ref{finalexp}). Some of the identities below have already appeared in the literature, see for example \cite{G,G1,KM,MOR}, but for clarity we are going to prove each of them with our notation.
\begin{lemma}
\label{lemma84}
The following identities are  satisfied:
\begin{align}
\label{diffOmega}
&\dfrac{\partial}{\partial t_k}\Omega=-\res[\lb=\infty^{1}][\lb^kv(\lambda)],\\
\label{diffa}
&\dfrac{\partial a_m}{\partial t_k}=-\dfrac{k\widehat{\sigma}_k(a_m)}{\widehat{\psi}(a_m)}=2\dfrac{\res[\lambda=\infty^1][\lambda^k\widehat{w}(\lambda, a_m)]}{\widehat{\psi}(a_m)},\\
\label{diffB}
&\dfrac{ \partial B}{\partial t_k}=\sum_{m=1}^{4}\dfrac{ \partial B}{\partial a_m}\dfrac{\partial a_m}{\partial t_k}=8\pi i\sum_{m}\widehat v(a_m)^2\dfrac{\res[\lambda=\infty^1][\lambda^k \widehat w(\lambda^{(1)}, a_m)]}{\widehat{\psi}(a_m)},\\
\label{diffpsi}
&\dfrac{ \partial \widehat{\psi}(a_m)}{\partial t_k} =3\dfrac{\widehat{\psi}'(a_m)}{\widehat{\psi}(a_m)}\res[\lambda=\infty^1][\lambda^k\widehat w(\lambda, a_m)]-\res[\lambda=\infty^1][\lambda^k\widehat{w}'(a_m,\lb)],
\end{align}
where
\[
\widehat\sigma_k(a_m):=\dfrac{\sqrt{\lambda-a_m}\sigma_k(\lambda)}{d\lambda}|_{\lambda=a_m} ,
\]
with $\sigma_k$ the meromorphic differential  defined in (\ref{sigmak}), 
$\widehat{w}(a_j,z)$ and $\widehat{w}'(a_j,z)$ defined in (\ref{Bevaluation}).
 \end{lemma}
%
%
\begin{proof}
Recall from (\ref{psi}) that  the density $\psi(\lambda) $ of the  equilibrium measure   takes the form  
\be
\label{psib}
\psi(\lb)=\dfrac{ h(\lb)}{   \pi i  }\mathcal{R}(\lb)_+^{1/2} ,\qquad x\in[a_1,a_2]\cup[a_3,a_4],
\ee
for a polynomial $h(\lambda)$.
The quantity 
\begin{equation}
\label{oneform}
h(\lb)\mathcal{R}(\lb)^{1/2}d\lb
\end{equation}
can be seen as a one-form on the Riemann surface, which 
\begin{itemize}
\item[(1)] has simple zeros at the branch points $(a_i,0)$ $i=1\dots, 4$,
\item[(2)] has a pole at $\infty^{1,2}=(\infty,\pm\infty)$ with asymptotic behaviour 
\be
\label{expansionh}
h(\lb)\mathcal{R}(\lb)^{1/2} d\lb=\pm\left(\dfrac{1}{2}V_{\vec t}'(\lb)-\dfrac{1}{\lb}-\bigO(\lb^{-2})\right)d\lb,
\ee
because of the moments conditions (\ref{moment}),
\item[(3)] is normalized to zero on the $\alpha$-cycle because of (\ref{loop}).
\end{itemize}
By (\ref{V0tdef}), it follows that (\ref{oneform}) can be written also in the form
\be
\label{sigma}
h(\lb)\sqrt{\mathcal{R}(\lb)} d\lb=\dfrac{1}{2}\sum_{j=1}^{2d}jt_j\sigma_j+2\sigma_4-4\sigma_2-\sigma_0,
\ee
where $\sigma_j=\sigma_j(\lb,y)$ are the second kind differentials defined in (\ref{sigmak}) ,
and $\sigma_0$ is a  differential of the third  kind uniquely determined by the conditions of having 
simple poles at $\infty^{1,2}$ with residue $\pm 1$ and of being normalized to zero on the $\alpha$-cycles. It should be noted that the right hand side of (\ref{sigma}) depends on the endpoints $a_1,\ldots, a_4$ through the normalization.

The one-form on the left hand side of  (\ref{sigma}) is built in such a way that  it satisfies conditions $(1)$, but conditions $(2)$ and $(3)$ are imposed by (\ref{moment}) and (\ref{loop}).
  The one-form on the right hand side of  (\ref{sigma}) is built in such a way that  it satisfies conditions $(2)$ and $(3)$, while the condition $(1)$ has to be imposed:
\be
\label{zero}
\left(\sum_{j=1}^{2d}jt_j\sigma_j(\lb,y)+2\sigma_4-4\sigma_2-2\sigma_0(\lb,y)\right)_{(\lb,y)=(a_p,0)}=0,\quad p=1,\dots,4.
\ee
These  $4$ equations determine the endpoints $a_1<a_2<a_3< a_{4}$ of the support of the equilibrium measure and are equivalent to (\ref{moment}) and (\ref{loop}).
Now we take the derivatives with respect to times $t_k$,
\be\begin{split}
\label{timeder}
&\dfrac{\partial}{\partial t_k}\left(\sum_{j=1}^{2d}jt_j\sigma_j+2\sigma_4-4\sigma_2-2\sigma_0\right)\\
&=k\sigma_k+\sum_p\left(\sum_{j=1}^{2d}jt_j\dfrac{\partial}{\partial a_p}\sigma_j+2\dfrac{\partial}{\partial a_p}\sigma_4-4\dfrac{\partial}{\partial a_p}\sigma_2-2\dfrac{\partial}{\partial a_p}\sigma_0\right)\dfrac{\partial a_p}{\partial t_k}\\
&=k\sigma_k,
\end{split}
\end{equation}
where in the last identity we use the fact that \[\sum_{j=1}^{2d}jt_j\dfrac{\partial}{\partial a_p}\sigma_j+2\dfrac{\partial}{\partial a_p}\sigma_4-4\dfrac{\partial}{\partial a_p}\sigma_2-2\dfrac{\partial}{\partial a_p}\sigma_0\]
is a normalized one-form, regular at $\infty^{1,2}$ because of (\ref{expansionh}), and which can only have a singularity  at the branch point $(a_p,0)$. However, by (\ref{zero}), this differential is regular at $(a_p,0)$, $p=1,2,3,4$. Therefore,
$\sum_{j=1}^{2d}jt_j\dfrac{\partial}{\partial a_p}\sigma_j+2\dfrac{\partial}{\partial a_p}\sigma_4-4\dfrac{\partial}{\partial a_p}\sigma_2-2\dfrac{\partial}{\partial a_p}\sigma_0$
 is a normalized holomorphic one form and by Riemann's theorem it is identically zero. Such an argument was first used in \cite{Krichever}.
Therefore we can conclude that 
\begin{equation}\label{111}
\dfrac{\partial }{\partial t_k}\dfrac{ h(\lb)}{   \pi i  }\mathcal{R}(\lb)^{1/2}d\lb=\dfrac{\partial }{\partial t_k}\Re\frac{1}{\pi i}\left(\dfrac{1}{2}\sum_{j=1}^{2d}jt_j\sigma_j+2\sigma_4-4\sigma_2-\sigma_0\right)=\Re(\dfrac{k\sigma_k}{2\pi i }).
\end{equation}
Hence, 
\[\dfrac{\partial}{\partial t_k}\Omega =\dfrac{1}{2\pi i}\dfrac{\partial}{\partial t_k} \int_{\beta}h(\lambda)\mathcal R(\lambda)^{1/2}d\lambda=\dfrac{k}{ 4 \pi i }\int_{\beta}\sigma_k.
\]
By Riemann's bilinear relations  \cite{fay}, we have
\[
\int_{\beta}\sigma_k=-\dfrac{4\pi i}{k}\res[(\lambda,y)=\infty^1][\lambda^k v(\lambda,y)],
\]
so that 
\[\dfrac{\partial}{\partial t_k}\Omega =-\res[(\lambda,y)=\infty^1][\lambda^k v(\lambda,y)],
\]
which proves (\ref{diffOmega}).

In order to prove (\ref{diffa}) we observe that, for general values of $a_1,\ldots, a_4$, we can write
\[
\dfrac{1}{2}\sum_{j=1}^{2d}jt_j\sigma_j+2\sigma_4-4\sigma_2-\sigma_0=\dfrac{Z(\lambda)}{\mathcal{R}(\lambda)^{1/2}}d\lambda,
\]
where $Z(\lambda)$ is a polynomial of degree $2d+1$.
We can write the polynomial $Z(\lambda)$ in the form
\begin{equation}
Z(\lambda)=h(\lambda)\prod_{j=1}^{4}(\lambda-a_j)+Z_0(\lambda),
\end{equation}
where $Z_0(\lb)$ is a polynomial of degree at most $3$.
By (\ref{sigma}), $Z_0(\lambda)$ is identically zero if $a_1,\ldots, a_4$ are such that (\ref{zero}) holds. 
By (\ref{timeder}), one has
\begin{align*}
\dfrac{\partial}{\partial t_k}\left((\lambda-a_m)^{1/2}\dfrac{Z(\lambda)d\lb}{\mathcal R(\lambda)^{1/2}}\right)&=-\dfrac{1}{2(\lambda-a_m)^{1/2}}\dfrac{Z(\lambda)d\lb}{\mathcal R(\lambda)^{1/2}}\dfrac{\partial a_m}{\partial t_k}+(\lambda-a_m)^{1/2}\dfrac{\partial}{\partial t_k}\dfrac{Z(\lambda)d\lb}{\mathcal R(\lambda)^{1/2}}\\
=&-\dfrac{1}{2(\lambda-a_m)^{1/2}}\dfrac{Z(\lambda)d\lb}{\mathcal R(\lambda)^{1/2}}\dfrac{\partial a_m}{\partial t_k}+\dfrac{k}{2}(\lambda-a_m)^{1/2}\sigma_k(\lb).
\end{align*}
This holds for any value of $\lb$. Keeping in mind that the left hand side has 
a zero at $\lambda=a_m$, we can evaluate at $\lambda=a_m$ and obtain
\[
-h(a_m)\sqrt{\prod_{j\neq m}(a_m-a_j)}\dfrac{\partial a_m}{\partial t_k}=-\left(\dfrac{1}{2\sqrt{\lambda-a_m}}\dfrac{Z(\lambda)d\lb}{\sqrt{\mathcal R(\lambda)}}\dfrac{\partial a_m}{\partial t_k}-\dfrac{k}{2}\sqrt{\lambda-a_m}\sigma_k(\lb)\right)_{\lb=a_m},
\]which yields
\begin{equation}
\dfrac{\partial a_m}{\partial t_k}=-k\dfrac{\widehat{\sigma}_k(a_m)}{\widehat{\psi}(a_m)}.
\end{equation}
This is the first equality in (\ref{diffa}). To obtain the second identity in (\ref{diffa}),  it is sufficient to use (\ref{sigmaberg}) and  to calculate the residue at the points $(a_m,0)$, $m=1,2,3,4.$ 
The  proof of (\ref{diffB}) follows from  (\ref{diffa}) and 
from Rauch's variational formula \cite{rauch}
\begin{equation}
\label{Rauch0}
\dfrac{1}{4\pi i}\dfrac{\partial B }{\partial a_i}=\res[\lambda=a_i] \dfrac{v(P)v(P)}{d\lambda^2}= \widehat{v}(a_i)\widehat{v}(a_i).
\end{equation}
In order to prove (\ref{diffpsi}), we  use  (\ref{111}) and (\ref{psihat}), and this leads to
\be
\begin{split}
\dfrac{\partial }{\partial t_k}\widehat{\psi}(a_m)&=\left.\dfrac{\partial }{\partial t_k} \left(\dfrac{\pi i}{(\lb-a_m)^{1/2}}\psi(\lb)\right|_{\lambda=a_m}\right)=\pi i
\left. \left(\dfrac{\partial }{\partial t_k}\dfrac{\psi(\lb)}{(\lb-a_m)^{1/2}}+\dfrac{\partial }{\partial \lb}\dfrac{\psi(\lb)}{(\lb-a_m)^{1/2}}\dfrac{\partial a_m}{\partial t_k}\right)\right|_{\lambda=a_m}\\
&=\left.\left(\dfrac{k\sigma_k(\lb)}{2(\lb-a_m)^{1/2}d\lambda}+\dfrac{\mathcal R(\lb)^{1/2}}{(\lb-a_m)^{1/2}}\dfrac{\partial a_m }{\partial t_k}\left(h'(\lb)+\dfrac{1}{2}h(\lb)\sum_{j}\dfrac{1}{\lb-a_j}\right)\right)\right|_{\lb=a_m}.
\end{split}
\ee
Using (\ref{sigmak}) and (\ref{diffa}), we obtain
\be\begin{split}
\dfrac{\partial }{\partial t_k}\widehat{\psi}(a_m)
&=-\dfrac{3}{2}k\widehat{\sigma}_k(a_m)\dfrac{h'(a_m)}{h(a_m)}+\dfrac{k}{2}\dfrac{\partial_{\lb}P_k(\lb)|_{\lb=a_m}}{\sqrt{\prod_{j\neq m}(a_m-a_j)}}-k\dfrac{P_k(a_m)}{\sqrt{\prod_{j\neq m}(a_m-a_j)}}\sum_{j\neq m}\dfrac{1}{a_m-a_j}\\
&=\left.-\dfrac{3}{2}k\dfrac{\widehat{\sigma}_k(a_m)}{\widehat{\psi}(a_m)}\dfrac{\partial }{\partial \lb} \left(h(\lb)\sqrt{\prod_{j\neq m}\lb-a_j}\right)\right|_{\lambda=a_m}+\dfrac{k}{2}\left.\left(\dfrac{\partial }{\partial \lb} \dfrac{P_k(\lb)}{\sqrt{\prod_{j\neq m}(\lb-a_j)}}\right)\right|_{\lambda=a_m},
\end{split}
\ee
which, by (\ref{psihat}) and (\ref{Bevaluation}), implies (\ref{diffpsi}).
\end{proof}
\subsection{ Proof of Theorem~\ref{theorem211} and  Theorem~\ref{theorem212}}
\label{sec:FinalProof}
Define
\begin{equation}
\label{F1}
F_1=\dfrac{1}{24}\log\left(\dfrac{1}{2^8}\left(\dfrac{|\mathcal{A}|}{2\pi}\right)^{12}\Delta(a)^3\prod_{j=1}^{4}\widehat{\psi}(a_j)\right),
\end{equation}
where $\Delta(a)=\prod_{1\leq i<j\leq 4}(a_j-a_i)$ is the Vandermonde determinant of the branch points,  $\mathcal{A}$ is the $\alpha$-period of the holomorphic differential  defined in (\ref{AB}),  and $\widehat{\psi}$ has been defined in (\ref{psihat}).
This quantity appeared first in \cite{Akemann}.

Combining Theorem~\ref{theorem86} and Lemma~\ref{lemma84}, we get, as $N\to\infty$,
\begin{multline*}
-\frac{\partial}{\partial t_k}\log Z_N(V_{\vec t})=N^2\dfrac{\partial F_0}{\partial t_k} -N(\log\theta\left(N\Omega;B\right))'\dfrac{\partial \Omega}{\partial t_k}+\dfrac{\partial F_1}{\partial t_k}-
\dfrac{\partial}{\partial B}\log\theta\left(N\Omega;B\right)\dfrac{\partial B}{\partial t_k} \\
+\left(\dfrac{F_0^{(3)}}{6}\left(\dfrac{\left(\theta\left( N \Omega;B\right)\right)'''}{\theta\left( N \Omega;B\right)}\right)'-(\log\theta\left(N\Omega;B\right))''F_1^{(1)}\right)\dfrac{\partial \Omega}{\partial t_k}+\bigO(N^{-1}).
\end{multline*}
Consequently, we have
\begin{multline}\label{Mainidentity}
-\frac{\partial}{\partial t_k}\log Z_N(V_{\vec t})=\dfrac{\partial}{\partial t_k}\left(N^2 F_0+F_1-\log\theta\left(N\Omega;B\right)-
\dfrac{\theta'\left(N\Omega;B\right)}{\theta\left(N\Omega;B\right)}\dfrac{F_1^{(1)}}{N}\right.\\
\left.+\dfrac{\theta'''\left(N\Omega\right)}{\theta\left(N\Omega;B\right)}\dfrac{F_0^{(3)}}{6N}+\bigO(N^{-1})\right),\qquad N\to\infty,
\end{multline}
where  $F_0$ and $F_1$ have been defined in (\ref{F0V}) and  (\ref{F1}),  and $F_0^{(3)}$ and $F_1^{(1)}$ have been defined in (\ref{F03}) and (\ref{F11}).
From the above identity the proof of Theorem~\ref{theorem211} follows easily.
\begin{remark}
The error $\bigO(N^{-1})$ in  (\ref{Mainidentity}) contains non-oscillatory terms of order $1/N$ and smaller order terms. To determine such terms, it is necessary to calculate an extra term in the asymptotic expansion 
of the  error matrix $R(z)$.
\end{remark}

\vskip 0.3cm
\noindent
From Theorem~\ref{theorem: 2},  we can choose a continuous path $\vec t:[0,1]\to\mathbb R^m:s\mapsto \vec t(\tau)=(t_1(\tau),\ldots, t_m(\tau))$ in $\mathbb R^m$ such that
 $V_{\vec t(0)}(z)=z^4-4z^2$, and   $V_{\vec t(1)}(z)=V_{\vec t}(z)$ and for all $\tau \in[0,1]$, $V_{\vec t(\tau)}$ is a two-cut regular potential.
 Therefore one can integrate (\ref{dlogZ}) to obtain
 \begin{multline*}
\log Z_N(V_{\vec t})-\log Z_N(V_{0})=-N^2(F_0(V_{\vec t})-F_0(V_0))-F_1(V_{\vec t})+F_1(V_0)\\
+\log\theta\left(N\Omega(V_{\vec t});B(V_{\vec t})\right)-\log\theta\left(N\Omega(V_{0});B(V_0)\right)+\bigO(1/N),\qquad N\to\infty.
\end{multline*}
By (\ref{ZNas}), this implies Theorem~\ref{theorem212}.

\section*{Acknowledgements} 
TC was supported by the European Research Council under the European Union's Seventh Framework Programme (FP/2007/2013)/ ERC Grant Agreement n.\, 307074, by the Belgian Interuniversity Attraction Pole P07/18, and by F.R.S.-F.N.R.S.

TG  was partially supported by  PRIN  Grant ÒGeometric and analytic theory of Hamiltonian systems in finite and infinite dimensionsÓ of Italian Ministry of Universities and Researches and by the FP7 IRSES grant RIMMP ÒRandom and Integrable Models in Mathematical PhysicsÓ. 

KDT-RM was supported in part by the National Science Foundation under grants DMS-0800979 and DMS-1401268 and by by the FP7 IRSES grant RIMMP ÒRandom and Integrable Models in Mathematical PhysicsÓ.

\end{document}